\PassOptionsToPackage{unicode}{hyperref}
\PassOptionsToPackage{hyphens}{url}
\PassOptionsToPackage{dvipsnames,svgnames,x11names}{xcolor}
\documentclass[12pt]{article}

\usepackage{algorithm}
\usepackage{algpseudocode}

\usepackage{enumitem}
\usepackage{amsmath,amssymb}
\DeclareMathOperator*{\argmin}{argmin}
\usepackage{dsfont}
\usepackage{setspace}
\usepackage{iftex}
\ifPDFTeX
  \usepackage[T1]{fontenc}
  \usepackage[utf8]{inputenc}
  \usepackage{textcomp} 
\else 
  \usepackage{unicode-math}
  \defaultfontfeatures{Scale=MatchLowercase}
  \defaultfontfeatures[\rmfamily]{Ligatures=TeX,Scale=1}
\fi
\usepackage{lmodern}
\ifPDFTeX\else  
\fi
\IfFileExists{upquote.sty}{\usepackage{upquote}}{}
\IfFileExists{microtype.sty}{
  \usepackage[]{microtype}
  \UseMicrotypeSet[protrusion]{basicmath} 
}{}
\makeatletter
\@ifundefined{KOMAClassName}{
  \IfFileExists{parskip.sty}{%
    \usepackage{parskip}
  }{
    \setlength{\parindent}{0pt}
    \setlength{\parskip}{6pt plus 2pt minus 1pt}}
}{
  \KOMAoptions{parskip=half}}
\makeatother
\usepackage{xcolor}
\makeatletter
\ifx\paragraph\undefined\else
  \let\oldparagraph\paragraph
  \renewcommand{\paragraph}{
    \@ifstar
      \xxxParagraphStar
      \xxxParagraphNoStar
  }
  \newcommand{\xxxParagraphStar}[1]{\oldparagraph*{#1}\mbox{}}
  \newcommand{\xxxParagraphNoStar}[1]{\oldparagraph{#1}\mbox{}}
\fi
\ifx\subparagraph\undefined\else
  \let\oldsubparagraph\subparagraph
  \renewcommand{\subparagraph}{
    \@ifstar
      \xxxSubParagraphStar
      \xxxSubParagraphNoStar
  }
  \newcommand{\xxxSubParagraphStar}[1]{\oldsubparagraph*{#1}\mbox{}}
  \newcommand{\xxxSubParagraphNoStar}[1]{\oldsubparagraph{#1}\mbox{}}
\fi
\makeatother

\usepackage{longtable,booktabs,array}
\usepackage{calc} 
\usepackage{etoolbox}
\makeatletter
\patchcmd\longtable{\par}{\if@noskipsec\mbox{}\fi\par}{}{}
\makeatother
\IfFileExists{footnotehyper.sty}{\usepackage{footnotehyper}}{\usepackage{footnote}}
\makesavenoteenv{longtable}
\usepackage{graphicx}
\makeatletter
\def\maxwidth{\ifdim\Gin@nat@width>\linewidth\linewidth\else\Gin@nat@width\fi}
\def\maxheight{\ifdim\Gin@nat@height>\textheight\textheight\else\Gin@nat@height\fi}
\makeatother
\setkeys{Gin}{width=\maxwidth,height=\maxheight,keepaspectratio}
\makeatletter
\def\fps@figure{htbp}
\makeatother

\makeatletter
\@ifpackageloaded{caption}{}{\usepackage{caption}}
\AtBeginDocument{%
\ifdefined\contentsname
  \renewcommand*\contentsname{Table of contents}
\else
  \newcommand\contentsname{Table of contents}
\fi
\ifdefined\listfigurename
  \renewcommand*\listfigurename{List of Figures}
\else
  \newcommand\listfigurename{List of Figures}
\fi
\ifdefined\listtablename
  \renewcommand*\listtablename{List of Tables}
\else
  \newcommand\listtablename{List of Tables}
\fi
\ifdefined\figurename
  \renewcommand*\figurename{Figure}
\else
  \newcommand\figurename{Figure}
\fi
\ifdefined\tablename
  \renewcommand*\tablename{Table}
\else
  \newcommand\tablename{Table}
\fi
}
\@ifpackageloaded{float}{}{\usepackage{float}}
\floatstyle{ruled}
\@ifundefined{c@chapter}{\newfloat{codelisting}{h}{lop}}{\newfloat{codelisting}{h}{lop}[chapter]}
\floatname{codelisting}{Listing}

\makeatother
\makeatletter
\@ifpackageloaded{caption}{}{\usepackage{caption}}
\@ifpackageloaded{subcaption}{}{\usepackage{subcaption}}
\makeatother

\ifLuaTeX
  \usepackage{selnolig}  
\fi
\usepackage[]{natbib}
\usepackage{bookmark}

\IfFileExists{xurl.sty}{\usepackage{xurl}}{} 
\hypersetup{
  pdftitle={Convex Mixed-Integer Programming for Causal Additive Models with Optimization and Statistical Guarantees},
  pdfauthor={Xiaozhu Zhang; Nir Keret; Ali Shojaie; Armeen Taeb },
  pdfkeywords={causal discovery, directed acyclic graphs, discrete optimization, generalized additive models, \texorpdfstring{$\ell_0$}{l0}-penalization, non-linear structural equation models},
  colorlinks=true,
  linkcolor={blue},
  filecolor={Maroon},
  citecolor={Blue},
  urlcolor={Blue},
  pdfcreator={LaTeX via pandoc}}

\newcommand{\anon}{1}

\newtheorem{theorem}{Theorem}
\newtheorem{lemma}[theorem]{Lemma}
\newtheorem{proposition}[theorem]{Proposition}

\newtheorem{example}{Example}

\newtheorem{assumption}{Assumption}
\newtheorem{remark}{Remark}

\newenvironment{proof}{\par\noindent\textit{Proof.}}{\hfill$\square$\par}

\newcommand{\tr}{\mathrm{tr}}
\newcommand{\diag}{\mathrm{diag}}

\usepackage{tikz}
\usetikzlibrary{arrows.meta, decorations.pathmorphing, positioning}
\usepackage{multirow}

\date{}

\begin{document}
\addtocontents{toc}{\protect\setcounter{tocdepth}{-1}}

\def\spacingset#1{\renewcommand{\baselinestretch}%
{#1}\small\normalsize} \spacingset{1}


\if1\anon
{
  \title{\bf Convex Mixed-Integer Programming for\\ Causal Additive Models with Optimization and Statistical Guarantees}
  \author{Xiaozhu Zhang\thanks{
   Corresponding author. Email: \texttt{xzzhang@uw.edu}.}\hspace{.2cm}\\
    Department of Statistics, University of Washington \vspace{.2cm}\\
    Nir Keret \\
    Department of Biostatistics, University of Washington \vspace{.2cm}\\
    Rahul Mazumder \\
    MIT Sloan School of Management \vspace{.2cm} \\
    Ali Shojaie \\
    Department of Biostatistics and Department of Statistics, University of Washington \vspace{.2cm} \\
    Armeen Taeb \\
    Department of Statistics, University of Washington
    }
  \maketitle
} \fi

\if0\anon
{
  \bigskip
  \bigskip
  \bigskip
  \begin{center}
    {\LARGE\bf Title}
\end{center}
  \medskip
} \fi

\begin{abstract}

We study the problem of learning a directed acyclic graph from data generated by an additive, nonlinear structural equation model with Gaussian noise. We represent each nonlinear function through a basis expansion and derive a maximum likelihood estimator with group $\ell_0$ regularization on the graph edges and $\ell_2$ regularization on their strengths, with extensions to accommodate interaction effects. A key computational contribution is a scalable mixed-integer formulation that projects out the continuous model parameters, yields a convex objective in the remaining variables, and enables a tailored cutting-plane algorithm that is guaranteed to return a solution with a prespecified optimality gap. Our formulation can also incorporate background knowledge, including known edge presence or absence and partial ordering constraints. We establish consistency for graph recovery when the number of variables grows with the sample size. By linking optimization guarantees to statistical error bounds, we further derive an early stopping criterion in the cutting-plane algorithm that preserves consistency. Compared with existing approaches that assume equal error variances, restrict attention to linear structural equation models, or rely on heuristics, our method provides both optimization and statistical guarantees. Extensive simulations and two real-data analyses show markedly improved graph recovery. 
\end{abstract}

\noindent%

\spacingset{1.8} 

\section{Introduction}\label{sec-intro}

Understanding causal relationships is a central aim of science. It enables us to predict how a system will behave under external interventions---a crucial step toward both understanding and engineering that system. Directed acyclic graphs (DAGs) offer a convenient and powerful framework for modeling causal relationships.  
Although
prior knowledge and expert intuition can suggest plausible causal models, the underlying causal structure is often unknown and must be inferred from observational or experimental data, a task known as \emph{causal discovery}. The learned DAG then serves as a basis for causal inference. 

Many methods in the causal discovery literature assume that the relationships among the variables are linear \citep{chickering2002optimal,xu2025integer,shimizu2006ALN,chen2019causal,Wang2018HighdimensionalCD,ghoshal2019direct}; see also the survey paper \citet{Glymour2019ReviewOC}. 
Such methods cannot directly capture non-linear effects. This paper focuses on causal discovery when relationships among variables may be non-linear.

Causal discovery in non-linear settings is commonly performed using a combinatorial search for a DAG structure, either via constraint-based or score-based methods.
Constraint-based methods identify conditional independencies from the data. An example is the PC algorithm \citep{spirtes2000causation} which initiates with a complete undirected graph and iteratively removes edges based on conditional independence tests. 
However, conditional independence testing is known to be challenging in non-linear settings \citep{Shah2018TheHO}; indeed, existing approaches require complicated tests \citep{zhang2011kernel,heinze2018invariant,Strobl2017ApproximateKC, chakraborty2022nonparametric}, with low power and poor finite-sample performance. Additionally, some constraint-based methods require a condition known as ``strong-faithfulness'', which is known to be restrictive \citep{uhler2013geometry, sondhi2019reduced}. Score-based methods, including the approach considered in this paper, often use penalized log-likelihood as a score function to seek the optimal graph within the entire space of DAGs.
Unlike some constraint-based methods, many score-based approaches do not require the ``strong-faithfulness'' assumption. 
Other related approaches include permutation-based causal discovery methods, which search over variable orderings and exploit sparsity. For example, \citet{raskutti2018learning} propose the sparsest permutation algorithm, while \citet{solus2021consistency} and \citet{lam2022greedy} develop greedy search strategies over variable orderings. However, these permutation-based approaches generally recover causal structures only up to Markov equivalence, rather than identifying fully oriented DAGs under nonlinear additive models.

A prominent example of a score-based method in non-linear settings assumes a \emph{causal additive model} (CAM) where the non-linear functions and the noise terms are additive \citep{buhlmann2014}. CAMs, which build on generalized additive models in regression \citep{hastie1986generalized}, are quite expressive and can model non-linear relationships that linear models miss. Furthermore, since the functional forms are non-linear and the noise is additive, the underlying DAG structure is identifiable from observational data \citep{peters2014causal,Hoyer2008NonlinearCD}. With this powerful model and under Gaussian noise, \cite{buhlmann2014} adopt a multi-step approach to estimate a DAG: they obtain a topological ordering of the variables by maximizing a likelihood score over the space of permutations, and then perform variable selection with respect to the ordering. They prove that the first step of their algorithm is asymptotically consistent, provided that they can solve the optimization problem exactly. However, they deploy a heuristic greedy algorithm (called \emph{IncEdge}) with no optimality guarantees. 

\cite{gao2020polynomial} prove that the solution produced by the heuristic \emph{IncEdge} algorithm in \cite{buhlmann2014} can be inconsistent, highlighting a discrepancy between their optimization and statistical guarantees. To address this issue, for general non-linear models, \cite{gao2020polynomial} propose a polynomial-time algorithm for estimating a topological ordering with statistical guarantees. Similar to \cite{buhlmann2014}, in a second step, they perform variable selection to obtain a DAG. Their approach relies critically on a condition on the residual variances. This assumption is closely related to assumptions made in prior work on linear models with equal noise variances \citep{peters2013identifiability,ghoshal2019direct,chen2019causal}. As we illustrate next, when this condition is violated, the accuracy of the estimated graph deteriorates.

\begin{example}
Consider the following structural equation model (SEM):
\vspace{-0.1in}
\begin{align}\label{ex:teaser}
\begin{split}
    & X_1 \sim \mathcal{N}(0,0.5^2) ,\quad
    X_2 = X_1^2 - 0.25 +z_2, \quad z_2 \sim \mathcal{N}(0,0.1^2), \quad 
    X_4 \sim \mathcal{N}(0, 0.5^2),\\
    & X_3 = 2X_1^2 - 0.5 +h(X_2)- \mathbb{E}[h(X_2)] +z_3, \quad z_3 \sim \mathcal{N}(0,0.3^2),  \quad X_5 \sim \mathcal{N}(0, 0.5^2),
\end{split}
\end{align}
where $X_1, z_2, z_3, X_4, X_5$ are mutually independent.
We generate $n$ samples according to the SEM using four different functions for $h(x)$: $h(x) = \arctan(x^2)$, $h(x) = \exp(x)$, $h(x) = 0.5x^3$, and $h(x) = |x|$. We apply the procedures in \cite{buhlmann2014} (denoted CAM-IncEdge) and \cite{gao2020polynomial} (denoted NPVAR) across the four choices of $h$, and evaluate the accuracy of their graph estimates, over 100 independent trials. We also apply the mixed-integer programming (MIP) method in \cite{xu2025integer} (denoted MIP-linear), which was designed for optimal causal discovery in linear settings, to highlight the benefits of deploying non-linear models. In Figure~\ref{fig:intro}, we display the results, in terms of the probability of exact recovery and structural Hamming distance (SHD), for varying sample size $n$. We observe that, regardless of the sample size, both CAM-IncEdge and NPVAR cannot estimate the correct graph in all settings. In fact, as the sample size increases, their performance does not seem to improve. Our mixed-integer programming approach (MIP-nonlinear), which provides both optimization and statistical guarantees and accounts for nonlinearities, consistently outperforms the other methods across different choices of $h(x)$ when the sample size $n$ is sufficiently large. In particular, it is the only approach that successfully recovers the correct graph. Furthermore, as we see in Section~\ref{sec:experiment:comparison}, a larger number of variables can result in greater deterioration in the performance of CAM-IncEdge and NPVAR. 
\end{example}

\begin{figure}[!ht]
\centering
\includegraphics[width = 0.75\textwidth]{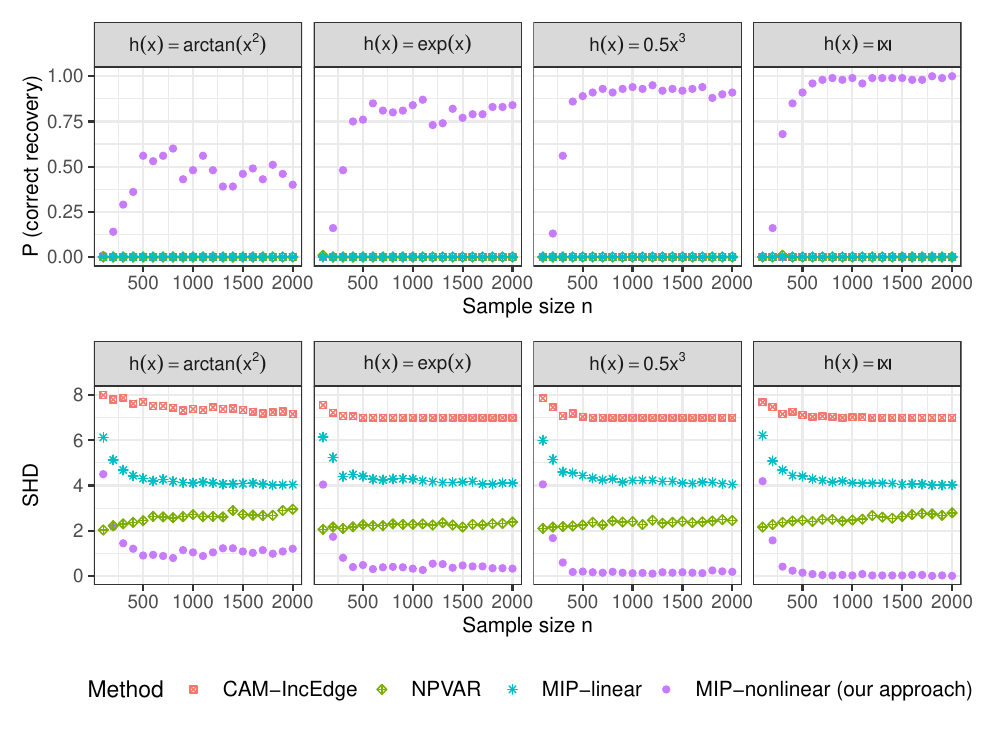}
\caption{\small Graph recovery performance of existing methods and our proposed approach. Top: the probability of recovering the true graph. Bottom: the structural Hamming distance (SHD) between the estimated graph and the true graph, defined as the number of incorrectly recovered or missed edges. The function $h(x)$ corresponds to the one used in model~\eqref{ex:teaser}. Implementation details are provided in~Appendix~\ref{sec:appendix:setup-intro}. }
\label{fig:intro}
\end{figure}

Other approaches have been developed for nonlinear causal discovery. \cite{peters2014causal} propose RESIT, which exploits independence between regression residuals and candidate causes to recover causal structures under additive noise models.
\cite{xu2022inferring} investigate causal discovery under heteroscedastic noise, but focus on the bivariate setting.
\cite{huang2018generalized} introduce generalized score functions for greedy equivalence search (GSGES), accommodating nonlinear relationships through kernel-based regression. More recently, \cite{mian2021discovering} propose GLOBE, which uses the minimum description length principle to discover fully oriented causal graphs, while \cite{xu2025information} develop TOPIC, an information-theoretic framework for causal discovery based on Kolmogorov complexity. As far as we can tell, these prior approaches do not result in certifiably optimal solutions to the optimization problems associated with these estimators. Our approach differs in that we obtain optimal solutions to our MIP-based estimator.

Global optimization approaches that provide optimality certificates have also been developed for DAG learning. \cite{de2011efficient} propose a branch-and-bound algorithm for discrete Bayesian networks, while \cite{lu2021improving} investigate exact A* search in linear settings and introduce Triplet A* as a scalable approximation. \cite{xu2025integer}, as discussed in Example 1, develop a mixed-integer programming approach for linear structural equation models. However, these methods do not directly address nonlinear causal models. \citet{jaakkola2010learning} develop an LP relaxation that provides optimality bounds for DAG learning with general decomposable scores but requires precomputing potentially exponentially many parent-set scores and handling potentially exponentially many constraints.  More importantly, except for \citet{xu2025integer}, these methods do not provide statistical guarantees for suboptimal solutions.

There are score-based methods that do not perform a discrete search over DAGs. Instead, they relax the discrete search space to a continuous search space, allowing gradient descent and
other techniques from continuous optimization to be applied to causal discovery \citep{zheng2018dags,yu2019dag}. 
However, the search space of these problems is highly non-convex, so
that the optimization procedure may become stuck in a local minima \citep{Wei2020DAGsWN}. In fact, there are indications that the performance of these algorithms relies on an artifact of the synthetic evaluation data, and can be matched by a simple algorithm that directly exploits it \citep{reisach2021varsort}. Furthermore, \cite{seng2022notears} found this reliance to be a vulnerability to adversarial attacks.
For additional discussion on non-linear causal discovery, we refer the reader to surveys by \citet{peters2017elements} and \citet{Glymour2019ReviewOC}.

\subsection{Our contributions}
\label{sec:contribution}
We propose a novel convex mixed-integer programming approach for score-based causal discovery in causal additive models with Gaussian errors. Each nonlinear function is represented as a linear combination of a fixed, possibly large, set of basis functions. Our starting point is a maximum-likelihood-based MIP formulation involving continuous variables for the basis coefficients and binary variables encoding the graph structure, with a group $\ell_0$ penalty~\citep{hazimeh2023grouped} promoting graph sparsity and a ridge penalty on the basis coefficients. Using the layered-network formulation of \citet{manzour2021integer}, the DAG constraints and the group $\ell_0$ penalty can be expressed directly in terms of the binary variables.

A key computational contribution is an efficient MIP algorithm, despite the discrete search over graph structures. 
Specifically, we analytically project out the continuous model parameters and show that the resulting objective is convex once the binary variables are relaxed to $[0,1]$. This substantially reduces the number of continuous variables, and enables a tailored cutting-plane algorithm that directly optimizes over DAG structure efficiently. The algorithm is guaranteed to find either a globally optimal solution or a solution within a pre-specified optimality gap. 
In practice, the resulting procedure can often solve moderate-sized graphs (around 30 nodes) to near-optimality (i.e., with a small, pre-specified optimality gap) in less than two minutes. Our proposed strategy yields computational gains even in the linear setting compared with \citet{xu2025integer}, as their method involves $p^2$ additional continuous variables and so-called big-$M$ constants, whose choice can substantially affect the strength of the relaxation and, consequently, the computational speed. Related approaches that combine projection with cutting-plane methods have been studied in~\citet{bertsimas2019sparse,behdin2026multi}, but primarily for quadratic losses. Our loss instead involves the logarithmic term arising from the Gaussian likelihood and the DAG constraint, to our knowledge, is the first application of this strategy to DAG learning.


Our MIP framework is different from previous work for nonlinear causal discovery in several key ways. 
First, it enables us to obtain a solution within a pre-specified optimality gap. Throughout the optimization procedure, the algorithm provides progressively tighter lower and upper bounds on the optimal objective value, allowing us to quantify how far the current solution is from the global optimum, a guarantee generally unavailable for many nonconvex optimization methods.
Second, unlike the method in \cite{gao2020polynomial}, it does not assume equal noise variance. Third, it involves a joint optimization over all model parameters, including a sparse DAG structure via a group $\ell_0$-regularization, as opposed to a two-step procedure that first estimates a topological ordering and then performs variable selection. Finally, unlike other methods, background knowledge---such as the presence or absence of certain edges and partial order constraints among variables---can be organically incorporated as linear constraints in our optimization.

We establish consistency guarantees for our estimator, allowing the number of variables to grow with the sample size. Our analysis combines, and substantially extends, techniques from \cite{buhlmann2014} and \cite{vandegeer2013}. The approach in \cite{buhlmann2014} addresses only the consistency of topological ordering and does not accommodate joint edge selection, whereas the analysis in \cite{vandegeer2013} is restricted to linear models and cannot handle non-linear settings. Our work integrates these techniques to handle both challenges simultaneously. Furthermore, by connecting the optimality and statistical guarantees of our estimator, we present an early stopping criterion for branch-and-bound that preserves both optimality and statistical consistency guarantees.

We demonstrate the utility of our method using both synthetic and two real datasets. 
First, in synthetic experiments with nonlinear additive models, our approach appears to consistently achieve the most accurate DAG estimates among existing methods, maintaining competitive computational speed. This advantage is observed across different graph structures, underlying functions generating the data, and choices of basis functions. Second, our method attains the most accurate recovery of the underlying DAG in a real physical system and exhibits the strongest agreement with the known biological pathway ordering in microarray data.

Finally, we describe how our convex MIP formulation can be extended to accommodate pairwise interaction effects. We focus on hierarchical interactions, where an interaction is present only if both corresponding parent edges are selected. 
Our approach differs from earlier work on linear models ~\citep{hazimeh2020learning,ibrahim2025predicting} by addressing the DAG learning problem, for which we develop a tailored cutting-plane algorithm. By introducing additional binary variables to encode interaction effects, we obtain a convex MIP formulation for this setting. 

The paper is organized as follows. Section~\ref{sec:setup} describes the problem setup; Section~\ref{sec:mip} presents our MIP formulation; Section~\ref{sec:speed-up} outlines computational strategies for accelerating the method; Section~\ref{sec:theory} provides consistency guarantees; Section~\ref{sec:interactions} provides extensions that account for interaction effects; Sections~\ref{sec:experiments}--\ref{sec:real-data} demonstrate the utility of our approach on synthetic and real data; and  Section~\ref{sec:discussion} discusses future work.

\subsection{Notation}\label{sec:notation}
Let $\mathrm{V} := \{1, \dots, p\}$ and $\mathcal{E} := \mathrm{V} \times \mathrm{V} \setminus \{(j,j): j \in \mathrm{V}\}$. A DAG $\mathcal{G} = (\mathrm{V}, E)$ over $p$ nodes consists of a node set $\mathrm{V}$ and a directed edge set $E \subseteq \mathcal{E}$, such that the graph contains no directed cycles. Here, we use $(k,j) \in E$ to indicate that $k \to j$ in the DAG $\mathcal{G}$. The sparsity level of a DAG is defined as its number of edges: $|\mathcal{G}|:=|E|$. The smallest eigenvalue of a matrix $A$ is denoted as $\lambda_{\rm min}(A)$. For any deterministic vector $x \in \mathbb{R}^n$, we define its $n$-norm as $\|x\|_n := \|x\|_2 / \sqrt{n}$. For a random variable $X \in \mathbb{R}$ with $n$ independent realizations $\{x_1, \dots, x_n\}$, we define its empirical $n$-norm as $\|X\|_n := (n^{-1} \sum_{i=1}^n x_i^2 )^{1/2}$.
We use $X_j$, $j = 1, \dots, p$, to denote the random variable corresponding to the $j$-th node. Unless otherwise specified (e.g., in certain proofs), the notation $X_j$, $\epsilon_j$, and their linear combinations refer to real-valued random variables, rather than vectors of their sample realizations.
Lastly, for two sequences $a_n$ and $b_n$, we write $a_n \lesssim b_n$ if there exists a constant $C>0$ and $n_{0} \in \mathbb{N}$ such that $a_{n}\le C b_{n}$ for all $n \geq n_{0}$. We write $a_n \asymp b_n$ if both $a_{n} \lesssim b_{n}$ and $a_{n} \gtrsim b_{n}$ hold. We write $a_n = o(b_n)$ if $a_n /b_n \rightarrow 0$ as $n \rightarrow \infty$. Finally, for any positive integer $k$, let $[k]:=\{1,\ldots,k\}$ and $[k]_0:=\{0,1,\ldots,k\}$.

\section{Problem setup}
\label{sec:setup}
\subsection{Model}
Consider an unknown oracle DAG $\mathcal{G}^\star := (\mathrm{V}, E^\star)$, where $\mathrm{V} = \{1,\dots, p\}$ is the node set and $E^\star \subseteq \mathcal{E}$ is the true set of directed edges. The $p$ nodes of $\mathcal{G}^\star$ correspond to observed random variables $X \in\mathbb{R}^p$.
We assume for each node $j \in \mathrm{V}$, the variable $X_j \in \mathbb{R}$ satisfies a SEM~\citep{lauritzen1996graphical}:
\vspace{-0.1in}
\begin{equation*}
    X_j = f_j(X_{\mathrm{pa}(j)}, \epsilon_j), \quad
    \epsilon_1, \dots, \epsilon_p \text{ mutually independent},
\vspace{-0.1in}
\end{equation*}
where $\mathrm{pa}(j):= \{k:(k,j) \in E^\star \}$ denotes parents for node $j$ in the graph $\mathcal{G}^\star$. The independence of the error terms ensures that there is no latent confounding (causal sufficiency). We primarily focus on the setting where $f_j(\cdot)$ is a non-linear function. Importantly, we assume that $f_j(X_{\mathrm{pa}(j)}, \epsilon_j)$ is additive in each of its arguments, allowing the noise and the causal effects from each parental node to be decoupled from one another. This is known as a causal additive model~\citep{ buhlmann2014, kertel2025boosting, maeda2021causal}, which applies generalized additive models~\citep{hastie1986generalized} to causal settings. An extension that allows for pairwise interactions among causal parents is provided in Section~\ref{sec:interactions}.
We also model each noise variable as Gaussian $\epsilon_j \sim \mathcal{N}(0, \sigma_j^{\star2})$ with non-degenerate variances $\sigma_j^\star > 0$ ($j\in[p]$). This specification enables us to formulate the distribution of $X$, which is crucial for establishing identifiability and a score-based estimator. We arrive at an additive SEM with Gaussian errors:
\vspace{-0.1in}
\begin{equation}\label{eq:DAG-model}
\begin{gathered}
    X_j = \sum_{k\in \mathrm{pa}(j)} f^\star_{kj}(X_k) + \epsilon_j, \\
    \epsilon_j \overset{}{\sim} \mathcal{N}(0, \sigma_j^{\star2} ), \quad
\sigma_j^\star > 0, \quad
\mathbb{E}[f^\star_{kj}(X_k)] = 0, \ \forall j,k \in \mathrm{V},
\vspace{-0.1in}
\end{gathered}
\end{equation}
where $f^\star_{kj}(\cdot)$ is a function from $\mathbb{R} \rightarrow \mathbb{R}$. Note that $f^\star_{kj}(\cdot) \not\equiv 0$ if and only if $k \in \mathrm{pa}(j)$, which is weaker than faithfulness~\citep{peters2014causal}. Without loss of generality, each additive component $f^\star_{kj}(X_k)$ has mean zero; see Appendix~\ref{sec:appendix:zeromean}. We denote by $P$ the joint distribution induced by the SEM in~\eqref{eq:DAG-model}:
\vspace{-0.2in}
$$
(X_1, \dots, X_p) \sim P.
\vspace{-0.2in}
$$
When the functions $f^\star_{kj}$ are linear, the model \eqref{eq:DAG-model} reduces to a linear Gaussian SEM. In that setting, identifiability is a challenge: multiple SEMs may induce the same distribution $P$, making it impossible to distinguish the true SEM based solely on $P$ \citep{vandegeer2013}. However, when the functions $f^\star_{kj}$ are non-linear and smooth enough, $\mathcal{G}^\star$ is the only DAG that is compatible with $P$:
\begin{proposition}[Corollary 30 in \cite{peters2014causal}, \hfill\break Lemma 1 in \cite{buhlmann2014}]
\label{prop:identifiability}
    Suppose all non-zero functions $\{f^\star_{kj}\}_{(k,j)\in \mathcal{E}}$ are non-linear and three times differentiable.
    Then any distribution $Q$ that is generated by model \eqref{eq:DAG-model} with a different DAG $\mathcal{G}' \neq \mathcal{G}^\star$ and non-constant, three-times differentiable functions $\{f'_{kj}\}$ differs from $P$.
\end{proposition}
Proposition \ref{prop:identifiability} establishes that the graph $\mathcal{G}^\star$ is \textit{identifiable} from the distribution $P$ when 
all non-zero functions $f^\star_{kj}$ are non-linear and three times differentiable. We adopt this assumption throughout the paper, as identifiability underpins our goal of recovering the true DAG $\mathcal{G}^\star$.

\subsection{The function class}
We specify the function class containing $\{f^\star_{kj}\}$. For each $(k, j) \in \mathcal{E}$, consider basis functions $\{b_{rkj}(\cdot)\}_{r=1}^\infty$, each three-times differentiable and with infinite support. Let
\vspace{-0.1in}
\begin{equation*}\label{eq:Fkj_functionclass}
\mathcal{F}_{kj} := \left\{
f_{kj}
: f_{kj}(\cdot) = \beta_{0kj} + \sum_{r=1}^\infty \beta_{rkj} b_{rkj}(\cdot) , \quad
\beta_{0kj} = - \mathbb{E}\left[ \sum_{r=1}^\infty \beta_{rkj} b_{rkj}(X_k) \right]
\right\},
\vspace{-0.1in}
\end{equation*}
where the choice of $\beta_{0kj}$ ensures zero mean of $f_{kj}(X_k)$.
We assume throughout that the non-linear true functions belong to these function classes: $f^\star_{kj} \in \mathcal{F}_{kj}$ for all $k$ and $j$. Hence,
\vspace{-0.1in}
$$
X_j = \sum_{k\in\mathrm{pa}(j)} \beta^\star_{0kj} + \sum_{r=1}^\infty \beta^\star_{rkj} b_{rkj}(X_k) + \epsilon_j.
\vspace{-0.1in}
$$
In practice, estimation uses finitely many basis functions. Given $n$ samples, we consider the space 
\vspace{-0.2in}
\begin{equation*}\label{eq:fkj_functionclassRn}
\mathcal{F}_{kj,n} := \left\{ f_{kj,n} \in \mathcal{F}_{kj} : \beta_{rkj} = 0, \quad \forall r \geq R_n +1
\right\}  \subseteq \mathcal{F}_{kj}, 
\vspace{-0.2in}
\end{equation*}
where $R_n$ is a truncation parameter that increases with $n$. In other words, estimation is carried out over a subset of the full function class, using only the first $R_n$ basis functions, and increasing the sample size expands the search space over which estimation is performed. 
In Section~\ref{sec:theory}, in order to facilitate theoretical analysis, we focus on slightly more restricted function classes than $\mathcal{F}_{kj}$ and $\mathcal{F}_{kj,n}$ by imposing additional conditions on the basis functions and coefficients.

\subsection{An intractable group \texorpdfstring{$\ell_0$}{l0}-penalized maximum likelihood estimator}
We assume we have $n$ independent and identically distributed samples of $X$ generated according to model \eqref{eq:DAG-model}. For a given set of basis functions, we carry out estimation over $\mathcal{F}_{kj,n}$ with each function $f_{kj}(\cdot) = \beta_{0kj}+\sum_{r=1}^{R_n} \beta_{rkj} b_{rkj}(\cdot) $. 
Let $\theta :=\left(\{\beta_{rkj}\}_{(k,j) \in \mathcal{E}, \, r\in [R_n]_0}, \{\sigma_j\}_{j \in \mathrm{V}}\right)$ collect the basis coefficients and noise standard deviations. Let $\beta_{kj}=(\beta_{1kj}, \dots, \beta_{R_n kj}) \in \mathbb{R}^{R_n}$ and $\beta_{0*j} = \sum_{k\neq j} \beta_{0kj}$. Accordingly, we define $\mathcal{G}(\theta) := (\mathrm{V},E(\theta))$ as the graph induced by $\theta$, where $(k,j)\in{E}(\theta)$ if and only if $\|\beta_{kj}\| \neq 0$. Then the negative log-likelihood of the model parameterized by $\theta$ is proportional to
\begin{equation}\label{eq:log-likelihood}
\begin{split}
    \ell_n(\theta) 
    &:= \sum_{j=1}^p \log \sigma_j^2 + \sum_{j=1}^p \sigma_j^{-2} \left\| X_j - \sum_{ k\neq j } f_{kj}(X_k) \right\|_n^2 \\
    &= \sum_{j=1}^p \log \sigma_j^2 + \sum_{j=1}^p \sigma_j^{-2} \left\| X_j -  \beta_{0*j} - \sum_{k\neq j }\sum_{r=1}^{R_n} \beta_{rkj} b_{rkj}(X_k) \right\|_n^2 
    .
\end{split}
\end{equation}
The norm $\|\cdot\|_n$ is defined in Section \ref{sec:notation}.
The parameter space is $\Theta=\{\theta: \sigma_j >0,\ \forall k\in \mathrm{V},\ \mathcal{G}(\theta) \text{ is a DAG} \}$. 

We seek a model that balances goodness of fit, measured by negative log-likelihood, with parsimony. To encourage parsimony, we penalize the number of edges in the DAG: 
$
|\mathcal{G}(\theta)| 
= \sum_{(k,j)\in\mathcal{E}} \mathds{1}\left\{ \|\beta_{kj}\|_2 \neq 0 \right\}
$. 
The $\ell_0$-penalty has been widely used in causal discovery \citep{chickering2002optimal, silander2006simple, hauser2012characterization, vandegeer2013, xu2025integer}. When multiple basis functions represent a single causal effect between two nodes, the group $\ell_0$-penalty naturally extends this idea by counting nonzero edges, encouraging the smallest edge set needed to characterize the distribution of $X$. In contrast, the group $\ell_1$-penalty, $\sum_{(k,j)\in \mathcal{E}} \|\beta_{kj}\|_2$, does not directly penalize edge count. Although group $\ell_1$-regularization offers convexity, the constraint set ${\Theta}$ remains nonconvex in our setting. Moreover, as shown in Section~\ref{sec:convex:constraint}, the group $\ell_0$-penalty admits a linear representation in the integer variables used to formulate $\Theta$, allowing it to be incorporated naturally. 


We additionally introduce a ridge penalty on the coefficients, $\sum_{j=1}^p \sum_{k\neq j}\|\beta_{kj}\|_2^2 / \sigma_j^2$, into the negative log-likelihood $\ell_n(\theta)$.
First, the penalty improves numerical stability in the presence of basis expansions and potentially strong correlations among the columns $b_{rkj}(X_k)$, particularly for dense graphs. It also regularizes the fitted functions, preventing overfitting. Second, the penalty improves computational tractability. In Section~\ref{sec:convex:reform}, for a fixed DAG, the coefficients $\{\beta_{rkj}\}_{(k,j)\in\mathcal{E},r\in[R_n]}$ minimizing the penalized negative log-likelihood admit closed-form solutions, enabling a convex MIP formulation.



Our penalized maximum likelihood estimator with the regularization parameters $\lambda_n, \gamma_n \geq 0$ is given by: 
\begin{equation}\label{eq:MLE-groupl0}
\begin{split}
    \widehat{\theta}_n \in
     \argmin_{\theta \in \Theta} \,  
    & \sum_{j=1}^p \left[ \log \sigma_j^2 + \sigma_j^{-2} \left( \left\| X_j - \beta_{0*j} -\sum_{k\ne j}\sum_{r=1}^{R_n} \beta_{rkj} b_{rkj}(X_k)  \right\|_n^2 
    + \gamma_n \sum_{k \neq j} \|\beta_{kj}\|_2^2  \right) \right] \\
    & \ + \lambda_n^2 \sum_{(k,j)\in\mathcal{E}} \mathds{1}\left\{ \|\beta_{kj}\|_2 \neq 0 \right\}.
\end{split}
\end{equation}
Both the objective function and the constraint set of \eqref{eq:MLE-groupl0} are non-convex, rendering this optimization problem generally intractable. In the next section, we reformulate our problem using a convex MIP framework to enable efficient computation.

\section{Convex Mixed Integer Programming Formulation}
\label{sec:mip}

We present some integer programming formulations for our DAG learning problem. Section \ref{sec:convex:constraint} presents formulations based on first principles. Section~\ref{sec:convex:reform} reformulates these programs as computationally appealing convex MIPs, whose continuous relaxations are convex. 


\subsection{A mixed-integer program}
\label{sec:convex:constraint}


We reformulate the objective function of \eqref{eq:MLE-groupl0} as follows. 
We concatenate the basis functions into an extended basis random vector, and similarly concatenate the corresponding coefficients:
\vspace{-0.1in}
\begin{align*}
\begin{split}
    Z_j&:= \left[
        1, \  b_{11j}(X_1), \ldots, b_{1pj}(X_p), \  \ldots, \  b_{R_n1j}(X_1), \ldots, b_{R_npj}(X_p) \right]^\top \in \mathbb{R}^{pR_n +1},  \\
    \beta_j &:= \left[ \beta_{0*j}, \ \beta_{11j}, \ldots, \beta_{1pj},\ \ldots, \ \beta_{R_n 1j}, \ldots, \beta_{R_n pj} \right]^\top \in \mathbb{R}^{pR_n + 1}.
\end{split}
\end{align*}
We also denote by $\beta_j^{-0} \in \mathbb{R}^{pR_n}$ the vector $\beta_j$ excluding the intercept coefficients $\beta_{0*j}$.
By acyclicity, $\beta_{rjj}=0$ for all $j\in\mathrm{V}$ and $r\in[R_n]$. 
Consequently, the sum of all causal effects pointing to node $X_j$ can be simplified as $Z_j\beta_j$. 
As a result, the objective function of \eqref{eq:MLE-groupl0} can be reformulated as:
\begin{align}\label{eqn:sigmaprofile}
\begin{split}
    & \min_{ \{\beta_{j}, \ \sigma_j >0\} }\   \sum_{j=1}^p \log \sigma_j^2 + \sum_{j=1}^p \sigma_j^{-2} \left( \left\| X_j - Z_j \beta_j  \right\|_n^2 + \gamma_n \left\| \beta_j^{-0} \right\|_2^2  \right) + \lambda_n^2 \sum_{(k,j)\in \mathcal{E}} \mathds{1}\left\{ \|\beta_{kj}\|_2 \neq 0 \right\}  \\
    ={}& \min_{ \{\beta_{j}\} }\  \sum_{j=1}^p \log \left( \left\| X_j - Z_j \beta_j  \right\|_n^2 + \gamma_n \left\| \beta_j^{-0} \right\|_2^2 \right) + p + \lambda_n^2 \sum_{(k,j)\in \mathcal{E}} \mathds{1}\left\{ \|\beta_{kj}\|_2 \neq 0 \right\} .
\end{split}
\end{align}
The equality follows by first minimizing over $\sigma_j$ for fixed $\beta_{j}$, thereby reducing the joint optimization over $\{(\beta_j, \sigma_j)\}$ to an optimization over $\beta_{j}$ alone. 
In Appendix~\ref{sec:appendix:equivalence}, we provide an alternative formulation of the loss function \eqref{eq:log-likelihood} following the strategy of \citet{xu2025integer}. This formulation yields a convex mixed-integer program. However, the resulting program is computationally less attractive than our formulation based on \eqref{eqn:sigmaprofile}, which is given in \eqref{eq:MIP-reduced}, because (i) it introduces $p^2$ additional continuous variables and additional constraints, and (ii) the additional constraints require choosing big-$M$ constants, which can substantially weaken the continuous relaxation and slow the branch-and-bound procedure. 

Next, we handle the acyclicity constraint in $\Theta$, that $\mathcal{G}(\theta)$ forms a DAG, which is a key computational challenge of the optimization \eqref{eq:MLE-groupl0}. Many integer programming approaches have been developed to encode this constraint \citep{grotschel1985acyclic, park2017bayesian, manzour2021integer}. Among these formulations, we adopt the \textit{layered network formulation} due to its demonstrated effectiveness in continuous-variable scenarios \citep{manzour2021integer}. This approach has also been applied to a closely related problem under linear SEMs \citep{xu2025integer}.

We introduce two sets of decision variables: (i) A set of binary variables $\{g_{kj} \in\{0,1\} \}_{(k,j) \in \mathcal{E}}$, where $g_{kj} = 1$ represents the presence of the edge from node $k$ to $j$ and $g_{kj} = 0$ otherwise; (ii) A set of continuous variables $\{ \psi_j \in [1,p] \}_{j \in \mathrm{V}}$, where $\psi_j$ represents the \textit{layer value} of each node. The layer values impose a topological ordering on the nodes, such that any ancestor node must have a strictly higher layer value than its descendants. 
Accordingly, we arrive at the following optimization problem:
\vspace{-0.1in}
\begin{subequations} \label{eq:MIP}
\begin{align} 
    \min_{ \substack{\beta_j \in \mathbb{R}^{p R_n + 1}, \ \forall j\in\mathrm{V} \\ 
    g_{kj} \in \{0,1\},\ \forall {(k,j)\in \mathcal{E}} \\
    \psi \in [1,p]^p } } \quad &  \sum_{j=1}^p \log \left( \left\| X_j - Z_j \beta_j  \right\|_n^2 + \gamma_n \left\| \beta_j^{-0} \right\|_2^2 \right) + \lambda_n^2 \sum_{(k,j)\in \mathcal{E}}g_{kj} \label{eq:MIP:a}  \\
    \mathrm{s.t.} \quad & 1 -  p + p g_{kj} \leq \psi_j - \psi_k, \quad \forall (k,j) \in \mathcal{E},  \label{eq:MIP:d} \\
     & \beta_{rkj} = 0 \text{ if } g_{kj}=0, \quad  \forall (k,j) \in \mathcal{E}, \ r\in \{1, \dots, R_n\}, \label{eq:MIP:b} \\ 
    & \beta_{0*j} = 0 \text{ if } \sum_{k=1}^p g_{kj}=0, \quad \forall j\in \mathrm{V} , \label{eq:MIP:b2}  \\
    & \beta_{rjj} = 0, \quad \forall j\in \mathrm{V},\ r\in \{1, \dots, R_n\}. \label{eq:MIP:e}
    \vspace{-0.1in} 
\end{align}
\end{subequations}
Any feasible solution to the optimization problem \eqref{eq:MIP} necessarily corresponds to a DAG. To see this, suppose there is an edge from node $k$ to node $j$, encoded by $g_{kj}=1$. Then, by constraint~\eqref{eq:MIP:d}, the layer value assigned to node $j$ must be strictly larger than that assigned to node $k$, i.e., $\psi_j>\psi_k$. More generally, along any directed path $k \rightarrow m \rightarrow \cdots \rightarrow j$, the layer values must strictly increase, implying $\psi_j>\psi_k$. However, if there were also a directed path from $j$ back to $k$, thereby forming a directed cycle, the same argument would imply $\psi_k>\psi_j$, yielding a contradiction. Therefore, no directed cycle can exist, and constraint~\eqref{eq:MIP} ensures that every feasible solution represents a DAG.

The constraints~\eqref{eq:MIP:b} and \eqref{eq:MIP:b2} link the binary edge-selection variables $g_{kj}$  to the continuous coefficient variables $\beta_{rkj}$, thereby enforcing group sparsity. If there is no edge between nodes $k$ and $j$, which is encoded by $g_{kj} = 0$, the corresponding coefficients $\beta_{rkj}$ must be zero for all $r \in \{1, \dots, R_n\}$; this is guaranteed by constraint \eqref{eq:MIP:b}. Moreover, if there are no parents for node $j$, which is encoded by $\sum_{k=1}^pg_{kj} =0$, the corresponding intercept coefficient $\sum_{k=1}^p\beta_{0kj}$ must also be zero; this is guaranteed by constraint \eqref{eq:MIP:b2}. Conversely, when $g_{kj} = 1$, encoding the presence of an edge between node $k$ and $j$, the associated coefficients $\beta_{rkj}$ are generally not exactly zero since additional degrees of freedom contribute to a better model fit. The constraint~\eqref{eq:MIP:e} prevents self-loop causal effects. Finally, the penalty term, which is the number of edges in the DAG, is naturally encoded as $\sum_{(k,j)\in \mathcal{E}} g_{kj}$. 

Under this framework, prior structural information can be incorporated into the estimation procedure through optional constraints. For example, if background knowledge or domain expertise indicates that there should be a directed edge from node $a$ to node $b$, we can impose the constraint $g_{ab}=1$. Similarly, if node $c$ is known to be an ancestor of node $d$, we can impose the ordering constraint $\psi_c<\psi_d$.

Note that neither the objective function nor the constraints~\eqref{eq:MIP:b} and \eqref{eq:MIP:b2} are convex in the decision variables. In Section~\ref{sec:computation:cutting-plane}, we show that Problem~\eqref{eq:MIP} can nevertheless be solved as a convex MIP. For tuning $\gamma_n,\lambda_n$, we recommend choosing them to minimize the Bayesian information criterion (BIC) score \citep{schwarz1978estimating}, which for estimates $\widehat{\beta}_j$ and $\widehat{g}$ is $\sum_{j=1}^p \log \|X_j - Z_j \widehat{\beta}_j\|_n^2  + (\sum_{(k,j) \in \mathcal{E}} \widehat{g}_{kj}) R_n \log(n)/n$. 

\subsection{A convex MIP reformulation}
\label{sec:convex:reform}

Define the loss function $F(\beta; g):= \sum_{j=1}^p \log \left( \|X_j - Z_j \beta_j\|_n^2 + \gamma_n \|\beta_j^{-0} \|_2^2 \right)$ where $\beta_j$ is subject to the constraints imposed by $g_j$.
To reduce the number of decision variables in \eqref{eq:MIP}, we minimize $F(\beta;g)$ over $\beta$ for a fixed binary vector $g$. This minimization decomposes into $p$ independent ridge regression problems, one per node. Substituting the corresponding solutions back into $F(\beta;g)$ yields the reduced objective 
\vspace{-0.1in}
$$
f(g):=\min_{\beta}F(\beta;g)~~\text{s.t.}~~\text{ constraints \eqref{eq:MIP:b} and \eqref{eq:MIP:b2}}.
\vspace{-0.1in}
$$ 
In Appendix~\ref{sec:appendix:convexity}, we prove that $f(g) =\sum_{j=1}^p\log(X_j^\top H (A_j(g))^{-1} H X_j)$, where $A_j(g):= nI_n+ \gamma_n^{-1}\sum_{(k,j)\in\mathcal{E}}g_{kj} H Z_{j|k}Z_{j|k}^\top H$ with $Z_{j \mid k} = [b_{1kj}(X_k), b_{2kj}(X_k), \dots, b_{R_nkj}(X_k)]\in\mathbb{R}^{n \times R_n}$ and $H = I_n -  \mathbf{1}\mathbf{1}^\top / n$. This expression also extends $f(g)$ to the domain $[0,1]^{|\mathcal{E}|}$. Problem~\eqref{eq:MIP} is equivalent to: 
\begin{align} \label{eq:MIP-reduced}
\begin{split}
    \min_{ \substack{g_{kj} \in \{0,1\},\ \forall {(k,j)\in \mathcal{E}} \\
    \psi \in [1,p]^p } } \quad & f(g) + \lambda_n^2 \sum_{(k,j)\in \mathcal{E}} g_{kj} \quad
    \mathrm{s.t.} \quad  1 -  p + p g_{kj} \leq \psi_j - \psi_k, \quad \forall (k,j) \in \mathcal{E}.
\end{split}
\end{align}
\begin{theorem} \label{thm:convexity}
    The function $f(g)$ on $g\in[0,1]^{|\mathcal{E}|}$ is convex. Hence, the problem~\eqref{eq:MIP-reduced} is a \textit{convex} MIP.
\end{theorem}

The proof of Theorem~\ref{thm:convexity} is given in Appendix~\ref{sec:appendix:convexity}. Section~\ref{sec:speed-up} presents a custom algorithm for solving \eqref{eq:MIP-reduced}. Given an optimal $\widehat{g}$ to \eqref{eq:MIP-reduced}, the corresponding coefficient estimates are obtained by one additional ridge regression for each node: $\widehat{\beta}_j = \argmin_{\beta_j} \|X_j - Z_j\beta_j\|_n^2 + \gamma_n \|\beta_j^{-0}\|_2^2~~\text{s.t.}~\beta_{rkj}=0 \text{ if }\widehat{g}_{kj}=0$ for all $r \in [R_n]$.

\section{Computation}\label{sec:speed-up}

Formulation~\eqref{eq:MIP-reduced} cannot be directly solved via commercial solvers such as Gurobi. We present a custom cutting plane algorithm \citep{duran1986outer} to certifiably solve \eqref{eq:MIP-reduced}. In this approach, we sequentially lower bound $f(g)$ via supporting hyperplanes and solve the resulting integer linear program.
Our cutting plane approach delivers both upper and lower bounds to  \eqref{eq:MIP-reduced} and hence produces optimality certificates. Due to the non-quadratic structure of the objective function and the DAG constraint, our approach differs from earlier applications of cutting plane method~\cite{bertsimas2019sparse,behdin2026multi}. 
We also develop additional strategies to accelerate the optimization procedure. Together, these computational strategies ensure our formulation scales to moderate graph sizes (over 30 nodes).

\subsection{A cutting-plane method}
\label{sec:computation:cutting-plane}


Here we discuss our cutting-plane method for Problem~\eqref{eq:MIP-reduced}. Given points $g^{(m)}\in[0,1]^{|\mathcal{E}|}$, $m=1,\ldots,M$, we construct supporting hyperplanes: $f(g^{(m)}) + \nabla f(g^{(m)})^\top (g - g^{(m)})$ to lower bound $f(g)$ over $g \in [0,1]^{|\mathcal{E}|}$.
These supporting hyperplanes denoted by $\mathcal C$ in Algorithm~\ref{alg:cutting-plane} form a piecewise-linear global lower approximation of the convex function $f(g)$.
Replacing $f(g)$ with this approximation 
yields an mixed integer linear program (MILP) which we call the master problem:
\vspace{-0.1in} 
\begin{align} \label{eq:MIP-CP}
\begin{split}
    \min_{ \substack{g_{kj} \in \{0,1\},\ \forall {(k,j)\in \mathcal{E}} \\
    \psi \in [1,p]^p, \vartheta\in\mathbb{R} } } \quad & \vartheta + \lambda_n^2 \sum_{(k,j)\in \mathcal{E}} g_{kj} \\
    \mathrm{s.t.} \quad & 1 -  p + p g_{kj} \leq \psi_j - \psi_k, \quad \forall (k,j) \in \mathcal{E}, \\
    & \vartheta \geq f(g^{(m)}) + \nabla f(g^{(m)})^\top (g - g^{(m)}), \quad \forall m\in [M]
\end{split}
\end{align}
\vspace{-0.1in} 
where $\vartheta$ serves as the piecewise linear lower bound to $f(g)$.
At iteration $t$, let $\left(g^{(t)},\vartheta^{(t)},\psi^{(t)}\right)$ denote an optimal solution to the MILP master problem. Define $L^{(t)} := \vartheta^{(t)} +\lambda_n^2\sum_{(k,j)\in\mathcal E} g_{kj}^{(t)}$, and $U^{(t)} := f(g^{(t)})+\lambda_n^2\sum_{(k,j)\in\mathcal E} g_{kj}^{(t)}$. By construction, $L^{(t)}$ and $U^{(t)}$ are lower and upper bounds on the objective value of ~\eqref{eq:MIP-reduced}. If $L^{(t)}=U^{(t)}$, global optimality is certified. Otherwise, we add a supporting hyperplane at $g^{(t)}$ and re-solve the master problem. We repeat until the bounds are within a prescribed tolerance $\tau^{\rm early}$. We discuss the choice of $\tau^{\rm early}$ and the early-stopping criterion further in Section~\ref{sec:computation:early-stop}. 

At the start of the cutting-plane algorithm, we obtain a collection of initial cuts corresponding to $g^{(m)}$, $m = 1,\ldots, T_{\rm frac}$ by applying the cutting-plane procedure to the {\emph{continuous relaxation}} of Problem~\eqref{eq:MIP-reduced}, which involves relaxing the binary variables $g_{kj}$ to $[0,1]$ in \eqref{eq:MIP-CP}. Solving the continuous relaxation of \eqref{eq:MIP-CP} can be done efficiently with linear programming. The full procedure is given in Algorithm~\ref{alg:cutting-plane}. 

\begin{algorithm}[!ht]
\caption{Cutting-plane algorithm}
\label{alg:cutting-plane}
\begin{algorithmic}[1]
\State \textbf{Input:} Initialize the set of cutting planes $\mathcal C$ corresponding to
$f(g^{(m)}) + \nabla f(g^{(m)})^\top(g-g^{(m)})$ where $m=1, \ldots, T_{\text{frac}}$;
stopping tolerance $\tau^{\rm early}\geq 0$.
\For{$t=T_{\rm frac}+1,T_{\rm frac}+2,\ldots$}
    \State Solve the master MILP~\eqref{eq:MIP-CP}
   using all cuts in $\mathcal C$.
    \State Let $(g^{(t)},\vartheta^{(t)},\psi^{(t)})$ denote an optimal solution to \eqref{eq:MIP-CP}.
    \If{$ f(g^{(t)}) - \vartheta^{(t)}\leq\tau^{\rm early}$}
        \State \Return $g^{(t)}$.
    \Else
        \State Expand $\mathcal C$ by including 
        $f(g^{(t)}) + \nabla f(g^{(t)})^\top(g-g^{(t)})$.
    \EndIf
\EndFor
\end{algorithmic}
\end{algorithm}

\begin{theorem}[global optimality] \label{thm:optimality}
Algorithm~\ref{alg:cutting-plane} terminates after finitely many iterations. Let $T$ denote the termination iteration. Compared to the optimal solution of the original problem~\eqref{eq:MIP}, the returned $g^{(T)}$ will have an optimality gap of at most $\tau^\text{early}$. 
\end{theorem}
See Appendix~\ref{sec:appendix:cutting-plane} for the proof of Theorem~\ref{thm:optimality}, and Appendix~\ref{sec:appendix:more-cutting-plane} for additional details of the algorithm. 

\subsection{Early stopping and computational enhancements} 
\label{sec:computation:early-stop}

Obtaining a small gap between the upper bound $U^{(t)}$ and the lower bound $L^{(t)}$ can be computationally expensive and unnecessary in practice. We often observe that the upper bound approaches the optimal value quickly, while the lower bound improves more slowly. We terminate the procedure once $U^{(t)}-L^{(t)}\leq \tau^{\rm early}$,
where $\tau^{\rm early}>0$ is a user-specified tolerance, and return the feasible solution $ g^{(t)}$. This early-stopping strategy is commonly used to accelerate mixed-integer optimization ~\citep{chen2023unified, kucukyavuz2023consistent, xu2025integer}.
Although the resulting solution need not be globally optimal, its objective value is guaranteed to lie within $\tau^{\rm early}$ of the global optimum, a guarantee generally unavailable for many nonconvex optimization methods. 
See Appendix~\ref{sec:appendix:more-cutting-plane} for a numerical illustration. In Section~\ref{sec:theory}, we provide a theoretical justification for the early-stopping estimator when $\tau^{\rm early}$ is sufficiently small.

Prior structural information can also accelerate the MIP procedure.
A common approach is to impose a super-structure~\citep{chickering2002optimal, buhlmann2014, solus2021consistency,  kucukyavuz2023consistent, xu2025integer}, which can be estimated through a preliminary neighborhood-selection step, for example using GAMBoost~\citep{hothorn2010model}. Prior information on the topological ordering can also be incorporated, drawing on domain-expert knowledge or stable ordering patterns identified by bootstrapping or other reliable methods. Both types of prior information can be encoded through additional constraints in the optimization problem. See Appendices~\ref{sec:appendix:superstruc} and \ref{sec:appendix:partial-stable} for details.

\section{Statistical Guarantees}
\label{sec:theory}

We establish statistical guarantees for our estimator \eqref{eq:MIP}. Section \ref{sec:theory:setup} introduces the necessary notation as well as Assumptions \ref{assump:superstructure-sparsity}--\ref{assump:beta-min}. 
Section \ref{sec:theory:recovery} introduces Theorem~\ref{thm:recovery}, which states that our estimator perfectly recovers the true graph with high probability.
We also provide an analogous guarantee for the early-stopping suboptimal solution in addition to the optimal solution.
In Appendix \ref{sec:supporting-theory}, we present supporting results for establishing correct permutation and convergence of noise variances of our estimator.

The theoretical analysis in this section substantially extends those of \cite{buhlmann2014} and \cite{vandegeer2013}, whose results do not suffice for our guarantees. Specifically, \cite{buhlmann2014} establish consistency only for the estimated topological ordering, not the graph, and do not accommodate regularization. Their approach treats graph estimation as a separate step following ordering estimation, whereas we {\emph{jointly}} estimate all model parameters, including the DAG structure. The analysis of \cite{vandegeer2013} is restricted to linear models. Our analysis unifies and extends these two lines of work to establish graph recovery consistency in non-linear settings.

\subsection{Setup and assumptions}
\label{sec:theory:setup}

Let $s^\star := |\mathcal{G}^\star|$ denote the sparsity level of the true graph $\mathcal{G}^\star = (\mathrm{V},E^\star)$. We consider the estimator \eqref{eq:MIP} with two additional constraints. First, edges are restricted to be in a super-structure $E^\circ$, for which we assume that $E^\star \subseteq E^\circ$. This can be achieved by the neighborhood selection approach described in Section~\ref{sec:computation:early-stop} or Appendix~\ref{sec:appendix:superstruc}. Existing theoretical results  \citep[e.g.,][Lemma 4]{buhlmann2014} guarantee that the estimated super-structure is a supergraph of the true graph. 
As a result, for simplicity, we assume access to such a super-structure in the remainder. Second, we add a constraint to the optimization problem that the number of edges in a candidate graph does not exceed $s_n$, where conditions on $s_n$ are described later.

The solution to \eqref{eq:MIP} yields the estimated parameters ${\theta}=\left(\{\beta_{rkj}\}_{(k,j)\in\mathcal{E},, r\in[R_n]},\ \{\sigma_j\}_{j\in\mathrm{V}}\right) \in {\Theta}$, where $\widehat{\sigma}_j^2 = \left\|X_j - Z_j \widehat{\beta}_j \right\|_n^2 + \gamma_n \|\widehat{\beta}_j^{-0}\|_2^2 $.
The induced graph $\mathcal{G}(\widehat{\theta})$ contains an edge from $k$ to $j$ if and only if $\sum_{r=1}^{R_n} |\widehat{\beta}_{rkj} |^2 \neq 0$.
Let $\widehat{s} := |\mathcal{G}(\widehat{\theta})|$ denote the number of selected edges. We similarly define the early-stopping estimates $\widehat{\theta}^{\rm early}= \left( \{ \widehat{\beta}_{rkj}^{\rm early} \}_{(k,j) \in E^\circ,\ r \in [R_n]},  \{\widehat{\sigma}^{\rm early}_j\}_{j\in \mathrm{V}} \right)$ and $\widehat{s}^{\rm early}$. Throughout, we allow $p$ to grow with $n$. The condition $n \gtrsim p (\log p)^2$ is sufficient for Theorem~\ref{thm:recovery}.


\begin{assumption}[parents and super-structure sparsity]\label{assump:superstructure-sparsity}
Each node $j\in \mathrm{V}$ has at most $K$ parents: $|\mathrm{pa}(j) | \leq K$.
Moreover, the super-structure $E^\circ$ contains all true edges and is sparse: $E^\star \subseteq E^\circ$ and $s_j:=|\{k: (k,j)\in E^\circ\}| \leq K$ for each node $j\in \mathrm{V}$. 
\end{assumption}

\begin{assumption}[model sparsity]\label{assump:model-sparsity}
    The quantity $s_n$ that constrains the number of edges satisfies $s^\star \leq s_n$.
\end{assumption}

\begin{assumption} [Sobolev function class with bounded basis functions]  \label{assump:sobolev}
For all function classes $\{ \mathcal{F}_{kj} \}_{(k,j) \in \mathcal{E}}$, the basis functions $\{b_{rkj}(\cdot)\}_{r=1}^\infty$ are orthogonal, satisfying $\mathbb{E}[b_{rkj}(X_k)]=0$, $\mathbb{E}[b_{rkj}(X_k) b_{r'kj}(X_k)] = \upsilon \mathds{1}(r = r')$ for some $\upsilon \in (0, 1]$, and uniformly bounded with $\sup_{r}|b_{rkj}(\cdot)| \leq 1$. Furthermore, their coefficients decay at a fast rate: $\sum_{r=1}^\infty |\beta_{rkj}| \cdot r^\eta \leq C$ and $\eta > 1$ for any $(k,j)\in \mathcal{E}$.
\end{assumption}

Assumption \ref{assump:superstructure-sparsity} requires the super-structure to contain all true parents while allowing at most $K$ candidate parents per node. This parallels Condition (B2) in \cite{buhlmann2014}, which assumes that preliminary neighborhood selection contains all relevant variables while selecting a bounded number of candidates per node, with probability tending to one. Such super-structures can be estimated using methods like sparse additive models \citep{ravikumar2009sparse}, and enjoy consistency guarantees under suitable conditions.   
Assumption \ref{assump:model-sparsity} imposes a global sparsity constraint on the entire true DAG. Again, sparsity assumptions are common in high-dimensional causal discovery, including Condition (B2) of \cite{buhlmann2014} and Conditions 3.3--3.4 of \cite{vandegeer2013}. Section 2.1.1 of \citet{vandegeer2013} also explicitly discusses different sparsity regimes based on the total number of true edges.

Assumption \ref{assump:sobolev} maps each function class $\mathcal{F}_{kj}$ into a Sobolev ellipsoid, allowing control over both the complexity of $\mathcal{F}_{kj}$ and the truncation error $\|f_{kj} - \beta_{0kj} - \sum_{r=1}^{R_n} \beta_{rkj} b_{rkj}\|_\infty$ for all $f_{kj} \in \mathcal{F}_{kj}$. This assumption was also adopted in \cite{ravikumar2009sparse}. In Appendix~\ref{sec:appendix:sobolev}, we provide a choice of basis functions that satisfy the conditions in Assumption~\ref{assump:sobolev}.

\vspace{0.1in}
\emph{Additional notation for variances}:
Let $\vec{g}_j := \{k: g_{kj} = 1, (k,j) \in E^\circ\}$ be an index set of underlying parents of node $j$. Then we can define the additive function classes $\mathcal{F}^{\oplus \vec{g}_j}:= \bigoplus_{k\in \vec{g}_j} \mathcal{F}_{kj}$ and $\mathcal{F}_n^{\oplus \vec{g}_j}:= \bigoplus_{k\in \vec{g}_j} \mathcal{F}_{kj, n}$. With the function $\mathcal{R}_j(f; \vec{g}_j):= [ X_j - \sum_{k\in \vec{g}_j} f_{kj}(X_k) ]^2$, the residual variances after projecting $X_j$ onto $\mathcal{F}^{\oplus \vec{g}_j}$ and $\mathcal{F}_n^{\oplus \vec{g}_j}$ are given by
$
    \nu_j(\vec{g}_j) := \min_{f \in \mathcal{F}^{\oplus\vec{g}_j}} \mathbb{E} \left[ \mathcal{R}_j(f; \vec{g}_j) \right] $
and
$
    \nu^n_j(\vec{g}_j) := \min_{f \in \mathcal{F}_n^{\oplus \vec{g}_j}} \mathbb{E} \left[ \mathcal{R}_j(f; \vec{g}_j) \right],
$ respectively.
Since each node has at most $K$ parents, we can denote the maximal gap between these variances as $d_{n,p} :=  \max_{j\in\mathrm{V}} \max_{|\vec{g}_j| \in \{0, \dots, K\} }  \left| \nu_j(\vec{g}_j) - \nu^n_j(\vec{g}_j) \right|$, the maximal value of $\nu_j(\vec{g}_j)$ as $\overline{\nu}_p := \max_{j\in\mathrm{V}} \mathbb{E}[X_j^2]$, and the minimal value as $\underline{\nu}_p := \min_{j\in\mathrm{V}} \min_{|\vec{g}_j| =K  } \nu_j(\vec{g}_j)$.

\begin{assumption}[bounded noise variances]\label{assump:bound-variance}
    The noise variances are uniformly bounded from above by a constant $\overline{\sigma}^2 < \infty$:
    $\sigma_j^{\star2} \leq \overline{\sigma}^2$ for all $j \in \mathrm{V}$. In addition, the variance $\nu_j(\vec{g}_j)$ is uniformly bounded below by a constant $\underline{c} > 0$: $\underline{\nu}_p \geq \underline{c}$. 
\end{assumption}

Assumption~\ref{assump:bound-variance} ensures that $\overline{\nu}_p \leq \overline{c} $ for some $\overline{c} < \infty$, thereby facilitating the boundedness of the estimated variances. This condition is similar to Condition (A3) in \cite{buhlmann2014} and Assumption 4 in \cite{xu2025integer}. 
Intuitively, the estimated variance $\widehat{\sigma}_j^2$ should be bounded below by an empirical quantity concentrating around $\nu_j^n(\widehat{g}_j)$, where $\widehat{g}_j:=\{k: \widehat{g}_{kj} = 1, (k,j)\in E^\circ\}$ denotes the estimated parent set on node $j$. Thus, the estimated variances are approximately bounded below by the smallest possible values of $\nu_j^n(\widehat{g}_j)$ as well. 
One may observe that $\underline{\nu}_p$ and $\overline{\nu}_p$ are appropriate bounds if $\nu_j^n(\vec{g}_j)$ does not deviate significantly from $\nu_j(\vec{g}_j)$, or equivalently $d_{n,p} = o(1)$. See Appendix~\ref{sec:appendix:variances} for a detailed discussion on the control of $\bar{\nu}_p$ and $d_{n,p}$.

\begin{assumption}[functional compatibility]\label{assump:compatibility}
    Under the distribution $P$ generated from \eqref{eq:DAG-model}, there exists $\phi^2 > 0$ such that for all $\gamma \in \mathbb{R}^p$ and $j\in[p]$, 
    $
    \| \sum_{k=1}^p \gamma_k f_{kj}(X_k) \|_{L_2}^2 \geq \phi^2 \sum_{k=1}^p \gamma_k^2\|f_{kj}(X_k)\|_{L_2}^2
    $
    for all $f_{kj} \in \mathcal{F}_{kj}$.
\end{assumption}
Assumption \ref{assump:compatibility} guarantees that the spaces $\mathcal{F}^{\oplus \vec{g}_j}$ and $\mathcal{F}_n^{\oplus \vec{g}_j}$ are closed under the $L_2$ norm. This assumption is standard in high-dimensional additive models \citep{meier2009high} and is directly adopted from Lemma~2 of \cite{buhlmann2014}.

\vspace{0.1in}
\emph{Additional notation for permutations}: A permutation $\pi$ of $\{1,\dots,p\}$ is consistent with a DAG if an edge from $\pi(k)$ to $\pi(j)$ implies $k < j$. In particular, $\pi(1)$ has no parents. A DAG may admit multiple consistent permutations. Let $\Pi^\star$ denote the set of permutations consistent with the true graph $\mathcal{G}^\star$, referred to as true permutations. 

Let $\pi^{-1}$ map each node to its position in the ordering given by permutation $\pi$. Define $\vec{g}(\pi)_j := \{k: \pi^{-1}(k) < \pi^{-1}(j), (k,j)\in E^\circ\}$, the set of parents of node $j$ under permutation $\pi$. The residual variance after projecting $X_j$ onto the function class $\mathcal{F}^{\oplus \vec{g}(\pi)_j}$, denoted by $\nu_j(\vec{g}(\pi)_j)$, quantifies the best possible fit of node $j$ to a DAG with the permutation $\pi$. Particularly, when $\pi \in \Pi^\star$, we have $ \nu_j(\vec{g}(\pi)_j) = \sigma_j^{\star2}$. When $\pi \notin \Pi^\star$, we define the separation between these residual variances and the true variances as: 
\vspace{-0.1in}
\begin{equation}\label{eq:separation}
    \xi_p:= \min_{\pi\notin \Pi^\star} p^{-1} \sum_{j=1}^p \left( \log \nu_j(\vec{g}(\pi)_j) - \log \sigma_j^{\star2} \right).
\vspace{-0.1in}
\end{equation}
Properties of the KL-divergence imply that $\xi_p \ge 0$. See Appendix~\ref{sec:appendix:permutation} for details.

\begin{assumption}[non-negligible separation]\label{assump:separation}
For incorrect permutations, the gap between the best-fitted residual variances and the true variances is non-negligible:
$
p \sqrt{\log p /n} + s_n R_n^{-2\eta}  = o(p\xi_p).
$
\end{assumption}
Assumption~\ref{assump:separation} requires the separation $\xi_p$ to dominate statistical estimation and basis approximation errors. We provide a simple example
illustrating this condition in Example~\ref{ex:assumption-separation} of Appendix~\ref{sec:appendix:permutation}.
Some possible choices for these quantities are: $s_n \asymp R_n \asymp n^{1/5}$. In the extreme case where the number of nodes $p$ is fixed, the assumption holds trivially. Generally speaking, when $p$ varies with $n$, the assumption simplifies to $\sqrt{\log p / n} = o(\xi_p)$ in both the low-dimensional case where $n \asymp p(\log p)^2$ and the high-dimensional case where $n \asymp (\log p)^2$. A similar condition appears in Theorem 3 of \cite{buhlmann2014}, where it is also assumed that $\sqrt{\log p / n} = o(\xi_p)$.



\begin{assumption}[moderately increasing number of basis functions]
\label{assump:bounded-Rn}
    The number of basis functions does not grow too fast: $R_n \lesssim \sqrt{n / \log p}$.
\end{assumption}

\begin{assumption}[beta-min condition] \label{assump:beta-min}
For any $(k,j) \in E^\star$, the true coefficients for the basis functions used in estimation are strong:
    $
    \min_{(k,j) \in E^\star} \sum_{r=1}^{R_n} |\beta^\star_{rkj} |^2 \ge 
    c_\beta \delta_{n,p} > 0$ for some $c_\beta >0$, where $(p +s_n R_n) \log p / n + s_n  R_n^{-2\eta }  = o(\delta_{n,p})
    $.
\end{assumption}


Assumption \ref{assump:bounded-Rn} regulates the growth rate of the number of basis functions. However, $R_n$ cannot be too small either, as this may hinder Assumptions~\ref{assump:separation} and~\ref{assump:beta-min}, due to poor approximation from an insufficient number of basis functions.
Assumption \ref{assump:beta-min} requires a lower bound on the strength of true causal effects so that they can be captured by the estimated graph. For example, we may choose $n \asymp p(\log p)^3$, $s_n \asymp n^{1/5}$, $R_n \asymp n^{1/5}\log p$, and $\min_{(k,j) \in E^\star} \sum_{r=1}^{R_n} |\beta^\star_{rkj} |^2 \asymp \delta_{n,p} \asymp p (\log p)^2 / n$. A similar beta-min condition appears as Condition~3.5 in \cite{vandegeer2013} for linear SEMs, which, without basis expansions, requires each nonzero $\beta_{kj}^\star$ to be at least of order $\sqrt{p \log p / (s_n n)}$. After translating their condition into our example setting, the two beta-min requirements align in order. 

\subsection{Graph recovery results}
\label{sec:theory:recovery}

\begin{theorem} \label{thm:recovery}
    Suppose the conditions in Proposition \ref{prop:identifiability} and Assumptions \ref{assump:superstructure-sparsity}--\ref{assump:beta-min} are satisfied. Assume also that $s_n\delta_{n,p}=o(p\xi_p)$. Take $\lambda_n^2 \in [c_l \delta_{n,p} ,\, c_u \delta_{n,p}]$  where $0<c_l<c_u$, with $c_u$ satisfying the constant restriction specified in Appendix~\ref{sec:appendix:prof-recovery}, and $\gamma_n>0$ satisfying $\gamma_n =o(s_n^{-1} \delta_{n,p})$. When $n$ is sufficiently large, it holds,  with probability at least $1 - 6/p$, that:
    \begin{enumerate}[label = (\arabic*)]
        \item The estimated graph recovers the true structure: $\mathcal{G}(\widehat{\theta}) = \mathcal{G}^\star$.
        
        \item By taking an early-stopping threshold satisfying $\tau^{\rm early} = o(p\xi_{p} \wedge \delta_{n,p})$, the resulting estimated graph recovers the true structure: $\mathcal{G}(\widehat{\theta}^{\rm early}) = \mathcal{G}^\star$.
    \end{enumerate}
\end{theorem}

The proof of Theorem~\ref{thm:recovery} is given in Appendix~\ref{sec:appendix:prof-recovery}. 
It shows that the optimal objective value attained under the true graph $\mathcal{G}^\star$ is strictly lower than that attained under any incorrect graphs.
The tuning parameter $\lambda_n^2$ must be sufficiently large to prevent the selection of redundant edges. Moreover, the sample size should be large enough that $\delta_{n,p} > (p + s_n R_n) \log p /n + s_n R_n^{-2\eta}$ for sufficiently large $n$, as guaranteed by Assumption~\ref{assump:beta-min}. Finally, the early-stopping sub-optimal solution can also achieve exact graph recovery. 

\section{Extension: Incorporating Hierarchical Interactions}
\label{sec:interactions}

In this section, we extend our MIP formulation to accommodate pairwise interactions. In particular, we posit the following generalization of the SEM \eqref{eq:DAG-model} for every $j\in[p]$:
$$ X_j = \sum_{k\in \mathrm{pa}(j)} f^\star_{kj}(X_k) +\sum_{\substack{k<l\\k,l\in\mathrm{pa}(j)}}f^\star_{klj}(X_k,X_l) +\epsilon_j,\quad \epsilon_j \overset{}{\sim} \mathcal{N}(0, \sigma_j^{\star2} ),\quad \sigma_j^\star >0,$$
where $f^\star_{kj}\in\mathcal{F}_{kj}$ and $f^\star_{klj}\in\tilde{\mathcal{F}}_{klj}$ with $\tilde{\mathcal{F}}_{klj} := \left\{
f_{klj}
: f_{klj}(a,b) = \sum_{r_1=1}^{\infty}\sum_{r_2=1}^\infty \zeta_{r_1r_2klj}b_{r_1kj}(a)b_{r_2lj}(b)
\right\}$. Based on computational considerations (see Remark~\ref{remark:interaction}), we impose a hierarchical interaction structure on the fitted model, meaning that an interaction between two variables can be included only if both corresponding parent edges are selected. As before, our estimator fits a model using $R_n$ basis functions, and concatenates the pairwise product of basis functions into an extended random vector, and similarly concatenates the corresponding coefficients:
\vspace{-0.2in}
\begin{eqnarray*}
\begin{aligned}
\Omega_j &:= 
\left[b_{r_1}(X_k)b_{r_2}(X_l),1 \leq r_1,r_2 \leq R_n, 1 \leq k< l \leq p\right] \in \mathbb{R}^{R_n^2 \scriptscriptstyle \binom{p}{2}}, \\
\zeta_j&:= \left[\zeta_{r_1r_2klj}, 1 \leq r_1,r_2 \leq R_n, 1 \leq k < l \leq p \right] \in \mathbb{R}^{R_n^2 \scriptscriptstyle \binom{p}{2}},
\end{aligned}
\end{eqnarray*}
where $\zeta_{r_1r_2klj} =0$ if $k=j$ or $l=j$. The combined main and interaction
effects on $X_j$ are then given by $Z_j\beta_j + \Omega_j\zeta_j$ where the first and second term correspond to main effects and interaction effects, respectively. This leads to a variant of \eqref{eqn:sigmaprofile}, with positive ridge parameters  $\gamma_{n,M},\gamma_{n,I}$ and non-negative regularization parameters $\lambda_{n,M}$ and $\lambda_{n,I}$: 
\begin{eqnarray*}
\begin{aligned}
	\min_{(\beta,\zeta) \in \tilde{\Theta},\sigma}&\quad \sum_{j=1}^p\left\{
	\log(\sigma_j^2)+ \sigma_j^{-2} \left(\left\|	X_j - Z_j\beta_j
	-\Omega_j\zeta_j\right\|_n^2 + \gamma_{n,M}\|\beta^{-0}_{j}\|_2^2+\gamma_{n,I}\|\zeta_j\|_2^2\right)\right\}\\&\quad+ \lambda_{n,M}^2 \sum_{(k,j)\in \mathcal{E}} \mathds{1}\left\{ \|\beta_{kj}\|_2 \neq 0 \right\}+ \lambda_{n,I}^2\sum_{(k,l,j)\in\tilde{\mathcal{E}}} \mathds{1}\left\{ \|\zeta_{\cdot\cdot{klj}}\|_2 \neq 0 \right\}.
    \end{aligned}
	\label{eq:coefficient-node-loss}
\end{eqnarray*}
Here, $\tilde{\mathcal{E}} = \{(k,l,j)\in V \times V \times V: k < l, k \neq j, l \neq j\}$ and $\tilde{\Theta} = \{(\beta,\zeta): \|\beta_{kj}\|_2=0 \text{ or } \|\beta_{lj}\|_2=0 \rightarrow  \|\zeta_{\cdot\cdot{klj}}\|_2 =0 \text{ and }E(\beta) \text{ forms a DAG}\}$. The first constraint imposes the hierarchical interaction structure, while the requirement that $E(\beta)$ form a DAG ensures that the fitted model corresponds to a DAG, with $(k,j)\in E(\beta)$ if and only if $\|\beta_{kj}\|_2\neq 0$. Because of the hierarchical structure, the interaction parameters $\zeta$ do not need to enter the DAG constraint directly.

To formulate a convex MIP, we introduce two sets of binary variables. As before, $\{g_{kj}\in\{0,1\}:(k,j)\in\mathcal{E}\}$ encodes the DAG structure, with the support of the coefficients $\beta_j$ linked directly to $g$. In addition, $\{w_{klj}\in\{0,1\}:(k,l,j)\in\tilde{\mathcal{E}}\}$ indicates whether the pairwise interaction between nodes $k$ and $l$ is included when modeling node $j$, with the support of $\zeta_j$ linked to $w$. We encode the hierarchical structure through $w_{klj} \leq \min\{g_{kj},g_{lj}\}$ for all $(k,l,j)\in\tilde{\mathcal{E}}$, which ensures that an interaction can be active only when both corresponding main effects are present. Following the same strategy as in Section~\ref{sec:convex:reform}, we then project out the continuous variables $\{(\beta_j,\zeta_j),j\in[p]\}$ to obtain the following convex MIP:  
\begin{eqnarray}
\begin{aligned}
\min_{\substack{g_{kj}\in\{0,1\},\forall (k,j)\in\mathcal{E}\\w_{klj}\in\{0,1\} \forall (k,l,j)\in\tilde{\mathcal{E}}\\\psi\in[1,p]^p}}&\quad \tilde{f}(g,w)+ \lambda^2_{n,M}\sum_{(k,j)\in\mathcal{E}}g_{kj} + \lambda^2_{n,I}\sum_{(k,l,j)\in\tilde{\mathcal{E}}}w_{klj} \\
\text{subject-to:}&\quad w_{klj} \leq \min\{g_{kj},g_{lj}\},\quad \forall (k,l,j)\in\tilde{\mathcal{E}},\\
    &\quad 1 -  p + p g_{kj} \leq \psi_j - \psi_k, \quad \forall (k,j) \in \mathcal{E}.
\end{aligned}
\label{eqn:interaction_MIP}
\end{eqnarray}
Here, $\tilde{f}(g,w)$, when extending the entries of $g$ and $w$ to $[0,1]$, is a convex function; see Appendix~\ref{sec:appendix:interaction} for its form and a more detailed description of how we arrive at formulation \eqref{eqn:interaction_MIP}.
\begin{remark}Hierarchical structure simplifies the DAG representation. Under this structure, the binary variables indicating the presence of main effects are sufficient to encode all edges in the DAG. Without hierarchical structure, an interaction could introduce an edge even when the corresponding main effect is absent, requiring additional variables to encode the DAG structure.
\label{remark:interaction}
\end{remark}

\section{Simulations}
\label{sec:experiments}
In this section, we evaluate the performance of our estimators \eqref{eq:MIP-reduced} and \eqref{eqn:interaction_MIP} in recovering the true DAG.
To evaluate structural accuracy, we report the structural Hamming distance (SHD)~\citep{tsamardinos2006max}, which measures discrepancies between the estimated and true graph structures via adjacency matrices. To further assess causal correctness, we report the structural intervention distance (SID)~\citep{peters2015structural}, which evaluates discrepancies in the interventional implications of the estimated and true DAGs. We additionally report the F1 score for recovering directed edges by comparing the estimated and true adjacency matrices.

\subsection{Simulations with estimator \texorpdfstring{\eqref{eq:MIP-reduced}}{(7)} }
\label{sec:experiment:comparison}
We demonstrate the effectiveness of our estimator \eqref{eq:MIP-reduced} on Erd\H{o}s-R\'{e}nyi random graphs, with the number of nodes $p$ ranging from $10$ to $25$ and graph density ranging from $0.15$ to $0.5$. For each setting, we conduct 50 independent trials. In each trial, we generate a new graph and independently sample the causal effect function $f^\star_{kj}$ for each edge $(k,j)\in E$ from a prescribed collection of nonlinear functions: $\{\sin(x),\cos(x),\tanh(x),\arctan(x),\tanh(x^3)\}$. We generate heterogeneous noise standard deviations as $\sigma^\star_j \propto \exp(\tau Z_j)$, where $Z_j \overset{\mathrm{i.i.d.}}{\sim} N(0,1)$, and normalize them so that $p^{-1}\sum_{j=1}^p \sigma_j^{\star2}=1$. The parameter $\tau$, which controls the degree of heterogeneity, is chosen from $\{0.1,0.2\}$. For each setting, we generate $n\in\{300,500\}$ independent and identically distributed samples from the SEM in~\eqref{eq:DAG-model}.

We compare our proposed estimator, referred to as \textit{MIP-m}, with the following approaches for learning DAG structures: the polynomial-time algorithm \textit{NPVAR} for nonparametric DAG learning~\citep{gao2020polynomial}; the equal-variance algorithm \textit{EqVar}, which uses a top-down procedure to estimate the node ordering in linear models~\citep{chen2019causal}; the score-based method \textit{NoTears}, which learns nonlinear SEMs through a continuous and nonconvex formulation of the acyclicity constraint~\citep{zheng2020learning}; the \textit{RESIT} algorithm, which uses nonparametric regression followed by independence tests~\citep{peters2014causal}; the heuristic algorithm \textit{CAM} for causal additive models~\citep{buhlmann2014}; the generalized score-based method \textit{GSGES}, which combines greedy equivalence search with generalized score functions based on nonparametric regression~\citep{huang2018generalized}, and for which
we obtain a consistent DAG extension of the estimated CPDAG; the information-theoretic method \textit{GLOBE} based on the minimum description length principle~\citep{mian2021discovering}; and the information-theoretic method \textit{TOPIC} which learns causal topological orderings using approximations to Kolmogorov complexity~\citep{xu2025information}.

\begin{figure}[!ht]
    \centering
    \includegraphics[width=\textwidth]{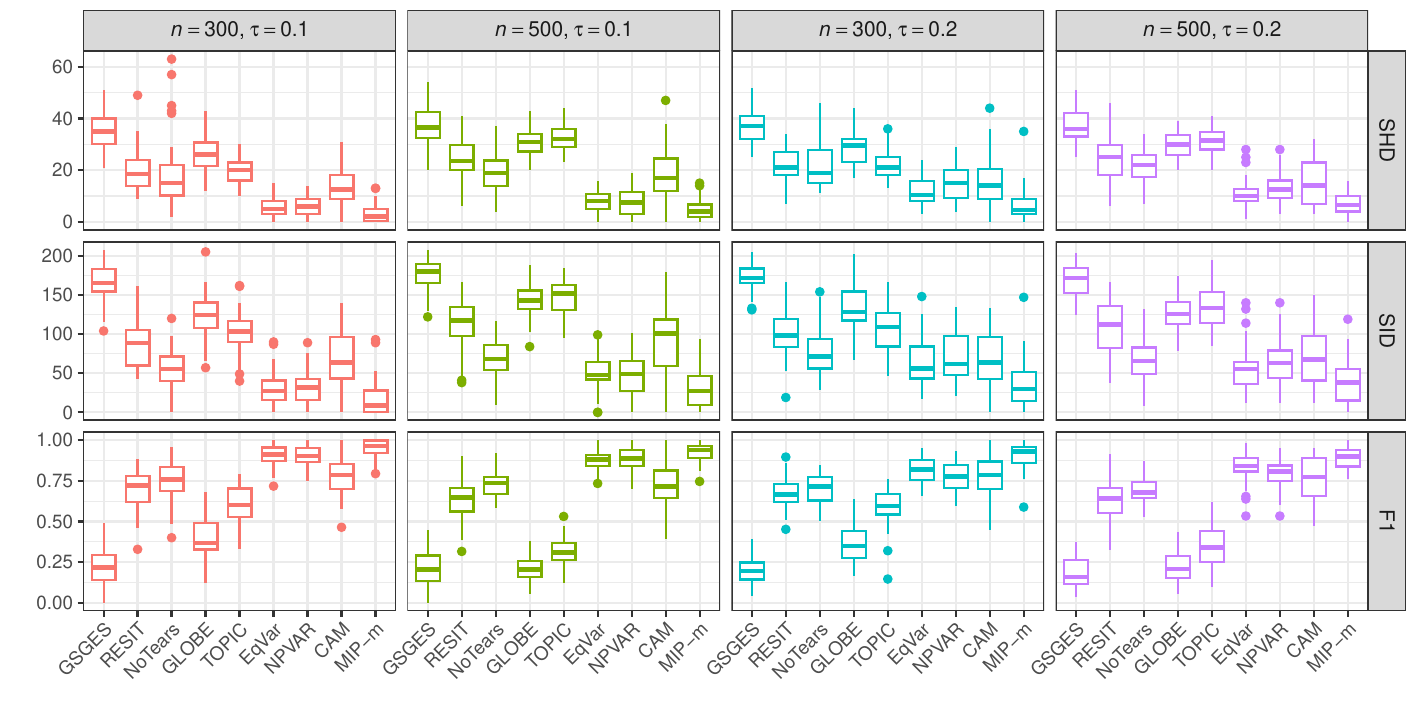}
    \caption{\small Performance comparison between our method \textit{MIP-m} and the baseline methods for Erd\H{o}s-R\'{e}nyi graphs with $p=15$ variables and edge density $0.3$. Each column represents a different combination of sample size $n$ and noise-heterogeneity parameter $\tau$, while each row corresponds to a different performance metric: SHD, SID, and F1. Each box summarizes results from 50 independent simulation trials.
    }
    \label{fig:compare_vary_ck}
\end{figure}

We use an estimated super-structure obtained from the procedure described in Appendix~\ref{sec:appendix:superstruc}. The same super-structure is also imposed on \textit{CAM}, \textit{NPVAR}, \textit{EqVar}, and \textit{RESIT} to ensure a fair comparison. Additional prior information, obtained by bootstrapping \textit{CAM} or \textit{NPVAR} (see Section~\ref{sec:appendix:partial-stable}), is used only for large graphs.
We use the same set of basis functions across all edges, constructed using cubic splines with three internal knots. The tuning parameter $\lambda_n^2$ in estimator~\eqref{eq:MIP} is selected from a grid of candidate values using the BIC criterion (see Section~\ref{sec:convex:constraint}), while $\gamma_n$ is fixed as $0.2$ to reduce computational burden; our sensitivity analysis in Appendix~\ref{sec:appendix:sensitive-gamma} shows that the results are insensitive to the choice of $\gamma_n$.
Finally, we impose a time limit of 180 seconds on each run of the cutting-plane procedure; if the time limit is reached before convergence, we return the best feasible solution available at termination.

\begin{figure}[!ht]
    \centering
    \includegraphics[width=\textwidth]{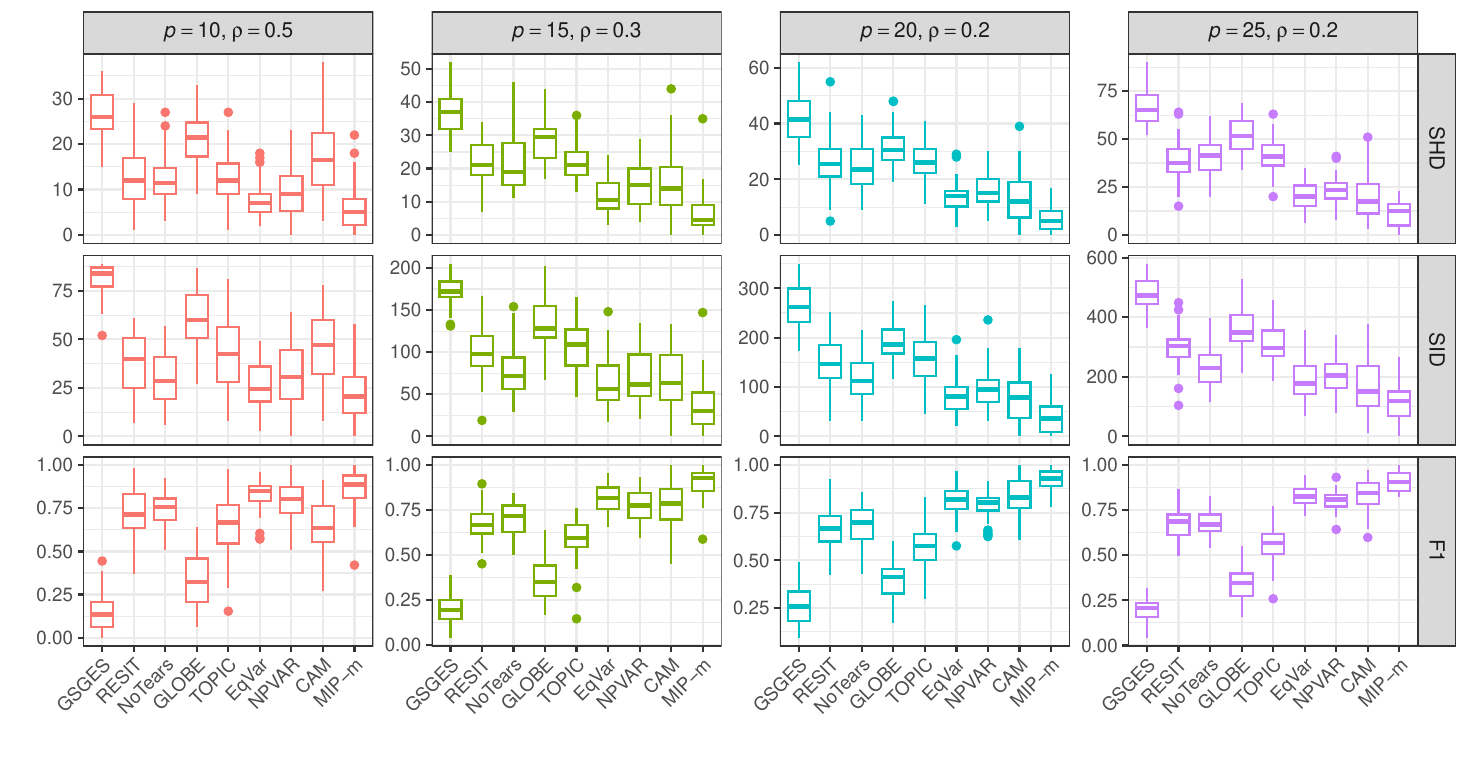}
    \caption{\small Performance comparison between our method \textit{MIP-m} and the baseline methods with $n=500$ samples and noise-heterogeneity parameter $\tau=0.2$. Each column represents a different graph complexity (size $p$, density $\rho$), while each row corresponds to a different performance metric: SHD, SID, and F1. Each box summarizes results from 50 independent simulation trials.
    }
    \label{fig:compare_vary_graph}
\end{figure}

Figures~\ref{fig:compare_vary_ck} and \ref{fig:compare_vary_graph} demonstrate that \textit{MIP-m} consistently achieves the lowest SHD and SID and the highest F1 score across the considered settings. Figure~\ref{fig:compare_vary_ck} shows that increasing the noise-heterogeneity parameter $\tau$ particularly affects \textit{EqVar} and \textit{NPVAR},  which rely on equal-variance assumptions, while also making the problem moderately more difficult for most methods. In comparison, increasing the sample size from $n=300$ to $n=500$ leads to relatively modest changes. Figure~\ref{fig:compare_vary_graph} shows that graph complexity substantially increases the difficulty of causal discovery: as the graphs become larger and denser, SHD and SID increase markedly for most methods. Nevertheless, \textit{MIP-m} remains stable and maintains strong F1 performance even in the more complex settings.
Complete results for all graph structures, together with running times and further implementation details, are presented in Appendix~\ref{sec:appendix:setup-comparison} and Tables~\ref{tab:compare-c1}--\ref{tab:compare-c4}. 

We provide several additional experiments. In Appendix~\ref{sec:appendix:10-graphs}, we compare \textit{MIP-m} with the baseline methods on 10 additional semi-synthetic graph structures motivated by real-world networks. In Appendix~\ref{sec:appendix:alter-basis}, we consider an alternative basis construction using radial basis functions. In Appendix~\ref{sec:appendix:role-rn}, we study the effect of increasing the number of basis functions $R_n$, which improves graph recovery. In Appendix~\ref{sec:appendix:sensitive-nongaussian}, we conduct a sensitivity analysis under non-Gaussian noise distributions. Finally, in Appendix~\ref{sec:appendix:interaction}, we investigate settings in which the causal effects are non-additive. In most scenarios, our \textit{MIP-m} method maintains its advantage over the competing baselines.

\subsection{Simulations with estimator \texorpdfstring{\eqref{eqn:interaction_MIP}}{(10)} }
To demonstrate the flexibility of our framework and assess the effect of modeling pairwise interactions on graph recovery, we compare \textit{MIP-m} with estimator \eqref{eqn:interaction_MIP}, denoted \textit{MIP-i}, and with \textit{GLOBE}, \textit{TOPIC}, \textit{NoTears} and \textit{GSGES}. We use $p=8$, $n=300$ and 20 replications per setting. We consider one or three interactions with nonlinear main effects and strengths $\eta\in\{0.5,1,2\}$, three interactions with linear main effects and $\eta\in\{0.5,2\}$, and three interactions with cosine main effects and $\eta=1$.  Here $\eta$ is the coefficient multiplying each interaction term after empirical centering and scaling to unit standard deviation. Additional settings include an additive model, a pure-product interaction and a third parent with effect strength $b\in\{0.2,0.5\}$.

We generate eight-node DAGs following the specifications described in Appendix~\ref{sec:appendix:interaction:dgp} and independently sample the nonlinear main-effect functions from $\{\sin(x)$, $\cos(x)$,$\tanh(x)$, $\arctan(x),\tanh(x^3)\}$ and the interaction functions from $\{uv,\cos(u+v),\sin(uv), u/\sqrt{v^2+0.25}\}$, each uniformly at random. We evaluate the interaction functions at empirically standardized parent variables and center and scale each main-effect and interaction term to have empirical mean zero and unit standard deviation before multiplying by its effect coefficient. We generate independent Gaussian noise with heterogeneous variances using the same scheme as in Section~\ref{sec:experiment:comparison}, with $\tau\in\{0.1,0.2\}$ and average noise variance normalized to one.

Both MIP estimators use cubic splines with three internal knots, yielding six columns per main effect and 36 per interaction, with both ridge parameters set to one. Because interaction blocks introduce many potentially correlated columns, we found that replacing the nominal parameter count in BIC by the effective degrees of freedom (EDF) generally improved graph recovery for \textit{MIP-i}. We therefore use this criterion, denoted $\mathrm{BIC}_{\mathrm{EDF}}$, for both MIP estimators, following its use in ridge regression \citep{boonstra2015small}. More details on $\mathrm{BIC}_{\mathrm{EDF}}$ are provided in Appendix~\ref{sec:appendix:interaction:screening}. For each MIP estimator, preliminary screening constructs candidate sets and ranks sparsity-penalty configurations using the $\mathrm{BIC}_{\mathrm{EDF}}$ scores of provisional models fitted by ridge regression. We retain the two configurations with the lowest scores and solve the MIP for each using its screened candidate set. The final model is the resulting model with the lowest $\mathrm{BIC}_{\mathrm{EDF}}$. Appendices~\ref{sec:appendix:interaction:screening} and~\ref{sec:appendix:interaction:dgp} provide methodological and DGP details. SHD results appear in Table~\ref{tab:interaction-simulations-tau01} and Appendix Table~\ref{tab:interaction-simulations-tau02} for $\tau=0.1$ and $0.2$, respectively. The SID results appear in Appendix Table~\ref{tab:interaction-sid-tau0102}.

\begin{table}[!ht]
\centering
\small
\resizebox{\textwidth}{!}{
\begin{tabular}{lcccccc}
\toprule
Setting & \emph{MIP-m} & \emph{MIP-i} & \emph{GLOBE} & \emph{TOPIC} & \emph{NoTears} & \emph{GSGES}\\
\midrule
\multicolumn{7}{l}{$\tau=0.1$}\\
\midrule
One interaction, $\eta=0.5$
& 1.50 (0.36) & 0.65 (0.27) & 6.75 (0.25) & 6.45 (0.39) & 3.35 (0.68) & 7.50 (0.46)\\
One interaction, $\eta=1$
& 1.50 (0.38) & 0.60 (0.21) & 6.40 (0.21) & 6.55 (0.39) & 4.00 (0.81) & 7.50 (0.46)\\
One interaction, $\eta=2$
& 2.20 (0.42) & 0.70 (0.24) & 6.55 (0.23) & 7.05 (0.48) & 4.90 (0.96) & 7.20 (0.41)\\
Three interactions, $\eta=0.5$
& 2.55 (0.46) & 0.85 (0.35) & 6.85 (0.33) & 8.10 (0.54) & 3.85 (0.85) & 7.40 (0.49)\\
Three interactions, $\eta=1$
& 3.95 (0.63) & 0.25 (0.16) & 6.55 (0.25) & 7.30 (0.48) & 3.65 (0.69) & 7.50 (0.46)\\
Three interactions, $\eta=2$
& 5.65 (0.72) & 0.25 (0.16) & 6.55 (0.15) & 7.55 (0.45) & 5.95 (1.03) & 7.90 (0.45)\\
Linear mains, three interactions, $\eta=0.5$
& 5.90 (0.48) & 1.90 (0.43) & 6.95 (0.33) & 8.60 (0.47) & 1.30 (0.45) & 7.40 (0.45)\\
Linear mains, three interactions, $\eta=2$
& 6.35 (0.77) & 0.35 (0.17) & 6.70 (0.22) & 8.40 (0.59) & 4.75 (0.91) & 7.50 (0.46)\\
Additive
& 0.60 (0.21) & 0.50 (0.20) & 6.35 (0.29) & 6.05 (0.43) & 3.35 (0.73) & 8.10 (0.59)\\
Pure interaction
& 2.50 (0.22) & 0.60 (0.21) & 6.50 (0.17) & 6.80 (0.43) & 8.35 (1.26) & 7.30 (0.49)\\
Three parents, $b=0.2$
& 2.65 (0.40) & 1.45 (0.20) & 7.30 (0.18) & 7.30 (0.46) & 4.45 (0.90) & 8.60 (0.46)\\
Three parents, $b=0.5$
& 2.05 (0.43) & 0.75 (0.31) & 7.25 (0.19) & 8.00 (0.53) & 4.05 (0.74) & 9.10 (0.51)\\
Cosine mains, three interactions, $\eta=1$
& 0.40 (0.18) & 0.00 (0.00) & 5.65 (0.18) & 5.55 (0.34) & 7.80 (0.89) & 7.10 (0.42)\\
\bottomrule
\end{tabular}
}
\caption{\small Mean directed SHD (standard error) over 20 replications with $p=8$ and $n=300$ and $\tau=0.1$. Both MIP estimators use
$\mathrm{BIC}_{\mathrm{EDF}}$ for final selection. Main effects are sampled nonlinear functions unless otherwise indicated.}
\label{tab:interaction-simulations-tau01}
\end{table}

\textit{MIP-i} substantially improves graph recovery when interactions are present, particularly with several moderate or strong interactions. It achieves the lowest mean SHD among the compared methods in 25 of the 26 conditions. The exception is the weak-interaction setting with linear main effects and $\tau=0.1$, where \textit{NoTears} performs slightly better. In the additive setting, \textit{MIP-m} and \textit{MIP-i} perform similarly, suggesting little loss from allowing interactions when none are present. Weak effects remain challenging, and screening is more prone to excluding true parents or interactions in these settings.

The computational cost nevertheless grows rapidly with dimension. A complete candidate graph has $p(p-1)$ parent indicators and $p\binom{p-1}{2}$ interaction indicators, with $R_n^2$ basis columns per interaction block. Screening accounts for approximately $91\%$ of the recorded computation time in these experiments. Consequently, scalability remains a concern, requiring more computationally efficient screening methods.

\section{Real Data Analysis}
\label{sec:real-data}

\subsection{Causal chambers}
\label{sec:experiment:causal_chamber}

Causal chambers~\citep{gamella2025causal} are computer-controlled devices that allow for manipulating and measuring variables from physical systems, providing a rich testbed for causal discovery algorithms. For the light tunnel chamber, we consider the true graph on eight variables. Among all the causal edges, the only non-linear causal effect is $\theta_1 - \theta_2 \rightarrow \widetilde{I}_3$, with the true function $f^\star_{(\theta_1 - \theta_2) \rightarrow \widetilde{I}_3}(\theta_1 - \theta_2)= \beta_0 + \beta_1\cos^2(\theta_1 - \theta_2)$ for some $\beta_0, \beta_1 \in\mathbb{R}$.
Figure~\ref{fig:causal_chamber_n} shows that, as the sample size varies from $200$ to $800$, our \textit{MIP} procedure \eqref{eq:MIP-reduced} consistently delivers the strongest overall performance, achieving substantially lower SHD, SID, and false discovery rate (FDR), together with substantially higher true discovery rate (TDR), than the competing methods. The advantage of \textit{MIP} is particularly pronounced in this setting, possibly because its flexible functional representation adapts well to data containing a mixture of linear and nonlinear causal effects. The ground truth, implementation as well as other details are provided in Appendix~\ref{sec:appendix:causal-chambers}.

\begin{figure}[!ht]
    \centering
    \includegraphics[width=0.98\linewidth]{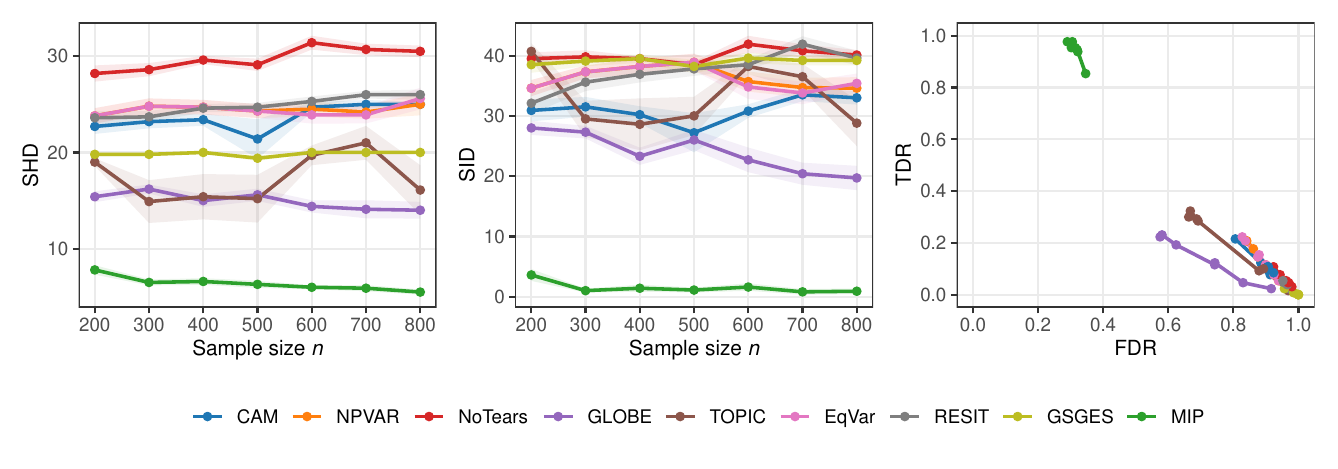}
    \caption{\small Performance on the Causal Chambers dataset across different sample sizes, measured by SHD, SID, TDR, and FDR. Solid lines represent the mean over 10 trials, and shaded regions indicate $\pm 1$ standard error.}
    \label{fig:causal_chamber_n}
\end{figure}

\subsection{Microarray data}
\label{sec:experiment:microarray}

\begin{figure}[!ht]
    \centering
    \includegraphics[width=0.8\textwidth]{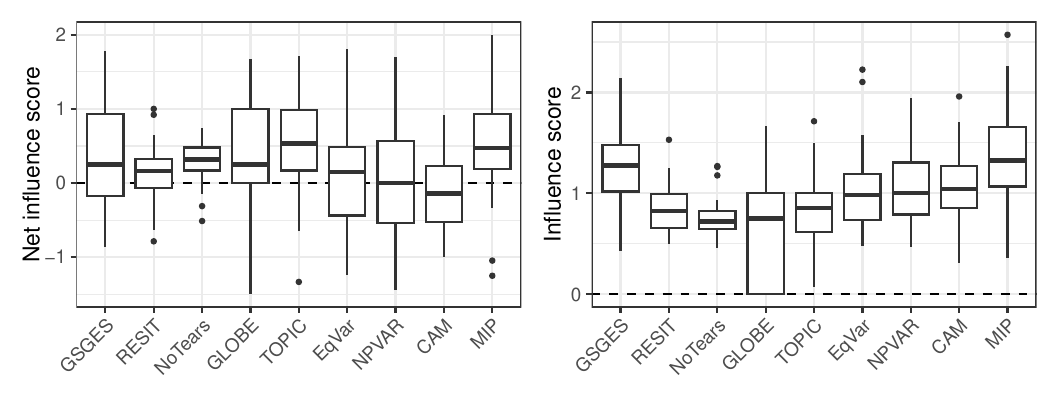}
    \caption{\small Performance comparison on the Arabidopsis MVA pathway data.
    Left panel: the net influence score $(\widehat C-\widehat R)/|\widehat{E}|$. Right panel: the influence score $\widehat C/|\widehat{E}|$. Larger values indicate better agreement with the known pathway ordering. Each box summarizes results across the 30 resamples.
    }
    \label{fig:microarray}
\end{figure}

The Arabidopsis thaliana microarray data~\citep{wille2004sparse} contain $n=118$ expression profiles of 39 genes belonging to two isoprenoid pathways. We focus on the 15 genes in the mevalonate (MVA) pathway, for which prior biological knowledge provides the causal ordering among genes. 

Let $\mathcal{P}_0$ denote the set of ordered gene pairs implied by the known pathway. For an estimated DAG with edge set $\widehat{E}$, we consider its influence graph, where $\widehat{R}_{kj}=1$ if there exists a directed path from gene $k$ to gene $j$. For each $(k,j)\in\mathcal{P}_0$, we classify the estimated relation as correct if $\widehat{R}_{kj}=1$, reversed if $\widehat{R}_{jk}=1$, and missing otherwise. Let $\widehat C$ and $\widehat R$ denote the numbers of correct and reversed relations, respectively. We measure agreement using the net influence score $(\widehat C-\widehat R)/|\widehat{E}|$ and the influence score $\widehat C/|\widehat{E}|$. To assess stability given the small sample size ($n = 118$), we generate 30 resamples. Figure~\ref{fig:microarray} shows that our \textit{MIP} procedure \eqref{eq:MIP-reduced} achieves competitive net influence scores alongside \textit{TOPIC} and \textit{NoTears}, with all three methods having interquartile ranges entirely above zero. \textit{MIP} also achieves the highest median influence score, recovering more correctly ordered pathway relations per estimated edge. Overall, \textit{MIP} demonstrates favorable agreement with the known pathway ordering while maintaining a relatively sparse graph. See Appendix~\ref{sec:appendix:microarray} for the known pathway and further details.

\section{Discussion}
\label{sec:discussion}
We develop a convex mixed integer programming (MIP) formulation for causal additive models with both optimization and statistical guarantees: solutions may have a prespecified optimality gap, and DAG recovery is consistent as the number of nodes grows with sample size. We extend the formulation to pairwise interactions and show improved performance over existing methods on synthetic and real data.

Our work opens several avenues for future research. First, our estimator relies on Gaussian errors; adapting it to non-Gaussian errors would broaden the utility of our framework. Second, in the linear causal discovery setting, \citet{Xu2024AnAO} proposed a tailored coordinate descent method to approximately solve their MIP and proved its asymptotic optimality. Applying a similar approach to causal additive models could be of practical interest.


\section*{Funding}
This work was partially supported by funding from NSF under grant DMS-2413074 and NIH under grant RF1-AG090462. 
This work also acknowledges funding from Office of Naval Research (N000142512504). 


\addtolength{\textheight}{-.2in}%

\newpage
\singlespacing

\addtocontents{toc}{\protect\setcounter{tocdepth}{2}}
\begingroup
\renewcommand{\contentsname}{Appendix Contents}
\tableofcontents
\endgroup

\appendix

\section{Discussion of the zero-mean formulation}
\label{sec:appendix:zeromean}
Model \eqref{eq:DAG-model} assumes
$$
X_j=\sum_{k\in \mathrm{pa}(j)} f^\star_{kj}(X_k)+\epsilon_j,
\qquad 
\epsilon_j\sim N(0,\sigma_j^{\star 2}),
\qquad 
\mathbb{E}[f^\star_{kj}(X_k)]=0 .
$$
We show that the zero-mean requirement and the absence of intercepts are without loss of generality.
Start from a general additive structural equation model on the same oracle DAG $\mathcal{G}^\star=(V,E^\star)$,
allowing intercepts and non-zero mean components,
\begin{equation}\label{eq:general_sem_app}
	X_j
	= \mu_j+\sum_{k\in \mathrm{pa}(j)} \tilde f_{kj}(X_k)+\tilde\epsilon_j,
	\qquad j \in \mathrm{V},
\end{equation}
where the $\tilde f_{kj}$ are arbitrary measurable functions and
$\tilde\epsilon_1,\dots,\tilde\epsilon_p$ are mutually independent Gaussians with variances
$\sigma_j^{\star 2}$ but possibly non-zero means.

For each $(k,j)\in E^\star$ let $m_{kj}:=\mathbb{E}[\tilde f_{kj}(X_k)]$, and for each $j$ let
$m_j:=\mathbb{E}[\tilde\epsilon_j]$.  Define
$$
f^\star_{kj}(x):=\tilde f_{kj}(x)-m_{kj}, 
\qquad
\epsilon_j:=\tilde\epsilon_j-m_j,
\qquad
\mu'_j:=\mu_j+\sum_{k\in \mathrm{pa}(j)} m_{kj}+m_j .
$$
Then \eqref{eq:general_sem_app} can be equivalently formulated as
$$
X_j=\mu'_j+\sum_{k\in \mathrm{pa}(j)} f^\star_{kj}(X_k)+\epsilon_j,
$$
with $\mathbb{E}[f^\star_{kj}(X_k)]=0$ and $\epsilon_j\sim N(0,\sigma_j^{\star 2})$ by construction.
Only constants were subtracted, so $\mathrm{pa}(j)$, $\mathcal{G}^\star$, and the induced distribution of $X$
are unchanged. Moreover, in this representation $\mu'_j=\mathbb{E}[X_j]$.  Center the variables by
$X_j^{c}:=X_j-\mu'_j$. By defining the translated functions
$$
g^\star_{kj}(x):=f^\star_{kj}(x+\mu'_k),
$$
we obtain
$$
X_j^{c}=\sum_{k\in \mathrm{pa}(j)} g^\star_{kj}(X_k^{c})+\epsilon_j .
$$
Each $g^\star_{kj}$ is a translation of $f^\star_{kj}$, hence it is non-linear and three-times
differentiable whenever $f^\star_{kj}$ is. Therefore, the centered model still satisfies the standing assumptions of Proposition~\ref{prop:identifiability}. Finally, by relabeling the centered variables $X_j^{c}$ again as $X_j$, and the translated functions $g^\star_{kj}$ again as $f^\star_{kj}$, we recover exactly the normalized form \eqref{eq:DAG-model}. Thus the zero-mean convention is without loss of generality.

Note that when $\mathrm{pa}(j)=\varnothing$, centering removes the intercept entirely, so the equation for $X_j$ has no intercept term. This matches the constraint \eqref{eq:MIP:b2} in our MIP that forces a zero intercept whenever node $j$ has no parents.

\vspace{0.2in}
\section{An alternative formulation of the loss function}
\label{sec:appendix:equivalence}

We provide an alternative formulation of the loss function as a vanilla baseline, which can be solved directly using a generic MIP solver.

We re-parameterize and reformulate the negative log-likelihood function $\ell_n(\theta)$ in \eqref{eq:log-likelihood} as follows. Recall that for each $f_{kj}\in\mathcal{F}_{kj,n}$, the function $f_{kj}(\cdot) = \beta_{0kj} + \sum_{r=1}^{R_n} \beta_{rkj} b_{rkj}(\cdot)$ is a linear combination of $R_n$ basis functions. Without loss of generality, we further assume that the same set of basis functions $\{b_r(\cdot)\}_{r=1}^{R_n}$ is used across all edges $(k, j) \in \mathcal{E}$, so that each $b_{rkj}(\cdot)$ reduces to $b_r(\cdot)$. Let $Z\in\mathbb{R}^{pR_n + p+1}$ be an extended basis random vector:
\begin{equation*}
    Z:=
    \begin{bmatrix}
        X_1, \ldots, X_p,  & 1, & b_1(X_1), \ldots, b_1(X_p), & \ldots, & b_{R_n}(X_1), \ldots, b_{R_n}(X_p)
    \end{bmatrix}^\top.
\end{equation*}
We then define a scaled coefficient matrix $\Gamma\in\mathbb{R}^{p \times (pR_n + p+1)}$ accordingly:
\begin{equation*}
    \Gamma:= \diag(\sigma_1, \ldots, \sigma_p)^{-1} \cdot
    \begin{bmatrix}
        I_{p\times p} & -B^{(0)}_{p\times 1} & -B^{(1)}_{p \times p} & \cdots & -B^{(R_n)}_{p\times p}
    \end{bmatrix},
\end{equation*}
where $I_{p\times p}$ is an identity matrix, the sub-matrix $B_{p\times p}^{(r)}$, ($r = 1, \dots, R_n$) encapsulates the coefficients associated with the sub-vector $[b_r(X_1), \dots, b_r(X_p)]$ in $Z$, and sub-matrix $B^{(0)}_{p\times 1}$ encapsulates the intercept coefficients. Specifically, for $r=1, \dots, R_n$, the $j$-th row and $k$-th column of $B^{(r)}_{p \times p}$ is equal to $\beta_{rkj}$, with $\beta_{rjj} = 0$ for all $j\in \mathrm{V}$. The $j$-th row of $B^{(0)}_{p\times 1}$ is equal to $\sum_{k=1}^p \beta_{0kj}$. 

There exists a bijective mapping between the parameters $\Gamma$ and $\theta\in\Theta$. Specifically, for any given $\Gamma$, the map to $\theta$, $\theta = \mu(\Gamma)$, is defined as $f_{kj}(\cdot) = [ \sum_{r=1}^{R_n} \Gamma_{j, pr + 1+k} \cdot (\mathbb{E}[b_r(X_k)]-b_r(\cdot)) ] \cdot \Gamma_{jj}^{-1}$ and $\sigma_j = \Gamma_{jj}^{-1}$.
Particularly, $\sum_{k=1}^p f_{kj}(\cdot)  = -[ \Gamma_{j,p+1} + \sum_{k=1}^p \sum_{r=1}^{R_n} \Gamma_{j,pr + 1+k} \cdot b_r(\cdot) ] \cdot \Gamma_{jj}^{-1}$.
Under this mapping, we establish the following equivalence.

\vspace{0.1in}
\begin{proposition}\label{prop:log-lik-equiv}
    Let $\{ Z^{(i)} \}_{i=1}^n$ be $n$ independent and identically distributed samples from the distribution of $Z$. Let $\widehat{\Sigma}_n := n^{-1} \sum_{i=1}^n Z^{(i)} Z^{(i) \top}$ be the Gram matrix. Then, the negative log-likelihood function can be equivalently formulated as:
    \vspace{-0.1in}
    \begin{equation} \label{eq:log-lik-equiv}
    \ell_n(\theta) = \ell_n(\mu(\Gamma))  := \sum_{j=1}^p -2 \log \Gamma_{jj} + \tr(\Gamma^\top \Gamma \widehat{\Sigma}_n).
    \vspace{-0.1in}
    \end{equation}
\end{proposition}
\begin{proof}
    Recall the definition 
\begin{equation*}
    Z:=
    \begin{bmatrix}
        X_1, \cdots, X_p,  &1, & b_1(X_1), \cdots, b_1(X_p), & \cdots, & b_{R_n}(X_1), \cdots, b_{R_n}(X_p)
    \end{bmatrix} \in \mathbb{R}^{p{R_n} + p+1}.
\end{equation*}
    Let $Z^{(i)} \in \mathbb{R}^{pR_n + p+1}$ denote the $i$-th sample of the random vector $Z$.
    Moreover, we defined the mapping $\theta = \mu (\Gamma)$: for any $\Gamma \in \mathbb{R}^{p \times (pR_n + p +1)}$, $\sum_{k=1}^p f_{kj}(\cdot) = -[\Gamma_{j,p + 1} + \sum_{k=1}^p\sum_{r=1}^{R_n} \Gamma_{j, pr + 1+k} \cdot b_r(\cdot) ] \cdot \Gamma_{jj}^{-1}$ and $\sigma_j = \Gamma_{jj}^{-1}$. Denote the $j$-th row of the matrix $\Gamma$ as $\gamma_j \in \mathbb{R}^{pR_n + p+1}$. Now replacing $\theta$ with $\mu (\Gamma)$ in $\ell_n(\theta)$, we obtain
    \begin{align*}
    \begin{split}
        \ell_n(\mu (\Gamma))
    &= \sum_{j=1}^p \log \Gamma_{jj}^{-2} + \sum_{j=1}^p \frac{ \|X_j + \Gamma_{jj}^{-1} [ \Gamma_{j,p+1} + \sum_{k=1}^p \sum_{r=1}^{R_n} \Gamma_{j,pr + 1+k} b_r(\cdot) ] \|_n^2 }{\Gamma_{jj}^{-2}}  \\
    &= \sum_{j=1}^p -2 \log \Gamma_{jj} + \sum_{j=1}^p  \left\| \Gamma_{jj} X_j + \Gamma_{j,p+1} +\sum_{k =1}^p \sum_{r = 1}^{R_n} \Gamma_{j, pr +1 + k} b_r(X_k) \right\|_n^2 \\
    &= \sum_{j=1}^p -2 \log \Gamma_{jj} + \frac{1}{n} \sum_{j=1}^p \sum_{i=1}^n (\gamma_j^\top Z^{(i)})^2 
    = \sum_{j=1}^p -2 \log \Gamma_{jj} + \tr(\Gamma^\top \Gamma \widehat{\Sigma}_n),
    \end{split}
    \end{align*}
    where we defined $\widehat{\Sigma}_n = n^{-1} \sum_{i=1}^n Z^{(i)} Z^{(i) \top }$.
\end{proof}

The equivalent formulation \eqref{eq:log-lik-equiv} is beneficial in two aspects. First, we observe that it is a convex function of $\Gamma$. Second, given the Gram matrix $\widehat{\Sigma}_n$, the trace term $\tr(\Gamma^\top \Gamma \widehat{\Sigma}_n)$ does not require iteration over all samples. It has computational complexity on the order of $R_n^2 p^3$, which is significantly smaller---when the sample size $n$ is sufficiently large---than $nR_np^2$, the complexity of evaluating $\sum_{j=1}^p \| X_j - \sum_{k\neq j} f_{kj}(X_k) \|_n^2 / \sigma_j^2$. 

The corresponding MIP is then given by:
\begin{subequations}  \label{eq:app:mip-vanilla}
\begin{align} 
    \min_{ \substack{\Gamma \in \mathbb{R}^{p \times (pR_n + p + 1)} \\ 
    g_{kj} \in \{0,1\},\ \forall {(k,j)\in \mathcal{E}} \\
    \psi \in [1,p]^p } } \quad & \sum_{j=1}^p -2 \log \Gamma_{jj} + \tr(\Gamma^\top \Gamma \widehat{\Sigma}_n) + \gamma_n \left\|\bar \Gamma \right\|_F^2 + \lambda_n^2 \cdot \sum_{(k,j) \in \mathcal{E}} g_{kj}   \\
    \mathrm{s.t.} \quad & 1 -  p + p g_{kj} \leq \psi _j - \psi _k, \quad \forall (k,j) \in \mathcal{E},  \\
    & -M g_{kj} \leq \Gamma_{j, pr+1 +k} \leq M g_{kj}, \quad  \forall (k,j) \in \mathcal{E}, \ r\in \{1, \dots, R_n\}, \label{eq:app:mip-vanilla:c}  \\ 
    & - M \sum_{k=1}^p g_{kj} \leq \Gamma_{j,p+1} \leq M \sum_{k=1}^p g_{kj}, \quad \forall j\in \mathrm{V}, \label{eq:app:mip-vanilla:d} \\
    & 0 < \Gamma_{jj} \leq M, \quad \forall j \in [p].  \label{eq:app:mip-vanilla:e}
\vspace{-0.1in}
\end{align}
\end{subequations}
The matrix $\bar{\Gamma}$ represents $\Gamma$ removing the first $p+1$ columns:
$$
\bar{\Gamma} = \diag(\sigma_1, \ldots, \sigma_p)^{-1} \cdot
\begin{bmatrix}
       -B^{(1)}_{p \times p} & \cdots & -B^{(R_n)}_{p\times p}
\end{bmatrix},
$$
resulting in a ridge penalty 
$$
\|\bar{\Gamma}\|_F^2 = \sum_{j=1}^p  \sigma_j^{-2}  \sum_{k=1}^p \sum_{r=1}^{R_n} |\beta_{rkj}|^2    .
$$
The constant $M > 0$ is a pre-specified upper bound that exceeds all possible values of $|\Gamma_{ij}|$ for $i \in [p]$ and $j \in [pR_n + p + 1]$. It is used to linearize non-linear terms, resulting in the \emph{big-$M$ constraints} given in \eqref{eq:app:mip-vanilla:c}--\eqref{eq:app:mip-vanilla:e}. A smaller value of $M$ can improve computational efficiency. In practice, one may first solve the optimization problem without the big-$M$ constraints to obtain a preliminary estimate $\widehat{\Gamma}$, and then set $M = 2 \max_{i\in [p],\ j \in[pR_n +p+1] } |\widehat{\Gamma}_{ij}|$  \citep{park2017bayesian, kucukyavuz2023consistent, xu2025integer}.


\vspace{0.2in}
\section{Algorithmic and computational details} 
\label{sec:appendix:computation}

\subsection{Proof of Theorem 2}
\label{sec:appendix:convexity}

For simplicity, we focus on the setting without an intercept. That is, $\beta_j = \beta_j^{-0}$. The case with an intercept follows by a straightforward modification. Recall the original optimization problem in \eqref{eq:MIP}:
\begin{align} \label{eq:app:mip-orig}
\begin{split}
    \min_{ \substack{\beta_j \in \mathbb{R}^{p R_n }, \ \forall j\in\mathrm{V} \\ 
    g_{kj} \in \{0,1\},\ \forall {(k,j)\in \mathcal{E}} \\
    \psi \in [1,p]^p } } \quad &  \underbrace{ \sum_{j=1}^p \log \left( \left\| X_j - Z_j \beta_j  \right\|_n^2 + \gamma_n \left\| \beta_j \right\|_2^2 \right) }_{F(\beta; g)} + \lambda_n^2 \sum_{(k,j)\in \mathcal{E}}g_{kj}   \\
    \mathrm{s.t.} \quad & 1 -  p + p g_{kj} \leq \psi _j - \psi _k, \quad \forall (k,j) \in \mathcal{E},  \\
     & \left( 1 - g_{kj} \right) \beta_{rkj} = 0, \quad  \forall (k,j) \in \mathcal{E}, \ r\in \{1, \dots, R_n\}, \\
     & \beta_{rjj} = 0, \quad \forall j\in\mathrm{V}, \ r\in\{1, \dots, R_n\}.
\end{split}
\end{align}
Let $g_j := [g_{1j}, \ldots, g_{pj}] \in \{0,1\}^p$. We use $Z_{j\mid g_j} \in\mathbb{R}^{n \times (R_n|g_j|) }$ and $\beta_{j \mid g_j} \in \mathbb{R}^{R_n|g_j|}$ to denote, respectively, the sample matrix of $Z_j$ and coefficient vector $\beta_j$ restricted to the parents selected by $g_j$.
Also, let $Z_{j\mid k}\in\mathbb{R}^{n \times R_n}$ be the sample matrix of $Z_j$ restricted to the parent $k$. In practice, we choose $Z_{j \mid k} = [b_1(X_k), b_2(X_k), \dots, b_{R_n}(X_k)]$. In the cutting-plane procedure, the variables $g_{kj}$ are sometimes relaxed to take fractional values. In this case, $g_j \in [0,1]^p$ becomes a fractional $p$-vector, and 
$$
Z_{j \mid g_j} = \left[ \sqrt{g_{1j}}Z_{j \mid 1},\ \cdots, \ \sqrt{g_{pj}}Z_{j \mid p} \right] \in \mathbb{R}^{n \times (R_np)}, \qquad
\beta_{j \mid g_j} \in \mathbb{R}^{R_n p}.
$$
We adopt the corresponding notation throughout whenever $g_j$ is fractional.

Recall that the loss function is denoted by $F(\beta,g)$. For fixed $g$, let $f(g)=\min_\beta F(\beta,g)$, where $\beta$ satisfies the support constraints imposed by $g$. The following lemma provides explicit expressions for $f(g)$.

\vspace{0.1in}
\begin{lemma}
    The function $f(g)$ is evaluated as
    $$
    f(g) = \sum_{j=1}^p \log \left( \frac{X_j^\top}{\sqrt{n}} \left( I_n + \gamma_n^{-1} \frac{ Z_{j \mid g_j} Z^\top_{j \mid g_j} }{n} \right)^{-1} \frac{X_j}{\sqrt{n}} \right) 
    = \sum_{j=1}^p \log \left( X_j^\top \left(n I_n + \gamma_n^{-1} Z_{j \mid g_j } Z_{j \mid g_j}^\top \right)^{-1} X_j \right).
    $$
\end{lemma}
\begin{proof}
    The result follows immediately from \citet[Lemma 1]{bertsimas2019sparse}.
\end{proof}

\vspace{0.2in}
However, evaluating $f(g)$ using the expression above requires inverting $p$ matrices in $\mathbb{R}^{n\times n}$, which can be computationally demanding. To reduce this cost, we further define
$$
\alpha_j(g_j) := \frac{X_j}{\sqrt{n}} - \frac{Z_{j\mid g_j}}{\sqrt{n}} \left( \frac{Z_{j \mid g_j}^\top Z_{j \mid g_j}}{n} + \gamma_n I \right)^{-1} \frac{Z_{j \mid g_j}^\top X_j }{n}.
$$
The following lemma provides computationally more efficient expressions for $f(g)$ and its gradient. The matrix inversion has dimension $R_n|g_j|$ for a binary
restricted design and $pR_n$ for the full fractional design, independent of the sample size $n$.
\begin{lemma} \label{lem:CP:f-and-df} 
The function $f(g)$ can also be evaluated as
$$
f(g) = \sum_{j=1}^p \log \left( \frac{X_j^\top}{\sqrt{n}} \alpha_j(g_j) \right).
$$
Furthermore, its gradient is given by 
$$
\nabla_{g_{kj}} f(g) = -\frac{ \gamma_n^{-1} \left\| Z^\top_{j \mid k} \alpha_j(g_j) \right\|_2^2 } { \sqrt{n}  X_j^\top  \alpha_j(g_j) }.
$$
\end{lemma}
\begin{proof}
    Denote 
    $$
    f_j(\beta_{j\mid g_j}; g_j) := \left\| X_j - Z_{j \mid g_j} \beta_{j \mid g_j} \right\|_n^2 + \gamma_n \left\| \beta_{j\mid g_j} \right\|_2^2, \qquad
    \beta_{j \mid g_j}^\dagger := \argmin_{\beta_{j \mid g_j}} f_j(\beta_{j\mid g_j}; g_j ). 
    $$
    By \citet[Lemma~1 and Lemma~2]{bertsimas2019sparse}, for any $j\in[p]$ and a fixed $g_j\in\{0,1\}^p$, we have
    $$
    f_j(\beta^\dagger_{j \mid g_{j}} ; g_j)
    = \frac{X_j^\top}{\sqrt{n}} \left( I_n + \gamma_n^{-1} \frac{Z_{j\mid g_j} Z_{j \mid g_j}^\top}{n} \right)^{-1} \frac{X_j}{\sqrt{n}}
    = \frac{X_j^\top}{\sqrt{n}} \alpha_j(g_j).
    $$
    Furthermore, 
    $$
    \nabla_{g_{kj}} f_j(\beta^\dagger_{j \mid g_{j}} ; g_j) 
    = - \frac{\gamma_n^{-1}}{n} \left\| Z_{j \mid k}^\top \alpha_j(g_j) \right\|_2^2.
    $$
    Consequently, 
    $$
    f(g) = \min_\beta F(\beta, g) = \sum_{j=1}^p \left[ \log \min_\beta f_j(\beta_{j \mid g_j}; g_j) \right] = \sum_{j=1}^p \log \left( \frac{X_j^\top}{\sqrt{n}} \alpha_j(g_j) \right),
    $$
    and
    $$
    \nabla_{g_{kj}} f(g) =  \nabla_{g_{kj}} \log \left( \frac{X_j^\top}{\sqrt{n}} \alpha_j(g_j) \right) =  \frac{ -\gamma_n^{-1} \left\| Z_{j\mid k}^\top \alpha_j(g_j) \right\|_2^2 /n }{ X_j^\top \alpha_j(g_j) / \sqrt{n} }  .
    $$
\end{proof}

\vspace{0.2in}
The original problem \eqref{eq:app:mip-orig} has been equivalently translated into the convex mixed-integer programming~\eqref{eq:MIP-reduced}.
Now we restate \textbf{Theorem~\ref{thm:convexity}} as the following proposition.
\begin{proposition} \label{lem:convexity}
    The function $f(g)$ is convex in $g$ when the binary variables are relaxed into $[0,1]$. That is, for any feasible $g$ and $g^\circ$,
    $$
    f(g)\geq f(g^\circ) + \nabla f(g^\circ)^\top (g - g^\circ).
    $$
\end{proposition}
\begin{proof}
    Note that we can rewrite $f(g)= \sum_{j=1}^p \log \left( X_j^\top A^{-1}(g_j) X_j \right)$, where $A(g_j) \succ 0$ is affine in $g_j$:
    $$
    A(g):= nI_n + \gamma_n^{-1} \sum_{k\in g_j} g_{kj} Z_{j\mid k} Z^\top_{j\mid k} .
    $$
    Now appealing to the inequality $-\log x \geq - \log x_0 - (x - x_0) / x_0 $ for any $x,x_0 >0$, we obtain that for any $\nu_j >0$,
    $$
    f(g) \geq \sum_{j=1}^p \left[ 1 + \log \nu_j - \frac{\nu_j}{ X_j^\top A^{-1}(g_j) X_j } \right] .
    $$
    Define the third term inside as $h_j(g_j):= -\nu_j / X_j^\top A^{-1}(g_j) X_j $, then
    $$
    h_j(g_j) = - \nu_j \left( \min_{z: X^\top_j z = 1} z^\top A(g_j) z \right).
    $$
    Note that $A(g_j)$ is affine in $g_j$, and so is $z^\top A(g_j) z$ for each fixed $z$. Therefore, its pointwise minimum over $z$ is concave in $g_j$. Taking the negative then implies that $h_j(g_j)$ is convex in $g_j$. Appealing to the first order condition of convexity, we have
    $$
    f(g) 
    \geq \sum_{j=1}^p \left[ 1 + \log \nu_j + h_j(g_j^\circ) + \nabla h_j(g_j^\circ)^\top (g_j - g_j^\circ) \right].
    $$
    Taking $\nu_j = X_j^\top A(g_j^\circ)^{-1} X_j$, we obtain 
    $$
    h_j(g_j^\circ) = -1, \ \forall j\in[p], \qquad
    \sum_{j=1}^p\log \nu_j = f(g^\circ), \qquad
    \left. \nabla_{g_j} h_j(g_j) \right|_{g = g^\circ} = \left. \nabla_{g_j} f(g) \right|_{g=g^\circ}.
    $$
    We then conclude that
    $$
    f(g) \geq f(g^\circ) + \nabla f(g^\circ)^\top (g - g^\circ).
    $$
\end{proof}

\vspace{0.1in}
Finally, when an intercept is included, we can center the data by setting $\check{X}_j = H X_j$ and $\check Z_{j|k}=HZ_{j|k}$, where $H = I_n - \mathbf{1}\mathbf{1}^\top/n$, and apply the same procedure as above directly to the objective function
$$
\check{F}(\beta;g) =\sum_{j=1}^p \log \left( \|\check{X}_j - \check{Z}_j \beta_j^{-0}\|_n^2 + \gamma_n \|\beta_j^{-0} \|_2^2 \right).
$$

\vspace{0.2in}
\subsection{Proof of Theorem 3}
\label{sec:appendix:cutting-plane}

Algorithm~\ref{alg:cutting-plane} presents the cutting-plane algorithm for solving this optimization problem. It iteratively solves the relaxed optimization problem~\eqref{eq:MIP-CP} while progressively tightening the lower approximation. See Section~\ref{sec:computation:cutting-plane} for the motivation behind this construction. In the first phase of Algorithm~\ref{alg:cutting-plane}, we relax the binary variables to $[0,1]$ and generate cutting planes at fractional solutions to strengthen the approximation of $f(g)$. 
In the second phase, we restore the integrality constraints and continue the outer approximation procedure until the prescribed optimality gap tolerance is reached.

Finally, we restate \textbf{Theorem~\ref{thm:optimality}} and provide its proof.
\begin{proposition}
Algorithm~\ref{alg:cutting-plane} terminates after finitely many iterations. Let $T$ denote the termination iteration. Then the objective value of the returned solution $g^{(T)}$ is within $\tau^{\rm early}$ of the global optimum of problem~\eqref{eq:app:mip-orig}.
\end{proposition}
\begin{proof}
    Since $f$ is convex over the relaxation by Proposition~\ref{lem:convexity}, \citet[Theorem~3]{duran1986outer} implies that Algorithm~\ref{alg:cutting-plane} terminates after finitely many cutting planes and returns a solution whose objective value is within $\tau^{\rm early}$ of the global optimum of problem~\eqref{eq:MIP-reduced}. By the equivalence between problems~\eqref{eq:MIP-reduced} and~\eqref{eq:app:mip-orig}, the same guarantee carries over to~\eqref{eq:app:mip-orig}.
\end{proof}

\vspace{0.2in}
\subsection{More on the cutting-plane method} \label{sec:appendix:more-cutting-plane}

In this section, we present additional experimental results on the performance of the cutting-plane method.

First, we compare the computational performance of estimator~\eqref{eq:MIP} with and without the cutting-plane method in Table~\ref{tab:oa-vanilla}.
A straightforward approach without cutting planes is to solve the formulation based on sufficient statistics in Appendix~\ref{sec:appendix:equivalence} directly using \texttt{Gurobi}. This formulation can be viewed as the nonlinear counterpart of that proposed by \citet{xu2025integer}. We refer to this approach as the vanilla procedure, which neither optimizes out the continuous variables nor exploits the additional problem structure.

When not using early stopping, both the cutting-plane and vanilla procedures attain the same objective value, confirming that the cutting-plane method solves the same optimization problem, although with less computation time. As the graph size grows, the vanilla procedure often has to be terminated before convergence, resulting in a higher objective value, whereas the cutting-plane method reaches optimality much faster. Consequently, the cutting-plane method also yields more accurate estimates. These advantages become more pronounced as the graph size increases.
\vspace{0.05in}
\begin{table}[!ht]
\centering
{\small
\begin{tabular}{lcccc} 
\toprule
Dataset.$p$.$s^\star$         & ObjVal diff & SHD diff & CP Time & Vanilla Time \\ \midrule
\textit{dsep}.6.6        & $<10^{-4}$                            & 0                          & 0.123      & 0.213           \\
\textit{asia}.8.8        & $<10^{-4}$                            & 0                          & 0.223      & 0.529           \\
\textit{bowling}.9.11    & $<10^{-4}$                            & 0                          & 0.333      & 7.51           \\
\textit{inssmall}.15.25  & 0.003 (stop early)                   & 2.5                        & 3.081      & 301                \\
\textit{insurance}.27.52 & 0.008 (stop early)                   & 5.3                        & 46.622     & 302                \\ 
\bottomrule
\end{tabular}
\caption{\small Comparison of the cutting-plane procedure with the vanilla procedure. ``ObjVal diff'' denotes the difference in objective values, while ``SHD diff'' denotes the difference in structural Hamming distance to the true DAG. Both are computed as the value from the vanilla procedure minus that from the cutting-plane procedure. The ``CP Time'' and Vanilla Time'' columns report the runtime, in seconds, of the cutting-plane and vanilla procedures, respectively. For \textit{inssmall} and \textit{insurance}, the vanilla procedure is terminated after 300 seconds. The implementation details are provided in Appendix~\ref{sec:appendix:setup-oa-vanilla}.
}
\label{tab:oa-vanilla}
}
\end{table}

Second, we examine the benefits of early stopping. Recall that, to reduce computational cost, the procedure may be terminated once $U^{(t)}-L^{(t)}\leq \tau^{\rm early}$,
where $\tau^{\rm early}>0$ is a user-specified tolerance, and return the corresponding feasible solution $\widetilde g^{(t)}$. 
Although the resulting solution need not be globally optimal, its objective value is guaranteed to lie within $\tau^{\rm early}$ of the global optimum. This guarantee is generally unavailable for many nonconvex optimization methods. 
As shown in Figure~\ref{fig:early_stopping}, increasing $\tau^{\rm early}$ from zero leads to only a modest deterioration in estimation performance, while substantially reducing the running time. 

\begin{figure}[ht]
    \centering
    \includegraphics[width=0.9\textwidth]{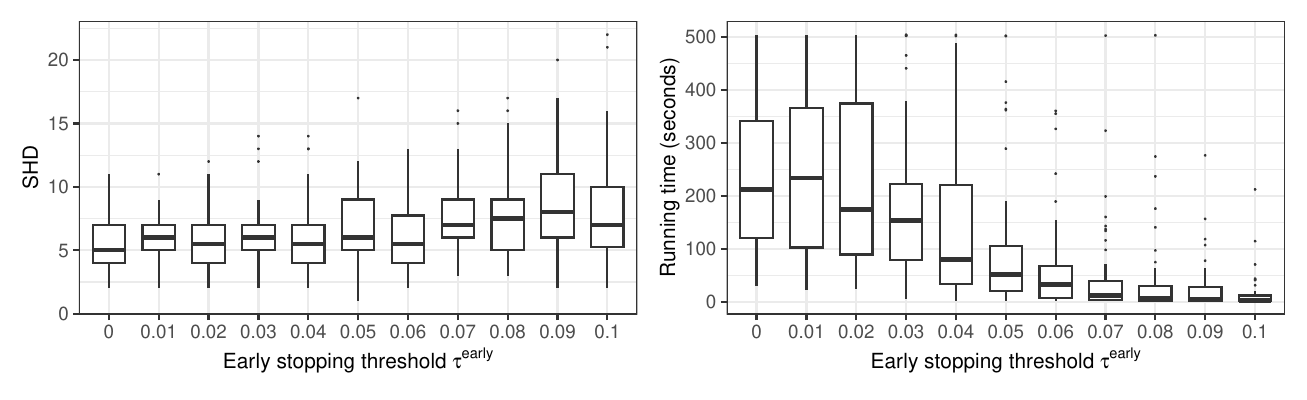}
    \caption{\small The performance of the estimator obtained from solving \eqref{eq:MIP} across a range of $\tau^{\rm early}$ values. Left panel: structural Hamming distance, which quantifies the number of incorrectly recovered and missed edges, with smaller values indicating better accuracy. Right panel: total running time in seconds. The implementation details are provided in Appendix~\ref{sec:appendix:setup-early}. 
    }
    \label{fig:early_stopping}
\end{figure}

\vspace{0.2in}
\subsection{Super-structure}\label{sec:appendix:superstruc} 
For the DAG $\mathcal{G}^\star = (\mathrm{V}, E^\star)$, a super-structure $E^\circ \subseteq \mathcal{E}$ is an edge set containing the true edge set $E^\star$: $E^\star \subseteq E^\circ$. The super-structure $E^\circ$ may permit bidirectional edges and cycles, though self-loops are excluded. The idea of imposing a super-structure has been adopted by various related methods \citep{chickering2002optimal, buhlmann2014, solus2021consistency, kucukyavuz2023consistent, xu2025integer}. We set $g_{kj} = 0$ and $\beta_{rkj} = 0$ for all $(k,j) \in \mathcal{E} \setminus E^\circ$ and $r \in\{1, \dots, R_n\}$. In other words, the search space of decision variables is restricted to edges within $E^\circ$. Consequently, when $E^\circ$ is sparse, the
search space is greatly reduced, which in turn can substantially accelerate the optimization algorithm.

Enforcing a super-structure offers multiple advantages. First, in many applications, prior information about the presence or absence of specific edges is available, for example, from domain expertise or some widely accepted knowledge. Such information can be naturally incorporated into the super-structure. Second, as shown above, edges outside the super-structure $E^\circ$ are excluded from consideration. Consequently, when $E^\circ$ is sparse, the search space is greatly reduced, which in turn can substantially accelerate the optimization algorithm. Lastly, the existence of a sparse super-structure greatly facilitates the theoretical analysis of our proposed estimator \eqref{eq:MIP}. 
As we will see in Section \ref{sec:theory}, we assume access to a super-structure $E^\circ$ within which the number of parents for each node is uniformly bounded by a constant. 

In addition to leveraging prior knowledge when available, we can reduce the search space through a preliminary neighborhood selection step, following the general strategy of \citet{buhlmann2014} and related work~\citep{meier2009high, voorman2014graph}. Specifically, as in the preliminary neighborhood selection step of CAM, for each node $X_j$, we fit an additive model of $X_j$ on all remaining variables $X_{-j}=\{X_k:k\neq j\}$ and retain only variables exhibiting sufficiently strong associations with $X_j$ as candidate neighbors. The resulting candidate sets define a data-driven superstructure that restricts the edges considered in the subsequent MIP. Such neighborhood selection can be implemented using, for example, the group Lasso, GAMBoost, or sparse additive models~\citep{buhlmann2007boosting, ravikumar2009sparse, haris2022generalized}.

\subsection{Partial order and stable edges} \label{sec:appendix:partial-stable}
The topological ordering, or sometimes referred to as the permutation of nodes, is a key concept in DAG estimation. 
If it is known a priori that node $k$ is an ancestor of node $j$, then this relation can be encoded by imposing the constraint $\psi_k < \psi_j$ in our MIP formulation \eqref{eq:MIP}. We thus define the preliminary \textit{partial order set}, as a collection of edges in $E^\circ$ whose ordering is predetermined: $E^p := \{(k,j) \in E^\circ: \psi_k < \psi_j \}$. 
The constraint $\psi_k < \psi_j$ excludes any path from node $j$ to $k$. However, the presence of edge from $k$ to $j$ remains undetermined: either $g_{kj} = 0$ or $g_{kj} = 1$. If the presence of an edge from $k$ to $j$ is further known from prior knowledge, then the constraint $g_{kj} = 1$ can be directly imposed. Accordingly, we define the preliminary \textit{stable set} as the collection of edges in $E^p$ that are predetermined to exist: $E^s := \{ (k,j) \in E^p: g_{kj} = 1 \}$; directly imposing this constraint eliminates the corresponding integer decision variables.

Enforcing the stable set and partial order set can substantially reduce the search space, thereby improving optimization efficiency. When prior information is unavailable or insufficient, a preliminary partial order set and stable set can be constructed using bootstrap procedures. We apply a baseline method to $B$ bootstrap samples, and an edge is included in the partial order set or stable set if its selection proportion exceeds a prescribed threshold. We denote the thresholds for the partial order set and stable set by $\tau^p$ and $\tau^s$, respectively. To ensure the nesting property $E^s \subseteq E^p$, the thresholds are chosen to satisfy $\tau^p \leq \tau^s$. Finally, note that the selected partial order set and stable set may contain cycles. To prevent this, edges are added sequentially in decreasing order of their selection proportion; any edge that would introduce a cycle with previously added edges is omitted. In practice, a good default can be $\tau^p = 0.95$ and $\tau^s = 1$.

\vspace{0.2in}
\section{Discussion and implications of assumptions}

\subsection{Valid basis functions and Sobolev ellipsoid}
\label{sec:appendix:sobolev}

Under Assumption \ref{assump:sobolev}, for each $(k, j) \in \mathcal{E}$, the function classes $\mathcal{F}_{kj}$ can be rewritten as
\begin{equation}\label{eq:F-kj}
    \mathcal{F}_{kj} := \left\{
\begin{aligned}
f_{kj} 
:\ & f_{kj}(\cdot) = \beta_{0kj} + \sum_{r=1}^\infty \beta_{rkj} b_{rkj}(\cdot), 
\quad \sum_{r=1}^\infty |\beta_{rkj}| \cdot r^\eta \leq C, 
\quad \eta > 1,  \\
& \beta_{0kj} = - \mathbb{E}\left[ \sum_{r=1}^\infty \beta_{rkj} b_{rkj}(X_k) \right]
\end{aligned}
\right\}.
\end{equation}
Each basis function $b_{rkj}(\cdot)$ for $r \geq 1$ satisfies the following: (i) it has infinite support; (ii) it is uniformly bounded with $|b_{rkj}(\cdot)| \leq 1$; (iii) it is and centered with $\mathbb{E}[b_{rkj}(X_k)] = 0$; (iv) it is orthogonal with $\mathbb{E}[b_{rkj}(X_k) b_{r'kj}(X_k)] = \upsilon  \mathds{1}(r = r')$ for some $\upsilon \in (0, 1]$; (v) it is three times differentiable.
One example construction of such basis functions, using the cosine system ~\citep{efromovich1999nonparametric}, is given by
$$
b_{rkj}(\cdot) = \cos\left( \pi r \Phi_k (\cdot) \right), \quad r\in\{1, \dots, R_n\},
$$
where $\Phi_k(\cdot)$ is the cumulative distribution function of $X_k$. Let $U = \Phi_k(X_k) \sim \mathsf{Unif}(0,1)$, then we see that
$$
\mathbb{E}\left[ b_{rkj}(X_k) \right]
= \mathbb{E} \left[ \cos(\pi r \Phi_k(X_k)) \right]
=\mathbb{E}\left[\cos(\pi rU) \right]=0, 
$$
\begin{equation*}
    \mathbb{E}\left[b^2_{rkj}(X_k) \right] 
= \mathbb{E}\left[ \cos^2(\pi r \Phi_k(X_k)) \right] 
= \mathbb{E}\left[ \cos^2(\pi r U) \right] 
= 1/2,
\end{equation*}
and for $r \neq r'$,
\begin{equation*}
    \mathbb{E}\left[b_{rkj}(X_k) \cdot b_{r'kj}(X_k) \right] 
= \mathbb{E}\left[ \cos(\pi r \Phi_k(X_k)) \cdot \cos(\pi r' \Phi_k(X_k)) \right] 
= \mathbb{E}\left[ \cos(\pi r U) \cdot \cos(\pi r' U) \right] 
= 0.
\end{equation*}

\vspace{0.1in}
The structure on coefficients, $\sum_{r=1}^\infty |\beta_{rkj}| \cdot r^\eta \leq C$ or equivalently $\sum_{r=0}^\infty |\beta_{rkj}| \cdot r^\eta \leq C$, corresponds to a Sobolev ellipsoid. This structure concentrates most of the ``mass'' in the early coefficients, resulting in a fast decay. The complexity of each function class $\mathcal{F}_{kj}$, as a result, can be effectively controlled.

\vspace{0.2in}
\begin{lemma} \label{lem:bracketing}
    Under Assumption \ref{assump:sobolev}, the bracketing entropy of $\mathcal{F}_{kj}$ in \eqref{eq:F-kj} is bounded by:
    $$
    \log N_{[]}(\epsilon, \mathcal{F}_{kj}, \|\cdot\|_\infty) 
    \lesssim 
    \left( \frac{1}{\epsilon} \right)^{1/\eta} \log (1/\epsilon).
    $$
\end{lemma}

\begin{proof}
    To simplify notation, we drop the subscripts $k$ and $j$, and write $\mathcal{F}$ for the function class $\mathcal{F}_{kj}$. By Assumption~\ref*{assump:sobolev} and the centering condition, its coefficient sequences belong to the ellipsoid set $\mathcal{B} = \{\beta: |\beta_0|+ \sum_{r=1}^\infty |\beta_{r}| \cdot r^\eta \leq C\}$ for a sufficiently large constant $C>0$, with $\eta > 1$. The idea is to represent the complexity of $\mathcal{F}$ via a mapping to the ellipsoid set $\mathcal{B}$.
    
    First, we find the covering number of the set $\mathcal{B}$. We define a truncated version $\widetilde{\mathcal{B}}= \{ \beta \in \mathcal{B}: \beta_{r} = 0 \text{ for all $r > t$} \}$, where $t$ is the smallest integer such that $t^{-\eta} \leq \epsilon$. We note that $\widetilde{\mathcal{B}}$ is contained in a $(t+1)$-dimensional $\ell_1$-ball with radius $C$, so
    $$
    N(\epsilon, \widetilde{\mathcal{B}}, \| \cdot\|_1 )
    \leq \left(1 + \frac{2C}{\epsilon} \right)^{t+1}.
    $$
    By selecting $t = \lceil (1/\epsilon)^{1/\eta} \rceil$, we have
    $$
    \log N(\epsilon, \widetilde{\mathcal{B}}, \| \cdot\|_1 ) \lesssim \left( \frac{1}{\epsilon} \right)^{1/\eta} \log (1/\epsilon).
    $$
    Now let $\{\beta^1, \dots, \beta^N\}$ be an $\epsilon$-cover for $\widetilde{\mathcal{B}}$ w.r.t. $\|\cdot\|_1$ norm. Then for any $\beta \in \mathcal{B}$, we take the $\beta^i$ such that $\sum_{j=0}^t |\beta_j - \beta_j^i| \leq \epsilon$, and then 
    $$
    \| \beta - \beta^i \|_1 = \sum_{r=0}^t |\beta_r - \beta_r^i | + \sum_{r=t+1}^\infty |\beta_r| \leq \epsilon + t^{-\eta} \sum_{r = t+1}^\infty |\beta_r| \cdot r^{\eta} \leq (C + 1) \epsilon.
    $$
    This implies that $\{\beta^1, \dots, \beta^N\}$ is a $(C+1) \epsilon$-cover of $\mathcal{B}$. Therefore,
    $$
  \log N((C+1)\epsilon,\mathcal B,\|\cdot\|_1)
\leq \log N(\epsilon,\widetilde{\mathcal B},\|\cdot\|_1)
\lesssim (1/\epsilon)^{1/\eta}\log(1/\epsilon),
    $$
    Rescaling $\epsilon$ gives the claimed bound.
    Next, we find the covering number for the function class $\mathcal{F}$. Note that for any $f\in \mathcal{F}$, there exists some $\beta \in \mathcal{B}$ we use to construct $f$: $f = \beta_0 + \sum_{r=1}^\infty \beta_r b_r$. Let $\{\beta^1, \dots, \beta^N\}$ be an $\epsilon$-cover of $\mathcal{B}$ w.r.t. $\|\cdot\|_1$ norm. Then there must exist $i \in \{1, \dots, N \}$ such that $\| \beta^i - \beta \|_1 \leq \epsilon$. Therefore,
    $$
    \left\|f - \beta_0^{i} - \sum_{r=1}^\infty \beta_r^i b_r  \right\|_{\infty}
    \leq \left| \beta_0 - \beta_0^{i} \right| + \left\| \sum_{r=1}^\infty (\beta_r - \beta_r^i) b_r \right\|_{\infty} 
    \leq  \sum_{r=0}^\infty \left| \beta_r - \beta_r^i \right| \leq \epsilon .
    $$
    This implies that $\{ \beta_0^i +\sum_{r=1}^\infty \beta_r^i b_r \}_{i=1}^N$ is an $\epsilon$-cover of $\mathcal{F}$ w.r.t $\|\cdot\|_{\infty}$ norm. Hence,
    $$
    \log N(\epsilon, \mathcal{F}, \|\cdot\|_{\infty})  \lesssim \left( \frac{1}{\epsilon} \right)^{1/\eta} \log (1/\epsilon).
    $$
    Finally, we arrive at the bracketing entropy of $\mathcal{F}$:
    $$
    \log N_{[]}(\epsilon, \mathcal{F}, \|\cdot\|_\infty) \leq 
    \log N(\epsilon/2, \mathcal{F}, \|\cdot\|_\infty) \lesssim 
    \left( \frac{1}{\epsilon} \right)^{1/\eta} \log (1/\epsilon).
    $$
\end{proof}

Recall the definition of truncated function class $\mathcal{F}_{kj,n} = \{f_{kj} \in \mathcal{F}_{kj}: f_{kj,n} = \beta_{0kj} + \sum_{r=1}^{R_n} \beta_{rkj} b_{rkj}\}$. For any $f\in \mathcal{F}_{kj}$ where $f_{kj} = \beta_{0kj} +\sum_{r=1}^\infty \beta_{rkj} b_{rkj}$, the residual after projecting it onto $\mathcal{F}_{kj,n}$ is given by $D_{kj} = f_{kj} - f_{kj,n} = \sum_{r=R_n + 1}^\infty \beta_{rkj} (b_{rkj} - \mu_{rkj} )$ where $\mu_{rkj}=\mathbb E[b_{rkj}(X_k)]$. The Sobolev ellipsoid structure allows for effective control of the truncation error $\|D_{kj}\|_\infty$.

\vspace{0.2in}
\begin{lemma} \label{lem:bounded-beta}
    Under Assumption~\ref{assump:sobolev}, the group norm of true coefficients is bounded:
    $$
    \max_{(k,j)\in\mathcal E} \sum_{r=0}^{R_n} |\beta_{rkj}^\star|^2 \leq 2C^2
    $$
\end{lemma}
\begin{proof}
    For any $(k,j) \in \mathcal{E}$, we have
    $$
    \sum_{r=1}^{R_n} |\beta^\star_{rkj}|^2 
    \leq \left( \sum_{r=1}^{R_n}  |\beta^{\star}_{rkj}| \right)^2 
    \leq \left( \sum_{r=1}^{R_n} |\beta^{\star}_{rkj}| \cdot r^\eta \right)^2
    \leq \left( \sum_{r=1}^\infty |\beta^\star_{rkj}| \cdot r^\eta \right)^2 \leq C^2,
    $$
    and
$$
|\beta^\star_{0kj}| = \left|\mathbb E\left[
\sum_{r=1}^{\infty}\beta^\star_{rkj}b_{rkj}(X_k)
\right]\right|
\leq\sum_{r=1}^{\infty}|\beta^\star_{rkj}|
\leq\sum_{r=1}^{\infty}|\beta^\star_{rkj}|r^\eta
\leq C.
$$
    The result hence follows.
\end{proof}

\vspace{0.2in}
\begin{lemma} \label{lem:bounded-d}
Under Assumption \ref{assump:sobolev}, the truncation error decreases with $R_n$:
    $\|D_{kj} \|_\infty \lesssim R_n^{-\eta} $.
\end{lemma}
\begin{proof}
    We note that $D_{kj} = f_{kj} - f_{kj,n} = \sum_{r=R_n + 1}^\infty \beta_{rkj} (b_{rkj} - \mu_{rkj} )$, where $\mu_{rkj} = \mathbb{E}[b_{rkj}(X_k)]$. Therefore,
    \begin{align*}
    \begin{split}
        \|D_{kj}\|_\infty 
        &= \left\| \sum_{r=R_n + 1}^\infty \beta_{rkj} (b_{rkj} - \mu_{rkj} ) \right\|_\infty 
        \leq 2\sum_{r = R_n + 1}^\infty |\beta_{rkj}| 
        =2R_n^{-\eta} \sum_{r = R_n + 1}^\infty  |\beta_{rkj}| 
        \cdot R_n^{\eta} 
        \leq 2 C R_n^{-\eta}.
    \end{split}
    \end{align*}
\end{proof}

Finally, we show that the function classes are uniformly bounded.
\vspace{0.2in}
\begin{lemma} \label{lem:bounded-f}
    Under Assumption \ref{assump:sobolev}, the functions in $\mathcal{F}_{kj}$ are uniformly bounded: 
    $$ 
    \sup_{(k,j) \in \mathcal{E}} \sup_{f_{kj} \in \mathcal{F}_{kj} } \|f_{kj}(\cdot)\|_\infty \leq 2C,
    $$ 
    where $C$ is the constant appearing in Assumption \ref{assump:sobolev}.
\end{lemma}
\begin{proof}
    For any $f_{kj} \in \mathcal{F}_{kj}$, $\forall (k,j) \in \mathcal{E}$, we have
    \begin{equation*}
        \|f_{kj}(x)\|_\infty = \left\| \beta_{0kj} + \sum_{r=1}^{\infty} \beta_{rkj} b_{rkj}(x) \right\|_\infty 
        \leq \sum_{r=0}^{\infty} |\beta_{rkj}| 
        \leq |\beta_{0kj}| + \sum_{r=1}^{\infty} |\beta_{rkj}| \cdot r^\eta \leq |\beta_{0kj}| + C,
    \end{equation*}
    where
    \begin{equation*}
        \beta_{0kj} = - \mathbb{E}\left[ \sum_{r=1}^\infty \beta_{rkj} b_{rkj}(X_k) \right] \leq \sum_{r=1}^\infty |\beta_{rkj}|\leq \sum_{r=1}^\infty |\beta_{rkj}| \cdot r^\eta \leq C.
    \end{equation*}
    The result hence follows.
\end{proof}

\vspace{0.2in}
\subsection{More on boundedness of variances}
\label{sec:appendix:variances}

Recall that in Section \ref{sec:theory}, we defined the maximal variances $\overline{\nu}_p := \max_{j\in\mathrm{V}} \mathbb{E}[X_j^2] $, and the minimal variance $\underline{\nu}_p := \min_{j\in\mathrm{V}} \min_{|\vec{g}_j| =K  } \nu_j(\vec{g}_j)$. First, we point out that for any noise variance $\sigma_j^{\star2}$, there exists $\vec{g}_j$ such that $\nu_j(\vec{g}_j) = \sigma_j^{\star2}$; see Lemma \ref{lem:permu1} for a proof. Next, we see that for any $\nu_j(\vec{g}_j)$, 
$$
\underline{\nu}_p 
= \min_{j\in\mathrm{V}} \min_{|\vec{g}_j| \in \{0, \dots, K\} } \nu_j(\vec{g}_j)
\leq \nu_j(\vec{g}_j) 
\leq \max_{j\in\mathrm{V} } \max_{ |\vec{g}_j| \in \{0, \dots, K\} }\nu_j(\vec{g}_j)
= \overline{\nu}_p.
$$ 
The conditions in Assumption \ref{assump:bound-variance} ensure uniform boundedness of these quantities.

\vspace{0.2in}
\begin{lemma} \label{lem:var-bound}
Denote $\underline{\sigma} = \min_{j \in \mathrm{V}} \sigma_j^\star$.
    Under Assumptions \ref{assump:sobolev} and \ref{assump:bound-variance}, the variances are uniformly bounded: $\underline c\leq\underline\nu_p
\leq\underline\sigma^2\leq\overline\sigma^2$,
and $\overline\nu_p\leq\overline c$ for some $\underline{c} > 0$ and $\overline{c} < \infty.$ 
\end{lemma}
\begin{proof}
    It suffices to show that there exists $\overline{c} < \infty$ such that $\overline{\nu}_p \leq \overline{c}$. The bounds $\underline c\leq\underline\nu_p$ and $\underline\sigma^2\leq\overline\sigma^2$ follow from Assumption~\ref{assump:bound-variance}. Fitting the true parent functions gives $\underline\nu_p\leq\underline\sigma^2$. It remains to show that $\overline\nu_p\leq\overline c$ for some finite constant $\overline c$. Note that
    \begin{align*}
    \begin{split}
        \overline{\nu}_p 
        &= \max_{j\in\mathrm{V} } \mathbb{E} \left[ X_j^2 \right] 
        \leq \max_{j\in\mathrm{V} } \mathbb{E} \left[ 2 \left( \sum_{k\in \mathrm{pa}(j)} f^\star_{kj}(X_k) \right)^2 + 2 \epsilon_j^2 \right]
        \leq 8C^2 K^2 + 2 \overline{\sigma}^2 < \infty.
    \end{split}
    \end{align*}
    The last inequality follows from the fact that $|\mathrm{pa}(j)|\leq K$, and Lemma \ref{lem:bounded-f}.
    
\end{proof}

We also defined the gap $d_{n,p} :=  \max_{j\in\mathrm{V}} \max_{|\vec{g}_j| \in \{0, \dots, K\} }  \left| \nu_j(\vec{g}_j) - \nu^n_j(\vec{g}_j) \right|$. As $n$ is sufficiently large, the function class $\mathcal{F}_n^{\oplus \vec{g}_j }$ is expected to closely approximate its non-truncated counterpart $\mathcal{F}^{\oplus \vec{g}_j }$. Particularly, when $n = \infty$, the function class $\mathcal{F}_n^{\oplus \vec{g}_j }$ and $\mathcal{F}^{\oplus \vec{g}_j }$ coincide exactly for any $\vec{g}_j$, and thus $\nu_j^n(\vec{g}_j) = \nu_j(\vec{g}_j)$. This observation motivates the following result: $d_{n,p} = o(1)$. 

\vspace{0.2in}
\begin{lemma}\label{lem:small-gap}
    Under Assumptions \ref{assump:sobolev} and \ref{assump:bound-variance},
    the gap between $\nu_j(\vec{g}_j)$ and $\nu_j^n(\vec{g}_j)$ is sufficiently small for any $\vec{g}_j$: $d_{n,p} = o(1)$.
\end{lemma}
\begin{proof}
    We denote 
    $$
    \bar{f} \in \argmin_{f \in \mathcal{F}^{\oplus\vec{g}_j}} \mathbb{E} [ X_j - \sum_{k\in \vec{g}_j} f_{kj}(X_k) ]^2,
    $$
    and let $\bar{f}^{(n)} \in \mathcal{F}_{n}^{\oplus \vec{g}_j}$ be the truncation of $\bar{f}$ to the first $R_n$ basis functions.
    Note that for any $\vec{g}_j$ and any $j\in \mathrm{V}$,
    \begin{align*}
    \begin{split}
     &   \nu^n_j(\vec{g}_j) - \nu_j(\vec{g}_j)
    \leq  \mathbb{E}\left[ X_j - \sum_{k\in \vec{g}_j} \bar{f}^{(n)}_{kj}(X_k) \right]^2 -  \mathbb{E}\left[ X_j - \sum_{k\in \vec{g}_j} \bar{f}_{kj}(X_k) \right]^2 \\
    ={}& \mathbb{E}\left[ X_j - \sum_{k\in \vec{g}_j} \bar{f}_{kj}(X_k) + \sum_{k\in \vec{g}_j} \bar{f}_{kj}(X_k) - \sum_{k\in \vec{g}_j} \bar{f}^{(n)}_{kj}(X_k) \right]^2 -  \mathbb{E}\left[ X_j - \sum_{k\in \vec{g}_j} \bar{f}_{kj}(X_k) \right]^2 \\
    ={}& 2 \cdot \mathbb{E}\left[ \left( X_j - \sum_{k\in \vec{g}_j} \bar{f}_{kj}(X_k) \right) \left( \sum_{k\in \vec{g}_j} \bar{f}_{kj}(X_k) - \sum_{k\in \vec{g}_j} \bar{f}^{(n)}_{kj}(X_k) \right) \right]
    + \mathbb{E}\left[ \sum_{k\in \vec{g}_j} \bar{f}_{kj}(X_k) - \sum_{k\in \vec{g}_j} \bar{f}^{(n)}_{kj}(X_k) \right]^2 \\
    \leq{}& 2 \cdot \left( \nu_j(\vec{g}_j) \right)^{1/2} \cdot \Delta_f^{1/2}
    + \Delta_f ,
    \end{split}
    \end{align*}
    where $\Delta_f = \mathbb{E}\left[ \sum_{k\in \vec{g}_j} \bar{f}_{kj}(X_k) - \sum_{k\in \vec{g}_j} \bar{f}^{(n)}_{kj}(X_k)  \right]^2$.
    By Lemma \ref{lem:var-bound}, we have $\left( \nu_j(\vec{g}_j) \right)^{1/2} \leq \bar{c}^{1/2}$. Moreover, letting $\mu_{rkj} = \mathbb{E}[b_{rkj}(X_k)] \leq 1 $, we have
    \begin{align*}
    \begin{split}
        \Delta_f 
    &\leq |\vec{g}_j| \sum_{k \in \vec{g}_j} \mathbb{E} \left[  \bar{f}_{kj}(X_k) -  \bar{f}^{(n)}_{kj}(X_k) \right]^2  \\
    &\leq K  \sum_{k \in \vec{g}_j} \sum_{r = R_n+1}^{\infty} \sum_{r'= R_n+1}^{\infty} \beta_{rkj} \beta_{r'kj} \mathbb{E}\left[ \left( b_{rkj}(X_k) - \mu_{rkj} \right)\left( b_{r'kj}(X_k) - \mu_{r'kj} \right) \right]  \\
    & \leq K^2 \upsilon \sum_{r = R_n+1}^{\infty} \beta^2_{rkj} + K^2 \left[ \sum_{r= R_n+1}^{\infty} \left| \beta_{rkj} \mu_{rkj} \right| \right]^2 \\
    &\leq K^2 \upsilon R_n^{-2\eta} \left[ \sum_{r = R_n+1}^{\infty} \beta^2_{rkj} r^{2\eta} \right] + K^2 R_n^{-2\eta} \left[ \sum_{r=R_n + 1}^{\infty} \beta_{rkj}^2 r^{2\eta} \right] \leq (1 + \upsilon) K^2 C^2 R_n^{-2\eta},
    \end{split}
    \end{align*}
    where the last inequality follows from Assumption \ref{assump:sobolev} where $\sum_{r=R_n + 1}^\infty \beta_{rkj}^2 \cdot r^{2\eta} \leq \left[ \sum_{r=1}^\infty |\beta_{rkj}| \cdot r^\eta \right]^2 \leq C^2 $. As a result, we have
    $$
    0 \leq \nu_j^n(\vec{g}_j) - \nu_j(\vec{g}_j) \le 2 \bar{c}^{1/2} (1 + \upsilon)^{1/2} K C R_n^{-\eta} + (1 + \upsilon) K^2 C^2 R_n^{-2\eta}.
    $$
    The result hence follows.

\end{proof}

\vspace{0.2in}
\subsection{More on permutations and the likelihood}
\label{sec:appendix:permutation}

Recall that for any permutation $\pi$, we define the parents of node $j$ as $\vec{g}(\pi)_j= \{ k: \pi^{-1}(k) < \pi^{-1}(j), (k,j) \in E^\circ \}$, and the corresponding log-likelihood  $\log p(X; \theta, \pi):= \sum_{j=1}^p \log p_j(X; \theta, \pi)$, where
\begin{align*}
\begin{split}
    \log p_j(X;\theta, \pi):=  -\frac{1}{2}\log (2\pi) - \log \sigma_j - \frac{ \left[X_j - \sum_{k\in \vec{g}(\pi)_j} f_{kj}(X_k) \right]^2 }{2\sigma_j^2}.
\end{split}
\end{align*}
Let $\theta^\star:=(\{f^\star_{kj}\}_{(k,j) \in E^\star}, \{\sigma^\star_j\}_{j \in \mathrm{V} })$, and consider the expected negative log-likelihood:
$$
\min_{\theta} \mathbb{E}_{\theta^\star} \left[ - \log p(X; \theta, \pi) \right] = \sum_{j=1}^p \min_{\theta} \mathbb{E}_{\theta^\star} \left[ -\log p_j(X; \theta, \pi) \right].
$$ 
In order to find $\theta^\pi := \argmin_{\theta} \mathbb{E}_{\theta^\star}[-\log p(X; \theta, \pi)] $, we analyze the minimum of each component individually. For each $j$, we observe that 
\begin{equation*}
    \argmin_{\sigma^2_j > 0} \mathbb{E}_{\theta^\star} \left[ -\log p_j(X;\theta, \pi) \right] = \mathbb{E}_{\theta^\star} \left[ \left( X_j - \sum_{k\in \vec{g}(\pi)_j} f_{kj}(X_k) \right)^2 \right] = \mathbb{E}_{\theta^\star} \left[ \mathcal{R}_j(f; \vec{g}(\pi)_j) \right].
\end{equation*}
Hence, 
\begin{align*}
\begin{split}
    \min_{\theta} \mathbb{E}_{\theta^\star} \left[ -\log p_j(X; \theta, \pi) \right]=  C  + \frac{1}{2} \min_{f\in \mathcal{F}^{\oplus \vec{g}(\pi)_j}} \log \mathbb{E}_{\theta^\star}\left[ \mathcal{R}_j(f; \vec{g}(\pi)_j) \right] = C + \frac{1}{2} \log \nu_j(\vec{g}(\pi)_j),
\end{split}
\end{align*}
where $C  = \log (2\pi) / 2 + 1/2$.

\vspace{0.2in}
\begin{lemma} \label{lem:permu1}
    For any permutation $\pi$, we have $\sum_{j=1}^p \log  \nu_j(\vec{g}(\pi)_j) \geq \sum_{j=1}^p \log \sigma^{\star2}_j$. Particularly, when $\pi \in \Pi^\star$, $\nu_j(\vec{g}(\pi)_j) = \sigma^{\star2}_j$ for any $j\in \mathrm{V}$.
\end{lemma}
\begin{proof}
    First, we show that when $\pi\in \Pi^\star$, $\{f^\star_{kj} \}_{k\in \mathrm{V}}  \in \argmin_{f} \mathbb{E}_{\theta^\star} \left[ \mathcal{R}_{j} (f; \vec{g}(\pi)_j) \right]$ and $\nu_j(\vec{g}(\pi)_j) = \sigma^{\star2}_j$. To see this, for any $\{f^0_{kj}\}_{k\in \mathrm{V}}$, we have
    \begin{align*}
    \mathbb{E}_{\theta^\star}\left[ \mathcal{R}_j(f^0; \vec{g}(\pi)_j) \right] 
    = \mathbb{E}_{\theta^\star}\left[ \epsilon_j^2 + \left\{\sum_{k\in \vec{g}(\pi)_j} \left( f_{kj}^\star(X_k) - f^0_{kj}(X_k) \right)\right\}^2 \right]  \geq \sigma^{\star2}_j,
    \end{align*}
    and the equality holds when $f^\star_{kj} = f^0_{kj}$. The desired result thus follows. In addition, the result implies that, when $\pi \in \Pi^\star$: (i) $\theta^\pi = \theta^\star$; (ii) $\mathbb{E}_{\theta^\star} \left[ - \log p_j(X; \theta^\star, \pi) \right] = C + \log (\sigma_j^\star)$.
    
    Next, we show that $\sum_{j=1}^p \log \nu_j(\vec{g}(\pi)_j) \geq \sum_{j=1}^p \log \sigma_j^{\star^2}$ for any permutation $\pi$. Note that
    $$
    \frac{1}{2}\sum_{j=1}^p \log  \nu_j(\vec{g}(\pi)_j) - \frac{1}{2}\sum_{j=1}^p \log (\sigma_j^{\star2}) 
    = \mathbb{E}_{\theta^\star} \left[-\log p(X; \theta^\pi, \pi) \right] - \mathbb{E}_{\theta^\star} \left[ - \log p(X; \theta^\star, \pi^\star)  \right] \geq 0 \, ,
    $$
    where $\pi^\star\in\Pi^\star$ is any true ordering. The last inequality holds since the KL-divergence is non-negative. The result hence follows.
\end{proof}

Lemma \ref{lem:permu1} shows that when $\pi \in \Pi^\star$, the residual variance $\nu_j(\vec{g}(\pi)_j)$ attains the irreducible noise variance. Moreover, the separation quantity $\xi_p$ defined in \eqref{eq:separation} is non-negative. We now establish that $\xi_p > 0$.

\vspace{0.2in}
\begin{lemma} \label{lem:positive_xi}
    Under the conditions of Proposition~\ref{prop:identifiability} and Assumption \ref{assump:compatibility}, the separation is strictly positive: $\xi_p > 0$.
\end{lemma}
\begin{proof}
    For any $\pi$, define $\vec{g}(\pi^\circ)_j  := \{k: \pi^{-1}(k) < \pi^{-1}(j) \}$. We have $\vec{g}(\pi)_j \subseteq \vec{g}(\pi^\circ)_j$. Then when $\pi \notin \Pi^\star$,
    $$
    \xi_p= \min_{\pi\notin \Pi^\star} p^{-1} \sum_{j=1}^p \left( \log \nu_j(\vec{g}(\pi)_j) - \log \sigma_j^{\star2} \right) 
    \geq \min_{\pi\notin \Pi^\star} p^{-1} \sum_{j=1}^p \left( \log \nu_j(\vec{g}(\pi^\circ)_j) - \log \sigma_j^{\star2} \right) > 0.
    $$
    The last strict inequality follows from Lemma 3 in \cite{buhlmann2014}.
\end{proof}

\vspace{0.2in}
Finally, we provide an illustration on Assumption~\ref{assump:separation}.
Although Lemma~\ref{lem:positive_xi} establishes $\xi_p>0$ under identifiability and regularity conditions, the gap may be very small. In finite samples, statistical estimation errors and basis approximation errors can obscure this population difference, potentially leading to an incorrect ordering. Assumption~\ref{assump:separation} therefore requires the population separation dominates these errors.

\begin{example} \label{ex:assumption-separation}
Consider the nonlinear Gaussian SEM with two variables:
$$
X_1 \sim N(0,1), \qquad
X_2 = a\big(\cos(X_1)-e^{-1/2}\big)+\varepsilon_2,
$$
where $\varepsilon_2\sim N(0,\sigma^2)$ is independent of $X_1$, and $a\neq 0$. The true DAG is $X_1\to X_2$, with population residual variances $1$ and $\sigma^2$. Under the reverse ordering $X_2\to X_1$, symmetry implies that $\mathbb{E}[X_1\mid X_2]=0$. Consequently, the population separation defined in
\eqref{eq:separation} is
$$
\xi_2 = \frac{1}{2}\log\left( 1+\frac{a^2\operatorname{Var}(\cos X_1)}{\sigma^2} \right)>0.
$$
This example illustrates how nonlinear causal relationships can yield a strictly positive likelihood gap between correct and incorrect orderings.

When $a\neq 0$ and $\sigma^2$ are fixed, $\xi_2$ is a positive constant, and Assumption~\ref{assump:separation} holds as $n\to\infty$, provided the estimation and approximation errors vanish. In contrast, if $a=a_n\to 0$, the separation decreases and may become too small to distinguish the two orderings reliably. This illustrates why Assumption~\ref{assump:separation} requires the population separation to dominate estimation and approximation errors.
\end{example}

\vspace{0.2in}
\section{Supporting theoretical results}
\label{sec:supporting-theory}

\subsection{Correct permutation}
\label{sec:theory:permutation}

\begin{theorem}\label{thm:correct-permutation}
    Suppose conditions in Proposition \ref{prop:identifiability} and Assumptions \ref{assump:superstructure-sparsity}--\ref{assump:separation} are satisfied. Take $\lambda_n^2  \vee \gamma_n = o(p\xi_p/s_n)$. When $n$ is sufficiently large, it holds, with probability at least $1 - 4/p$, that:
    \begin{enumerate}[label = (\arabic*)]
        \item The estimated graph has correct permutations: 
    $\widehat{\Pi} \subseteq \Pi^\star.$

        \item When the early-stopping threshold satisfies $\tau^{\rm early} = o(p\xi_p)$, the resulting estimated graph has correct permutations: $\widehat{\Pi}^{\rm early} \subseteq \Pi^\star$.
    \end{enumerate}
\end{theorem}

A proof is given in Appendix~\ref{sec:appendix:prof-correct-permutation}. This proof generally follows the general idea of Theorem~1 in \cite{buhlmann2014}, but relies on the KL-divergence gap between any incorrect permutation distribution and the distribution induced by our estimator with regularization. The chosen $\lambda_n^2$ cannot be too large to avoid eliminating relevant edges, while $\gamma_n$ cannot be too large to avoid over-shrinking relevant signals. Interestingly, even a sub-optimal solution obtained via early stopping can identify the correct permutations, provided that the optimality gap threshold is appropriately controlled. This result enables substantial computational savings without
compromising estimation accuracy.

\vspace{0.2in}
\subsection{Convergence of noise variances}
\label{sec:theory:noise-convergence}


\begin{theorem}\label{thm:var-converge}
Under conditions in Proposition \ref{prop:identifiability} and Assumptions \ref{assump:superstructure-sparsity}--\ref{assump:separation}, it holds, with probability at least $1 - 5/p$,
$$
\sum_{j=1}^p \left( \sigma_j^{\star2} - \widehat{\sigma}_j^2 \right)^2 \lesssim (p + s_n R_n) \log p / n  + s_n R_n^{-2\eta} + \gamma_n s_n + \lambda_n^2 s_n.
$$
When $s_n R_n^{-2\eta} + s_nR_n \log p /n \lesssim p \log p /n$, and the tuning parameters are chosen as $\lambda_n^2 \asymp p\log p / (n s_n)$, $\gamma_n \lesssim \lambda_n^2$, then
\begin{enumerate} [label=(\arabic*)]
    \item 
    $
    \sum_{j=1}^p \left( \sigma_j^{\star2} - \widehat{\sigma}_j^2 \right)^2 \lesssim \lambda_n^2 s_n $.

    \item When the early-stopping threshold satisfies $\tau^{\rm early} \lesssim (\lambda_n^2 + \gamma_n)  s_n  + (p + s_nR_n ) \log p /n + s_n  R_n^{-2\eta}$ and $\tau^{\rm early} = o(p\xi_p)$, we have $\sum_{j=1}^p ( \sigma_j^{\star2} - (\widehat{\sigma}_j^{\rm early})^2 )^2 \lesssim \lambda_n^2 s_n$.
\end{enumerate}

\end{theorem}

Theorem~\ref{thm:var-converge} is proved in Supp. Mat.~\ref{sec:appendix:prof-var-converge}. The proof generally follows the strategy of Lemma~7.1 in \cite{vandegeer2013}, but requires adaptations to accommodate the non-linearity. Some possible choices for the scaling quantities are: $n \asymp p (\log p)^2$, $s_n \asymp R_n \asymp n^{1/5}$. Since $\lambda_n^2 s_n \vee \gamma_ns_n = o(1)$, this result implies convergence of the estimated noise variances: $| \widehat{\sigma}^2_j - \sigma_j^{\star2}| = o(1)$ for all $j \in \mathrm{V}$ with high probability. This upper bound $\lambda_n^2 s_n$ mirrors the result of linear SEMs in \cite{vandegeer2013}. Moreover, as noted in their Remark~3.1, a natural normalization of the quantity $\sum_{j=1}^p ( \sigma_j^{\star2} - \widehat{\sigma}_j^2 )^2$ is to divide it by $p$. Under this perspective, one can instead choose $n \asymp (\log p)^2$, $s_n \asymp n$, and $R_n \asymp n$, and obtain a pooled variance convergence: $p^{-1}\sum_{j=1}^p ( \sigma_j^{\star2} - \widehat{\sigma}_j^2 )^2 \lesssim \log p /n $. This version accommodates a high-dimensional setting where $n \ll p$. Finally, for variances convergence, the early-stopping threshold $\tau^{\rm early}$ may need to be much smaller than that required to ensure correct permutations in the worst case scenario. To see this, recall that we could choose $o(p\xi_p) = p \sqrt{\log p / n}$ and $(\lambda^2_n + \gamma_n) s_n + (p + s_n R_n) \log p / n + s_n R_n^{-2\eta} \asymp p \log p / n $; then in order to ensure the correct permutation $\tau^{\rm early}$ is at most $ p \sqrt{\log p / n} $, whereas to ensure variance convergence it is at most $p \log p / n $.

\vspace{0.2in}
\section{Preliminary results}
\subsection{Events that hold with high probability}
In this section, $\epsilon_j \in \mathbb{R}^n$, $X_j \in\mathbb{R}^n$ and $D^\star_{kj} \in \mathbb{R}^n$ represent vectors containing $n$ independent samples, and $X^{(i)} \in \mathbb{R}^p$ represents the $i$-th sample containing $p$ variables.

\begin{lemma} \label{lem:sigma-concen}
    Suppose $\epsilon_j \sim \mathcal{N}(0, \sigma_j^2)$, $j = 1, \dots, p$ are independently distributed. Define the event
    \begin{equation*}
        \mathcal{T}_1:= \left\{ \left| \left\| \epsilon_j \right\|_n^2 - \sigma_j^{\star2} \right| \lesssim  \sqrt{ \log p / n} , \quad \forall j \in [p] \right\}.
    \end{equation*}
    Then $\mathbb{P}(\mathcal{T}_{1}) \geq 1 - 2 / p$.
\end{lemma}
\begin{proof}
    By the Bernstein's inequality~\citep{bennett1962proabbility} we have $\mathbb{P}(\mathcal{T}_{1,j}) \geq 1 - 2e^{-t} / p$, where
    $$
    \mathcal{T}_{1,j} := \left\{ \frac{\left| \left\| \epsilon_j \right\|_n^2 - \sigma_j^{\star2} \right|}{\sigma_j^{\star2}} \leq 2 \sqrt{(t + \log p) / n} + 2 (t+\log p) / n \right\}
    $$
    By letting $t = \log p$ and using the fact $\mathcal{T}_1 = \cap_{j=1}^p \mathcal{T}_{1,j}$, the result follows.
\end{proof}

\vspace{0.2in}
Let $P$ be the true distributions induced by model \eqref{eq:DAG-model}, and $P_n$ the empirical distribution. We denote $dx:= (dx_1, \dots, d x_p)$. 

\vspace{0.2in}
\begin{lemma}
    Define $\mu_{kk'j}(f)(x_1, \dots, x_p) := f_{kj}(x_k) f_{k'j}(x_{k'})$. Let the combined function class for node $j$ be
    $\mathcal{F}^\cup_j:=\bigcup_{k=1}^p \bigcup_{k'=1}^p \{\mu_{kk'j} (f_{kj}, f_{k'j}) : f_{kj} \in \mathcal{F}_{kj}, f_{k'j} \in \mathcal{F}_{k'j} \} $. Define the function
    $
    g_j(x^{(1)}, \dots, x^{(n)}) := \sup_{f\in\mathcal{F}_j^\cup}  \left| \int f(x) (P_n - P)(dx) \right|
    $, and the event
    \begin{equation*}
    \mathcal T_2:=\left\{
    \begin{aligned}
    g_j(X^{(1)},\ldots,X^{(n)}) &\lesssim \sqrt{\log p/n},\\
    \max_{k\neq j}\sup_{f_{kj}\in\mathcal F_{kj}} \left|(P_n-P)[\epsilon_j f_{kj}(X_k)]\right|
    &\lesssim \sqrt{\log p/n},
    \qquad \forall j\in[p]
    \end{aligned}
    \right\}.
    \end{equation*}
    Then $\mathbb{P}(\mathcal{T}_2) \geq 1- 1/p$.
\end{lemma}

\begin{proof}
    Note that each function class $\mathcal{F}_{kj}$ has finite bracketing entropy by Lemma \ref{lem:bracketing}. Moreover, each $f_{kj} \in \mathcal{F}_{kj}$ is uniformly bounded with $|f_{kj}| \leq 2C$ by Lemma \ref{lem:bounded-f}.
    Now for the function class $\mathcal{F}^\cup_j$, its bracketing entropy is given by
    \begin{equation*}
    \log N_{[]} (\delta,\mathcal{F}^\cup_j, \|\cdot\|_\infty ) \leq \log (p^2 \cdot \max_{k,j} N^2_{[]}\left( \delta/(4C), \mathcal{F}_{kj} , \|\cdot\|_\infty) \right)  
    \lesssim \log p + (1 / \delta)^{1/\eta} \log (1/\delta).
    \end{equation*}
    When $\delta \asymp 1$, the corresponding bracketing integral is given by
    \begin{equation*}
        J_{[]} (\delta, \mathcal{F}^\cup_j, \|\cdot\|_\infty) 
    = \int_{0}^\delta  \sqrt{\log N_{[]} (u,\mathcal{F}^\cup_j, \|\cdot\|_\infty )} du 
    \lesssim \sqrt{\log p} .
    \end{equation*}
    Hence, from Corollary 19.35 of \cite{van2000asymptotic}, with $F(x) = \sup_{f\in \mathcal{F}^\cup_j} |f(x)| \asymp 1$,
    \begin{equation*}
        \mathbb{E} \left[ g_j(X^{(1)}, \dots, X^{(n)}) \right] 
    \lesssim \frac{  J_{[]} (\|F\|_\infty, \mathcal{F}_j^\cup, \|\cdot\|_\infty ) }{\sqrt{n}} 
    \lesssim \sqrt{\frac{\log p}{n}}.
    \end{equation*}
    Note that the function $g_j(x^{(1)}, \dots, x^{(n)})$ satisfies the bounded difference property with bound $8C^2 /n$ for all $n$ samples.
    To see this, note that for any $i = 1, \dots, n$,
    \begin{align*}
    \begin{split}
        &\left|g_j(x^{(1)}, \dots, x^{(i)} , \dots, x^{(n)} ) - g_j (x^{(1)}, \dots, x^{(i)'}, \dots, x^{(n)}) \right| \\
    \leq{}& \sup_{f\in \mathcal{F}_j^{\cup}} \left| \frac{1}{n} f(x^{(i)}) - \frac{1}{n}\mathbb{E}[f(X^{(i)})] -\frac{1}{n} f(x^{(i)'}) + \frac{1}{n} \mathbb{E}[f(X^{(i)})] \right|  
    \leq \frac{1}{n} \sup_{f\in \mathcal{F}_j^{\cup}} \left| f(x_i) - f(x^{(i)'}) \right|
    \leq \frac{8C^2}{n}.
    \end{split}
    \end{align*}
    Then by the McDiarmid's inequality, we have $\mathbb{P}(\mathcal{T}_{2,j}) \geq 1 - e^{-t} / p$, where
    \begin{equation*}
        \mathcal{T}_{2,j} := \left\{ g_j(X^{(1)}, \dots, X^{(n)}) \leq  \mathbb{E}\left[ g_j(X^{(1)}, \dots, X^{(n)}) \right] + C^2\sqrt{32(t + \log p) / n} \right\}.
    \end{equation*}
    Taking $t=\log(2p)$ and a union bound over $j$ establishes the first set of bounds with probability at least $1-1/(2p)$.

To establish the second set, let $b_{0kj}=1$. Since $|\epsilon_j b_{rkj}(X_k)|\leq|\epsilon_j|$, these products are uniformly sub-Gaussian. Concentration for averages of independent sub-Gaussian variables \citep{vershynin2018high}, followed by a union bound over
$j,k,r$, gives, with probability at least $1-1/(2p)$, 
$$
\left|(P_n-P)[\epsilon_j b_{rkj}(X_k)]\right| \lesssim \sqrt{\frac{\log p+\log(r+1)}{n}},
\qquad \forall j,k,\ r\geq0.
$$
The implicit constant is chosen sufficiently large that the failure probability for each $(j,k,r)$ is at most $\{4p^3(r+1)^2\}^{-1}$. Summing these probabilities over all $j,k,r$ gives a total failure probability at most $1/(2p)$. Consequently,
\begin{align*}
\left|(P_n-P)[\epsilon_j f_{kj}(X_k)]\right|
&\lesssim |\beta_{0kj}|\sqrt{\frac{\log p}{n}} +\sum_{r\geq1}|\beta_{rkj}| \sqrt{\frac{\log p+\log(r+1)}{n}}\\
&\lesssim \left( |\beta_{0kj}|+ \sum_{r\geq1}|\beta_{rkj}|r^\eta \right)\sqrt{\frac{\log p}{n}}\\
&\lesssim \sqrt{\frac{\log p}{n}},
\end{align*}
uniformly over $\mathcal F_{kj}$. Here we used Assumption~\ref{assump:sobolev}, the centering condition, and $\sqrt{\log p+\log(r+1)}\lesssim r^\eta\sqrt{\log p}$. Combining the two probability bounds yields $\mathbb P(\mathcal T_2)\geq1-1/p$.
    
\end{proof}

\vspace{0.2in}
\begin{lemma}
    Let $D^\star_{kj}=\sum_{r=R_n+1}^{\infty} \beta^\star_{rkj}(b_{rkj}(X_k)-\mu_{rkj})$, where $\mu_{rkj}=\mathbb E[b_{rkj}(X_k)]$. Define the event 
    \begin{equation*}
        \mathcal{T}_3 := \left\{ n^{-1} \sum_{k \in \mathrm{pa}(j)} \left|  \epsilon_j^\top D^\star_{kj}  \right|  \lesssim s_j R_n^{-2\eta} + s_j \log p /n, \quad \forall j \in [p]  \right\}.
    \end{equation*}
    Then $\mathbb{P}(\mathcal{T}_3) \geq 1 -1/p$.
\end{lemma}
\begin{proof}
    First consider the set 
    $$
    \mathcal{T}_{3,j}:= \left\{
    n^{-1} \sum_{k \in \mathrm{pa}(j)} \left|  \epsilon_j^\top D^\star_{kj}  \right|  \lesssim s_j R_n^{-2\eta} + s_j (t +\log p) /n 
    \right\}.
    $$
    Since $\epsilon_j$ and $D^\star_{kj}$ are independent for all $k\in\mathrm{pa}(j)$, we condition on $ D^\star_{kj}$. By Lemma 7.4 of \cite{vandegeer2013}, applied with the one-column matrix $Z=D^\star_{kj}$, with conditional probability
at least $1-e^{-t}/p$,
    \begin{equation*}
         (n\sigma_j^\star)^{-1} \sup_{\substack{a\in\mathbb R\\ \|aD^\star_{kj}\|_n\leq1}} |\epsilon_j^\top(aD^\star_{kj})|\leq  \sqrt{2/n} + \sqrt{2(t + \log p) / n}.
    \end{equation*}
    Consequently, for any fixed $\delta_1 \in (0,1)$,
    \begin{align*}
    2(n\sigma_j^\star)^{-1}|\epsilon_j^\top D^\star_{kj}| &\leq 2\left(\sqrt{2/n}+\sqrt{2(t+\log p)/n}\right) \|D^\star_{kj}\|_n\\
    &\leq \delta_1\|D^\star_{kj}\|_n^2 +\frac{4}{n\delta_1} +\frac{4(t+\log p)}{n\delta_1}\\
    &\lesssim R_n^{-2\eta}+(t+\log p)/n.
    \end{align*}
    The last inequality follows from Lemma \ref{lem:bounded-d}. Taking a union bound over the at most $K$ parents gives $\mathbb{P} (\mathcal{T}_{3,j})\geq1-Ke^{-t}/p$.
    By letting $t=\log(pK)$ and using $\mathcal{T}_3=\cap_{j=1}^p\mathcal{T}_{3,j}$, the result follows, since $K$ is bounded.
\end{proof}

\vspace{0.2in}
We define the following notations used in Lemma \ref{lem:event:T4} and \ref{lem:event:T5}. Let
$$
Z_j:=\left[1_n, b_{1 1j}(X_1), \dots, b_{1 pj}(X_p), \dots, b_{R_n 1j} (X_1), \dots, b_{R_n pj}(X_p)  \right] \in \mathbb{R}^{n \times (p R_n + 1)},
$$
and
$$
\beta_j = \left[ {\beta}_{0 *j}, \beta_{11j}, \dots, \beta_{1pj}, \dots, \beta_{R_n1j}, \dots, \beta_{R_n pj} \right]^\top \in \mathbb{R}^{pR_n + 1}, \quad
\beta_{0*j} = \sum_{k=1}^p \beta_{0kj}.
$$
Moreover, denote all permutations corresponding to $\mathcal{G}({\beta})$ as $\Pi_{{\beta}}$, where node $k$ is a parent of node $j$ if $\sum_{r=1}^{R_n}|{\beta}_{rkj}|^2 \neq 0$, and ${\beta}_j = 0$ if node $j$ has no parents. 

\vspace{0.2in}
\begin{lemma}\label{lem:event:T4}
    We define the event 
    \begin{equation*}
        \mathcal{T}_4:= \left\{ 2n^{-1} \left| \epsilon_j^\top Z_j (\beta^\star_j - \widetilde{\beta}_j) \right| 
        \leq \delta_2 \|Z_j (\beta^\star_j - \widetilde{\beta}_j)\|_n^2 + \Delta_{n,p}, \quad \text{$\forall \widetilde{\beta}_j$ with ${\Pi}_{\widetilde{\beta}} \subseteq \Pi^\star$},  \quad \forall j \in [p] \right\},
    \end{equation*}
    where $\delta_2 \in (0,1)$ and $\Delta_{n,p} \lesssim s_j R_n \log p / n$.
    Then $\mathbb{P}(\mathcal{T}_4) \geq 1 - 1 / p$.
\end{lemma}
\begin{proof}
    First we consider the set
\begin{equation*}
    \mathcal{T}_{4,j} := \left\{ 2n^{-1} \left| \epsilon_j^\top Z_j(\beta_j^\star - \widetilde{\beta}_j) \right| 
    \le \delta_2 \|Z_j(\beta_j^\star - \widetilde{\beta}_j)\|_n^2 + \Delta(t)_{n,p}, \quad {\Pi}_{\widetilde{\beta}} \subseteq \Pi^\star \right\},
\end{equation*}
where $\Delta(t)_{n,p} \lesssim s_j ( R_n +t + \log p)/n$.
When the permutations ${\Pi}_{\widetilde{\beta}}$ are correct, the regressors involved are non-descendants of node $j$. We condition on the basis design matrix of all candidate non-descendants, which is independent of $\epsilon_j$. By Lemma 7.4 of \cite{vandegeer2013}, applied uniformly over coefficient vectors, with probability at least $1-e^{-t}/p$,
\begin{align*}
\begin{split}
    (n\sigma_j^\star)^{-1} \sup_{\|Z(\beta^\star_j - \widetilde{\beta}_j)\|_n \leq 1}  \left| \epsilon_j^\top Z_j (\beta_j^\star - \widetilde{\beta}_j) \right|
    & \leq \sqrt{2( 2 R_n s_j + 1) / n} + \sqrt{2 (t + \log p) / n} \\
    & \leq \sqrt{2( 2R_n s_j + 1) / n} + \sqrt{2s_j (t + \log p) / n}.
\end{split}
\end{align*}
The last inequality can be justified by considering two cases: (1) when $s_j \geq 1$, the inequality is trivially satisfied; (2) when $s_j = 0$, we have $\beta_j^\star = \widetilde{\beta}_j = 0$ due to $\{k:(k,j) \in E^\circ\} = \varnothing$, and thus the inequality again holds trivially.
Consequently, for some $\delta_2 \in (0,1)$,
\begin{align*}
\begin{split}
&    2 (n\sigma_j^\star)^{-1} \left|\epsilon_j^\top Z_j (\beta^\star_j - \widetilde{\beta}_j) \right| = 2(n\sigma_j^\star)^{-1} \left| \epsilon_j^\top Z_j (\beta^\star_j - \widetilde{\beta}_j) \right| \cdot \frac{ \|Z_j (\beta^\star_j - \widetilde{\beta}_j)\|_n }{ \|Z_j (\beta^\star_j - \widetilde{\beta}_j)\|_n }  \\
\leq{}&  2 \left( \sqrt{2 (2R_n s_j + 1) / n} + \sqrt{2 s_j(t + \log p) / n} \right)  \|Z_j (\beta^\star_j - \widetilde{\beta}_j)\|_n  \\
\leq{}& \delta_2 \|Z_j (\beta^\star_j - \widetilde{\beta}_j)\|_n^2 + (8R_n s_j + 4) / (n \delta_2) + 4 s_j (t + \log p) / (n \delta_2) .
\end{split}
\end{align*}
This implies that $\mathbb{P}(\mathcal{T}_{4,j} \mid Z_j(\beta^\star_j - \widetilde{\beta}_j)) \geq 1 - e^{-t} / p$, and consequently $\mathbb{P}(\mathcal{T}_{4,j}) \geq 1 - e^{-t} / p$. By letting $t = \log p$ and using the fact that $\mathcal{T}_4 = \cap_{j=1}^p \mathcal{T}_{4,p}$, the result follows.

\end{proof}

\vspace{0.2in}
\begin{lemma} \label{lem:min-eigen}
For the given super-structure $E^\circ$, let
$S_j^\circ:=\{k:(k,j)\in E^\circ\}$ and
$h_j(X):=[1,\{b_{rkj}(X_k)\}_{k\in S_j^\circ,\ r\in[R_n]}]^\top$.
Then $\min_{j\in[p]}\lambda_{\min}
(\mathbb E[h_j(X)h_j(X)^\top])\geq \min\{1, \phi^2\upsilon\}$.
\end{lemma}

\begin{proof}
For each $j$, consider an arbitrary coefficient vector $a=(a_0,\{a_{rkj}\}_{k\in S_j^\circ,r\in[R_n]})$.
Define
$g_{kj}(x)=\sum_{r=1}^{R_n}a_{rkj}b_{rkj}(x)$.
Note that there must exist a constant $t_k>0$ such that $t_kg_{kj}\in\mathcal F_{kj}$.
Applying Assumption~\ref{assump:compatibility} with $f_{kj}=t_kg_{kj}$ and $\gamma_k=t_k^{-1}$, and Assumption~\ref{assump:sobolev} yields
$$
\left\|\sum_{k\in S_j^\circ}g_{kj}(X_k)\right\|_{L_2}^2 \geq
\phi^2\sum_{k\in S_j^\circ}
\|g_{kj}(X_k)\|_{L_2}^2 =\phi^2\upsilon
\sum_{k\in S_j^\circ}\sum_{r=1}^{R_n}a_{rkj}^2.
$$
Therefore, 
\begin{align*}
\begin{split}
    \mathbb{E}\left[ (a^\top h_j(X))^2 \right]
&= \mathbb{E}\left[ \left( a_0 + \sum_{k\in S_j^\circ} g_{kj}(X_k) \right)^2 \right]
= a_0^2 + \left\| \sum_{k\in S^\circ_j} g_{kj}(X_k) \right\|^2_{L^2}  \\
&\geq a_0^2 + \phi^2 \upsilon \sum_{k \in S^\circ_j} \sum_{r=1}^{R_n} a^2_{rkj} 
\geq \min\{1, \phi^2\upsilon\} \|a\|_2^2. 
\end{split}
\end{align*}
Since $a$ is arbitrary, we have $\min_{j\in[p]}\lambda_{\min}
(\mathbb E[h_j(X)h_j(X)^\top]) \geq \min\{1, \phi^2 \upsilon\} > 0$.

\end{proof}

\vspace{0.2in}
\begin{lemma}\label{lem:event:T5}
    
    Define the event 
    \begin{equation*}
    \mathcal{T}_5 := \left\{ \|Z_j (\beta^\star_j - \beta_j) \|^2_n \ge  \frac{c_0}{2}\|\beta_j^\star -  \beta_j\|_2^2 , \quad
     \text{$\forall \beta_j$ feasible in \eqref{eq:MIP}}, \quad \forall j\in[p]
    \right\}.
    \end{equation*}
    Then $\mathbb{P}(\mathcal{T}_5) \geq 1 - 1/p$.
\end{lemma}
\begin{proof}
    Denote the $i$-th row of $Z_j$ by $Z_j^{(i)}\in\mathbb{R}^{pR_n+1}$. 
    Let $\widetilde Z_j$ be the submatrix of $Z_j$ containing the intercept and the basis functions corresponding to $S_j^\circ$. Its dimension is $d_j=1+s_jR_n\le 1+KR_n=:d_n$. By Lemma~\ref{lem:min-eigen}, we have
    $\lambda_{\min}\left(\mathbb{E}[\widetilde Z_j^{(i)} \widetilde Z_j^{(i)\top}] \right)\geq c_0$ for some $c_0>0$. Note that all entries of $\widetilde Z_j$ are bounded in $[-1,1]$, and thus $0\preceq \widetilde{Z}_j^{(i)} \widetilde{Z}_j^{(i)\top}\preceq d_n I$.
    By the matrix Chernoff inequality, with probability at least $1 - e^{-t}/p$, we have
    $$
    \lambda_{\min}\left(n^{-1}\widetilde{Z}_j^\top\widetilde{Z}_j\right) \geq c_0  - \sqrt{\frac{2c_0 d_n(t + \log p + \log d_n)}{n}}.
    $$
    Taking $t = \log p$ and by Assumption~\ref{assump:bounded-Rn}, this concentration result implies that $\mathbb{P}(\mathcal{T}_{5,j}) \geq 1 - 1 / p^2$, where
    $$
    \mathcal{T}_{5,j} = \left\{
    \|Z_j (\beta_j^\star - \beta_j)\|^2_n  \ge \frac{c_0}{2}\left\| \beta_j^\star - \beta_j \right\|_2^2, \quad \forall\beta_j \text{ feasible}
    \right\}.
    $$
    Finally, the result follows using the fact that $\mathcal{T}_{5} = \cap_{j=1}^p \mathcal{T}_{5,j}$.

\end{proof}

\vspace{0.2in}
\subsection{Lemmas for the main results}
\label{sec:appendix:prelim-lemma}

For any $\vec{g}_j$, we introduce the empirical counterpart of the residual variance $\nu_j^n(\vec{g}_j)$, denoted by $\widehat{\nu}_j^n(\vec{g}_j):= \min_{f\in \mathcal{F}_n^{\oplus \vec{g}_j}} \int \mathcal{R}_j(f; \vec{g}_j) P_n(dx)$.

\vspace{0.2in}
\begin{lemma} \label{lem:Pn-2-P}
    Under event $\mathcal{T}_1 \cap \mathcal{T}_2$, we have 
    $$ 
    \left| \nu_j^n(\vec{g}_j) - \widehat{\nu}_j^n(\vec{g}_j) \right| \lesssim (1+s_j) \sqrt{\log p /n}.
    $$
\end{lemma}
\begin{proof}
    Note that we can rewrite
    $$
    \nu_j^n(\vec{g}_j) = \min_{f \in \mathcal{F}_n^{\oplus \vec{g}_j}} \int \mathcal{R}_j(f; \vec{g}_j) P(dx), \quad
    \widehat{\nu}_j^n(\vec{g}_j) = \min_{f \in \mathcal{F}_n^{\oplus \vec{g}_j}} \int \mathcal{R}_j(f; \vec{g}_j) P_n(dx).
    $$
    For any $f\in\mathcal F_n^{\oplus\vec g_j}$, write
$$
A_j(f):= \sum_{k\in\mathrm{pa}(j)}f^\star_{kj}(X_k) - \sum_{k\in\vec g_j}f_{kj}(X_k).
$$
Since $\mathcal R_j(f;\vec g_j)=(A_j(f)+\epsilon_j)^2$,
\begin{align*}
\left|\int\mathcal R_j(f;\vec g_j)(P-P_n)(dx)\right|
&\leq \left|\int A_j(f)^2(P-P_n)(dx)\right|\\
&\quad + 2\left|\int\epsilon_jA_j(f)(P-P_n)(dx)\right| +\left|\int\epsilon_j^2(P-P_n)(dx)\right|\\
&\lesssim (s_j^2+s_j+1)\sqrt{\log p/n}\\
&\lesssim (1+s_j)\sqrt{\log p/n}
\end{align*}
on $\mathcal T_1\cap\mathcal T_2$, uniformly over $f$. The last inequality uses $s_j\leq K$ with bounded $K$.

Now let 
$$
\widetilde{f}^{j} \in \argmin_{f \in \mathcal{F}_n^{\oplus \vec{g}_j} } \int \mathcal{R}_j(f; \vec{g}_j) P(dx), \quad
\bar{f}^{j} \in \argmin_{f \in \mathcal{F}_n^{\oplus \vec{g}_j} } \int \mathcal{R}_j(f; \vec{g}_j) P_n(dx),
$$
and then
$$
\nu^n_j(\vec{g}_j) = \int \mathcal{R}_j( \widetilde{f}^j; \vec{g}_j) P(dx) , \quad
\widehat{\nu}^n_j(\vec{g}_j) = \int \mathcal{R}_j( \bar{f}^j; \vec{g}_j) P_n(dx).
$$
For a sufficiently large constant $C>0$, we conclude that
$$
\int \mathcal{R}_j(\bar{f}^j; \vec{g}_j) P_n (dx)
\leq \int \mathcal{R}_j(\widetilde{f}^j; \vec{g}_j ) P_n(dx) 
\leq \int \mathcal{R}_j(\widetilde{f}^j; \vec{g}_j ) P(dx) +  C(1+s_j) \sqrt{\log p / n},
$$
and
$$
\int \mathcal{R}_j(\widetilde{f}^j; \vec{g}_j) P (dx)
\leq \int \mathcal{R}_j(\bar{f}^j; \vec{g}_j ) P(dx) 
\leq \int \mathcal{R}_j(\bar{f}^j; \vec{g}_j ) P_n(dx) +  C(1+s_j) \sqrt{\log p / n} .
$$
In summary,
$$
\left| \nu_j^n(\vec{g}_j) - \widehat{\nu}_j^n(\vec{g}_j) \right|  \lesssim  (1+s_j) \sqrt{\log p / n}.
$$
\end{proof}

\vspace{0.2in}
\begin{lemma}\label{lem:estimated-var-bound}
    Suppose conditions of Lemma \ref{lem:var-bound}, \ref{lem:small-gap} and \ref{lem:Pn-2-P} are satisfied. Then $\underline{c}  \leq {\nu}_j^n(\vec{g}_j) \leq \overline{c} + o(1)$ and $\underline{c} - o(1)\leq \widehat{\nu}_j^n(\vec{g}_j) \leq \overline{c} + o(1)$ for any $\vec{g}_j$. 
    
\end{lemma}
\begin{proof}
    First, by Lemma \ref{lem:var-bound}, we have 
    $
    \nu^n_j(\vec{g}_j) \geq \nu_j(\vec{g}_j) 
    \geq \nu_j(\vec{g}(\pi)_j)
    \geq \underline{\nu}_p 
    \geq \underline{c} ,
    $
    where $\vec{g}_j \subseteq \vec{g}(\pi)_j$. By Lemma \ref{lem:Pn-2-P}, we further conclude that $\widehat{\nu}_j^n(\vec{g}_j) \geq \nu^n_j(\vec{g}_j) -o(1) \geq \underline{c} - o(1)$. 
    
    Next, recall that $d_{n,p} :=  \max_{j\in\mathrm{V}} \max_{|\vec{g}_j| \in \{0, \dots, K\} }  \left| \nu_j(\vec{g}_j) - \nu^n_j(\vec{g}_j) \right|$. Then by Lemma \ref{lem:var-bound} and \ref{lem:small-gap}, we have
    $
    \nu^n_j(\vec{g}_j) \leq \nu_j(\vec{g}_j) + d_{n,p}  \leq \overline{\nu}_p + d_{n,p} \leq \overline{c} + o(1).
    $
    By Lemma \ref{lem:Pn-2-P}, we further conclude that $\widehat{\nu}_j^n (\vec{g}_j) \leq \nu_j^n(\vec{g}_j) + o(1) \leq \overline{c} + o(1)$.

\end{proof}

\vspace{0.2in}
\begin{lemma}\label{lem:lower-bound-for-contra}
    Suppose the conditions of Lemma \ref{lem:Pn-2-P} and Lemma \ref{lem:estimated-var-bound} are satisfied. For any $\pi \notin \Pi^\star$ and any $\vec{g}_j \subseteq \vec{g}(\pi)_j$, we have
    $$
    p^{-1} \sum_{j=1}^p \left[ \log \widehat{\nu}_j^n(\vec{g}_j) - \log \sigma^{\star2}_j  \right] \gtrsim \xi_p - C \sqrt{\log p /n},
    $$
where $C>0$ is a constant independent of $n$ and $p$.
\end{lemma}
\begin{proof}
    First, note that $\nu_j^n(\vec{g}_j) \geq \underline{c}$ by Lemma \ref{lem:estimated-var-bound}.
    Then
    \begin{align*}
\begin{split}
&\sum_{j=1}^p
\left[\log\widehat{\nu}_j^n(\vec{g}_j)-\log\sigma_j^{\star2}\right]
\geq \sum_{j=1}^p
\left[\log\nu_j^n(\vec{g}_j)-\log\sigma_j^{\star2}\right]
-Cp\sqrt{\log p/n}\\
\geq{}&
\sum_{j=1}^p
\left[\log\nu_j(\vec{g}(\pi)_j)-\log\sigma_j^{\star2}\right]
-Cp\sqrt{\log p/n}
\geq p\xi_p-Cp\sqrt{\log p/n}.
\end{split}
\end{align*}
The first inequality follows from Lemma~\ref{lem:Pn-2-P}, $\sum_{j=1}^p(1+s_j)\leq p(1+K)$, and $\log(a-x)\geq\log a-2x/a$ for $a>0$ and $0\leq x\leq a/2$, which applies for sufficiently large $n$. The second inequality follows from the fact that $\nu_j^n(\vec{g}_j)\geq\nu_j(\vec{g}_j) \geq\nu_j(\vec{g}(\pi)_j)$ for any $\vec{g}_j$. The last inequality follows from the definition of $\xi_p$.
\end{proof}

\vspace{0.2in}
\section{Proof of Theorem \ref{thm:correct-permutation}}
\label{sec:appendix:prof-correct-permutation}
\subsection{Proof of Theorem \ref{thm:correct-permutation} Part 1}
\begin{proof}
First, note that for any feasible $f_{kj}$ and $\vec{g}_{j}$, 
the optimization
$$
\min_{\sigma_j \geq 0} \quad
\sum_{j=1}^p \left[ \log \sigma_j^2 + \frac{ \left\| X_j - \sum_{k\in \vec{g}_j} {f}_{kj}(X_k) \right\|_n^2 + \gamma_n \left\| \beta_j^{-0} \right\|_2^2   }{ \sigma_j^2 } \right] + \lambda_n^2 \sum_{j=1}^p |\vec{g}_j|
$$
is solved at 
$
\sigma_j^2 = \left\| X_j - \sum_{k\in \vec{g}_j } f_{kj}(X_k) \right\|_n^2 + \gamma_n \|\beta_j^{-0}\|_2^2.
$
As a result, we can re-write the objective function as
\begin{align*}
\begin{split}
    \min_{\ \{f_{kj}\}, \ \{\vec{g}_j\} }\quad & \sum_{j=1}^p  \log \left( \left\| X_j - \sum_{k \in \vec{g}_j } {f}_{kj}(X_k)  \right\|_n^2 + \gamma_n \left\| \beta_j^{-0} \right\|_2^2 \right) + p + \lambda_n^2 \sum_{j=1}^p |\vec{g}_j|
\end{split}
\end{align*}
with the constraint that each $f_{kj} \in \mathcal{F}_{kj,n}$ and the selected edges form a DAG. Given the solution $\widehat{f}_{kj}$ and $\widehat{g}_j$, the estimated variance is thus $\widehat{\sigma}_j^2 = \|X_j - \sum_{k\in\widehat{g}_j} \widehat{f}_{kj}(X_k)\|_n^2 + \gamma_n \| \widehat{\beta}_j^{-0} \|_2^2 $. Recall that $D^\star_{kj} = \sum_{r=R_n+1}^{\infty} \beta^\star_{rkj}(b_{rkj}(X_k)-\mu_{rkj})$, where $\mu_{rkj}=\mathbb E[b_{rkj}(X_k)]$. Appealing to the basic inequality, 
$$
\sum_{j=1}^p \log \widehat{\sigma}_j^2 + p + \lambda_n^2 \widehat{s} \leq \sum_{j=1}^p \log \sigma_j^{\star2} + \sum_{j=1}^p \frac{\|\epsilon_j + \sum_{k\in \mathrm{pa}(j)} D^\star_{kj} \|_n^2 + \gamma_n \| \beta^{\star-0}_j \|_2^2 }{\sigma_j^{\star2}} + \lambda_n^2 s^\star ,
$$
consequently,
\begin{align} \label{eq:basic-inequality} 
\begin{split}
    \sum_{j=1}^p \log  \frac{\widehat{\sigma}_j^2}{ \sigma_j^{\star2} } \leq \sum_{j=1}^p  \frac{\|\epsilon_j\|_n^2 - \sigma_j^{\star2} }{ \sigma_j^{\star2} } + \sum_{j=1}^p \frac{ \| \sum_{k\in\mathrm{pa}(j)} D^\star_{kj} \|_n^2 }{\sigma_j^{\star2}} + \frac{2}{n}\sum_{j=1}^p \frac{ |\epsilon_j^\top \sum_{k\in\mathrm{pa}(j)} D^\star_{kj} | }{\sigma_j^{\star2}} + \gamma_n \sum_{j=1}^p \frac{ \|\beta_j^{\star-0}\|_2^2 }{\sigma_j^{\star 2}}  + \lambda_n^2 (s^\star - \widehat{s}).
\end{split}
\end{align}
Under event $\mathcal{T}_1 \cap \mathcal{T}_3$, Lemma~\ref{lem:bounded-beta}, Lemma~\ref{lem:bounded-d}, as well as Lemma~\ref{lem:var-bound}, we have
\begin{equation}\label{eq:log-diff-ubd}
   p^{-1} \sum_{j=1}^p \left[ \log \widehat{\sigma}_j^2 - \log \sigma_j^{\star2}  \right] \lesssim p^{-1} \left( \lambda_n^2  s_n + \gamma_n s_n  + p\sqrt{\log p /n} + s_n  R_n^{-2\eta} \right)  .
\end{equation}
Under Assumption \ref{assump:separation} and the choice of $\lambda_n^2$ and $\gamma_n$, the upper bound is $\xi_p o(1)$. This implies that $\widehat{\Pi} \subseteq \Pi^\star$. Otherwise, there must exist $\widehat{\pi} \in \widehat{\Pi}$ such that $\widehat{\pi} \notin \Pi^\star$. Since $\widehat{\sigma}_j^2 \geq \widehat{\nu}_j^n (\widehat{g}_j)$, appealing to Lemma \ref{lem:lower-bound-for-contra}, we have
\begin{equation*}
   p^{-1} \sum_{j=1}^p \left[ \log \widehat{\sigma}_j^2 - \log \sigma_j^{\star2}  \right] \geq \xi_p-C\sqrt{\log p/n} .
\end{equation*}
This lower bound, however, is $\xi_p(1 + o(1))$, contradicting with the upper bound $\xi_p  o(1)$ when $\xi_p > 0$. 

Finally, note that this proof relies on the event $\mathcal{T}_1 \cap \mathcal{T}_2 \cap \mathcal{T}_3 $, which holds with probability at least $1 - 4/p$.

\end{proof}

\vspace{0.2in}
\subsection{Proof of Theorem \ref{thm:correct-permutation} Part 2} \label{sec:appendix:prof-correct-permutation-2}
\begin{proof}
Denote the loss function with the ridge penalty:
$$
\widetilde{\ell}_n(\widetilde{\theta}) := \ell_n(\widetilde{\theta}) + \gamma_n \sum_{j=1}^p \frac{ \|\beta_j^{-0}\|_2^2 }{\sigma_j^2} .
$$
The basic inequality for the early stopping solution $\widehat{\theta}^{\rm early}$ can be written as: 
\begin{equation*}
    \widetilde{\ell}_n(\widehat{\theta}^{\rm early}) + \lambda_n^2 \widehat{s}^{\rm early} \leq \widetilde{\ell}_n(\widehat{\theta}) + \lambda_n^2 \widehat{s} + \tau^{\rm early} \leq \widetilde{\ell}_n\left( \{ \beta^{\star}_{rkj} \}_{(k,j)\in E^\star,\ r\in[R_n]},\ \{\sigma_j^\star\}_{j\in \mathrm{V}}  \right) + \lambda_n^2 s^\star + \tau^{\rm early},
\end{equation*}
We use the same proof technique as in Part 1.
When the early stopping threshold satisfies $\tau^{\rm early} = o(p\xi_p)$, the upper bound of $p^{-1} \sum_{j=1}^p [\log (\widehat{\sigma}_j^{\rm early})^2 - \log \sigma_j^{\star2}]$, similar as in \eqref{eq:log-diff-ubd}, remains to be $\xi_p o(1)$. Moreover, the estimated variance from the early-stopping solution satisfies $(\widehat{\sigma}_j^{\rm early})^2 \geq \widehat{\nu}^n_j(\widehat{g}_j^{\rm early}) $, so when $\widehat{\pi}^{\rm early} \notin \Pi^\star$, the lower bound of $p^{-1} \sum_{j=1}^p [\log (\widehat{\sigma}_j^{\rm early})^2 - \log \sigma_j^{\star2}]$ is at least $\xi_p(1-o(1))$. With these two key facts in place, the early-stopping permutations are guaranteed to be correct.
    
\end{proof}

\vspace{0.2in}
\section{Proof of Theorem \ref{thm:var-converge}}
\label{sec:appendix:prof-var-converge}

\vspace{0.2in}
\subsection{Proof of Theorem \ref{thm:var-converge} Part 1}
\begin{proof}
We first note that, by Theorem \ref{thm:correct-permutation}, the estimated permutations are correct: $\widehat{\Pi} \subseteq \Pi^\star$. 
Recall that $\widehat{\sigma}_j^2 =\|X_j-Z_j\widehat{\beta}_j\|_n^2 +\gamma_n\|\widehat{\beta}_j^{-0}\|_2^2$.
Applying the inequality 
$\log (1 +x) \leq x - x^2 / (2(1 + c)^2)$, $-1 < x \leq c$, with $c\ge 0$, we obtain
\begin{align} \label{eq:interm0}
\begin{split}
    \log \left( \frac{\widehat{\sigma}_j^2}{\sigma_j^{\star2}} \right)
&= -\log \left( \frac{\sigma_j^{\star2}}{\widehat{\sigma}_j^2} \right) 
\geq \frac{ \widehat{\sigma}_j^2 - \sigma_j^{\star2} }{\widehat{\sigma}_j^2} + \frac{ (\overline{\sigma}^2\wedge
\min_{j\in\mathrm{V}}\widehat{\sigma}_j^2)^2 }{ 2\overline{\sigma}^4 } \left( \frac{ \sigma_j^{\star2} - \widehat{\sigma}_j^2 }{\widehat{\sigma}_j^2} \right)^2  .
\end{split}
\end{align}
Moreover, by \eqref{eq:basic-inequality} and \eqref{eq:interm0}, we have
\begin{equation} \label{eq:interm1}
    K_0\sum_{j=1}^p \left( \frac{ \sigma_j^{\star2} - \widehat{\sigma}_j^2 }{\widehat{\sigma}_j^2} \right)^2 
    \leq \sum_{j=1}^p \left( \frac{\|\epsilon_j\|_n^2}{ \sigma_j^{\star2} } - 1 \right) - \sum_{j=1}^p \frac{\widehat{\sigma}_j^2 - \sigma^{\star2}_j}{ \widehat{\sigma}_j^2 } + \sum_{j=1}^p(\mathcal{D}_{1j} + \mathcal{D}_{2j}) + \gamma_n \sum_{j=1}^p \frac{\|\beta^{\star-0}_j\|_2^2}{\sigma_j^{\star 2}} + \lambda_n^2(s^\star - \widehat{s}),
\end{equation}
where 
$$
\mathcal{D}_{1j} = \frac{ \| \sum_{k\in\mathrm{pa}(j)} D^\star_{kj} \|_n^2 }{\sigma_j^{\star2}}, \quad
\mathcal{D}_{2j} = \frac{2}{n} \frac{ \sum_{k\in\mathrm{pa}(j)} |\epsilon_j^\top  D^\star_{kj} | }{\sigma_j^{\star2}} ,
$$
and $K_0 > 0$ is a constant such that $0<K_0 \leq (\overline{\sigma}^2\wedge\min_{j\in\mathrm{V}}
\widehat{\sigma}_j^2)^2/(2\overline{\sigma}^4)$. Such a constant $K_0$ exists since by Lemma \ref{lem:estimated-var-bound}, we have $\widehat{\sigma}_j^2 \geq \widehat{\nu}_j^n(\widehat{g}_j) \geq \underline{c} - o(1)$ for all $j\in \mathrm{V}$.
In the following, we use the parameterization with $f_{kj}(\cdot) = \beta_{0kj} + \sum_{r=1}^{R_n} \beta_{rkj} b_{rkj}(\cdot)$. Recall the definition 
$$
Z_j=\left[1, b_{11j}(X_1), \dots, b_{1pj}(X_p), \dots, b_{R_n 1j}(X_1), \dots, b_{R_n pj}(X_p)  \right] \in \mathbb{R}^{p R_n + 1}, \quad \text{(as a random vector)}
$$
and
$$
\beta_j = \left[\beta_{0*j}, \beta_{11j}, \dots, \beta_{1pj}, \dots, \beta_{R_n1j}, \dots, \beta_{R_n pj} \right]^\top \in \mathbb{R}^{pR_n + 1}, \quad \beta_{0\star j} = \sum_{k=1}^p \beta_{0kj}.
$$
Then we have $\sum_{k=1}^p f_{kj}(X_k) = \beta_j^\top Z_j $. For each $j\in \mathrm{V}$,
\begin{align*}
\begin{split}
&\widehat{\sigma}_j^2 - \sigma^{\star2}_j 
= \left\| X_j -Z_j \widehat{\beta}_j \right\|_n^2 +\gamma_n\|\widehat{\beta}_j^{-0}\|_2^2 - \sigma_j^{\star2} 
=\left\| Z_j (\beta^\star_j - \widehat{\beta}_j) + \epsilon_j + \sum_{k\in\mathrm{pa}(j)}D^\star_{kj} \right\|_n^2 +\gamma_n\|\widehat{\beta}_j^{-0}\|_2^2 - \sigma^{\star2}_j \\
={}& \left\| Z_j (\beta_j^\star - \widehat{\beta}_j) + \sum_{k\in \mathrm{pa}(j)} D^\star_{kj} \right\|_n^2 + \left( \|\epsilon_j\|_n^2 - \sigma^{\star2}_j \right)  +\frac{2}{n} \epsilon_j^\top (Z_j(\beta^\star_j - \widehat{\beta}_j)) +  2 \frac{\epsilon_j^\top}{n} \sum_{k\in \mathrm{pa}(j)} D^\star_{kj} +\gamma_n\|\widehat{\beta}_j^{-0}\|_2^2
\end{split}
\end{align*}
Under the event $\mathcal{T}_4$, we obtain that for some $\delta_3 > 0$ with $0 < \delta_2 + \delta_3 < 1 $,
\begin{align} \label{eq:interm2} 
\begin{split}
    &\left( \|\epsilon_j\|_n^2 - \sigma^{\star2}_j  \right) - \left( \widehat{\sigma}_j^2  - \sigma^{\star2}_j \right)   \\
    \le {}& -\left\| Z_j (\beta_j^\star - \widehat{\beta}_j) + \sum_{k\in \mathrm{pa}(j)} D^\star_{kj} \right\|_n^2 + \delta_2 \left\|  Z_j(\beta^\star_j - \widehat{\beta}_j)\right\|_n^2  + \Delta_{n,p} + \sigma^{\star2}_j \mathcal{D}_{2j}  \\
    \le{}& -(1-\delta_2 - \delta_3) \left\| Z_j (\beta_j^\star - \widehat{\beta}_j) \right\|_n^2 - (1 - 1/\delta_3) \left\| \sum_{k\in \mathrm{pa}(j)} D^\star_{kj} \right\|_n^2  + \Delta_{n,p} + \sigma_j^{\star2} \mathcal{D}_{2j}   \\
    \le{}&  (1/\delta_3 - 1) \sigma_j^{\star2} \mathcal{D}_{1j}  + \sigma_j^{\star2}\mathcal{D}_{2j} + \Delta_{n,p}.
\end{split}
\end{align}
Furthermore, for some $0 <\delta_4 < K_0 $,
\begin{align} \label{eq:interm3}
\begin{split}
  \frac{ \|\epsilon_j\|_n^2 - \sigma^{\star2}_j }{\sigma^{\star2}_j} - \frac{ \|\epsilon_j\|_n^2 - \sigma^{\star2}_j }{ \widehat{\sigma}_j^2}
= \frac{ \|\epsilon_j\|_n^2 - \sigma_j^{\star2} }{ \sigma^{\star2}_j } \cdot \frac{ \widehat{\sigma}_j^2 - \sigma_j^{\star2} }{ \widehat{\sigma}_j^2 } 
\leq \frac{1}{\delta_4} \left( \frac{ \|\epsilon_j\|_n^2 - \sigma_j^{\star2} }{ \sigma_j^{\star2} } \right)^2 +\delta_4 \left( \frac{ \widehat{\sigma}_j^2 - \sigma_j^{\star2} }{ \widehat{\sigma}_j^2 } \right)^2. 
\end{split}
\end{align}
The $\Delta_{n,p}$ contribution is needed only for nodes having a true or selected parent; otherwise $\beta_j^\star=\widehat{\beta}_j=0$ and the corresponding noise cross term vanishes. There are at most $s^\star+\widehat{s}\leq2s_n$ such nodes. Since $s_j\leq K$ with bounded $K$, their total contribution is $\lesssim s_nR_n\log p/n$. Now we plug \eqref{eq:interm2} and \eqref{eq:interm3} back into \eqref{eq:interm1}. Under event $\mathcal{T}_1 \cap \mathcal{T}_2 \cap \mathcal{T}_3 \cap \mathcal{T}_4$, Lemma~\ref{lem:bounded-beta}, Lemma \ref{lem:bounded-d} and Lemma~\ref{lem:var-bound}, we obtain
\begin{align*}
\begin{split}
    &(K_0 - \delta_4) \sum_{j=1}^p \left( \frac{ \sigma_j^{\star2} - \widehat{\sigma}_j^2 }{\widehat{\sigma}_j^2} \right)^2 \\
    \leq{}&  \sum_{j=1}^p \left[ \frac{ \|\epsilon_j\|_n^2 - \sigma^{\star2}_j }{\sigma^{\star2}_j} -  \frac{ \|\epsilon_j\|_n^2 - \sigma^{\star2}_j }{ \widehat{\sigma}_j^2} \right]+ \sum_{j=1}^p \left[ \frac{ \|\epsilon_j\|_n^2 - \sigma^{\star2}_j }{ \widehat{\sigma}_j^2} - \frac{ \widehat{\sigma}_j^2 - \sigma^{\star2}_j }{ \widehat{\sigma}_j^2} \right] - \delta_4 \sum_{j=1}^p \left( \frac{ \sigma_j^{\star2} - \widehat{\sigma}_j^2 }{\widehat{\sigma}_j^2} \right)^2 \\ 
    &+ \sum_{j=1}^p(\mathcal{D}_{1j} + \mathcal{D}_{2j}) + \gamma_n \sum_{j=1}^p \frac{\|\beta_j^{\star-0}\|_2^2}{\sigma_j^{\star2}} + \lambda_n^2(s^\star - \widehat{s}) \\
    \lesssim{}& (p + s_n R_n) \log p / n  + s_n R_n^{-2\eta} + \gamma_n s_n + \lambda_n^2 s_n.
\end{split}
\end{align*}
Since zero coefficients are feasible in the selected model, $\widehat{\sigma}_j^2\leq\|X_j\|_n^2\lesssim1$
on $\mathcal T_1$, using bounded generating functions and bounded $K$. Hence the preceding relative-error bound also bounds $\sum_j(\sigma_j^{\star2}-\widehat{\sigma}_j^2)^2$ up to a constant.
The result hence follows. Note that this proof relies on the event $\mathcal{T}_1 \cap \mathcal{T}_2 \cap \mathcal{T}_3 \cap \mathcal{T}_4 $, which holds with probability at least $1-5/p$.
\end{proof}

\vspace{0.2in}
\subsection{Proof of Theorem \ref{thm:var-converge} Part 2}
\begin{proof}
    Note that the estimated permutations are correct: $\widehat{\Pi}^{\rm early} \subseteq \Pi^\star$ by Theorem \ref{thm:correct-permutation}.
    We define the ridge-penalized loss function $\widetilde{\ell}_n(\widetilde{\theta})$ as in Appendix~\ref{sec:appendix:prof-correct-permutation-2}. 
    Then using the same proof technique as in Part 1, this is a direct result from the basic inequality 
    \begin{equation*}
    \widetilde{\ell}_n(\widehat{\theta}^{\rm early}) + \lambda_n^2 \widehat{s}^{\rm early} \leq \widetilde{\ell}_n(\widehat{\theta}) + \lambda_n^2 \widehat{s} + \tau^{\rm early} \leq \widetilde{\ell}_n\left( \{ \beta^{\star}_{rkj} \}_{(k,j)\in E^\star,\ r\in[R_n]},\ \{\sigma_j^\star\} \right) + \lambda_n^2 s^\star + \tau^{\rm early}.
\end{equation*}
\end{proof}

\vspace{0.2in}
\section{Proof of Theorem \ref{thm:recovery}}
\label{sec:appendix:prof-recovery}

\subsection{Proof of Theorem \ref{thm:recovery} Part 1}
\begin{proof}
    In this proof, we keep using the parameterization with $Z_j \in \mathbb{R}^{pR_n + 1}$ and $\beta_j\in\mathbb{R}^{pR_n+1}$. Let $\mathcal{G}(\beta)$ be the graph induced by $\beta$, where an edge between node $k$ and node $j$ is present if $\sum_{r=1}^{R_n} |\beta_{rkj}|^2 \neq 0$, and absent otherwise. 
    We aim to show that the solution to the following optimization problem, which is equivalent to the original one, recovers the true graph $\mathcal{G}^\star$:
    \begin{equation}\label{eq:obj-beta} 
        \min_{ \beta \in \mathfrak{B} } l(\beta),  \quad l(\beta):=
    \sum_{j=1}^p  \log \frac{\left\| X_j - Z_j \beta_j \right\|_n^2 + \gamma_n \|\beta_j^{-0}\|_2^2 }{ \sigma_j^{\star2} }  - \sum_{j=1}^p \frac{ \|\epsilon_j\|_n^2 - \sigma^{\star2}_j }{\sigma^{\star2}_j}  +\lambda_n^2\left(\sum_{(k,j)\in\mathcal E} \mathds{1}\{\|\beta_{kj}\|_2\neq0\}-s^\star\right),
    \end{equation}
    where $\mathfrak{B}$ corresponds to the feasible set in \eqref{eq:MIP}.
    The proof proceeds in three steps. In each step, we solve an optimization problem that is more constrained than the original one.
\begin{itemize}
    \vspace{0.05in}
    \item In step I, we restrict the feasible region to $\mathcal{R}_1 :=\{\beta \in \mathfrak{B}: \mathcal{G}(\beta) = \mathcal{G}^\star\}$, and derive an upper bound for the minimized objective function over $\beta \in \mathcal{R}_1$. We denote this upper bound by $l_1$.

    \vspace{0.05in}
    \item In step II, we restrict the feasible region to 
    $$\mathcal{R}_2 := \{\beta \in \mathfrak{B}: |\mathcal{G}(\beta)| \leq s^\star,\ \mathcal{G}(\beta) \neq \mathcal{G^\star},\ \text{$\mathcal{G}(\beta)$ has correct permutations}\},$$
    and derive a lower bound for the minimized objective function over $\beta \in \mathcal{R}_2$. We denote this lower bound by $l_2$, and verify that $l_1 < l_2$ when $n$ is sufficiently large.

    \vspace{0.05in}
    \item In step III, we restrict the feasible region to 
    $$\mathcal{R}_3 := \{\beta \in \mathfrak{B}: |\mathcal{G}(\beta)| > s^\star, \text{$\mathcal{G}(\beta)$ has correct permutations}\},$$
    and derive a lower bound for the minimized objective function over $\beta \in \mathcal{R}_3$. We denote this lower bound by $l_3$, and verify that $l_1 < l_3$ when $n$ is sufficiently large.
\end{itemize}

\vspace{0.05in}
The three steps above guarantee that the solution to the original problem falls inside $\mathcal{R}_1$ when $n$ is sufficiently large. Note that we restrict the search only to the space of correct permutations, since, by Theorem~\ref{thm:correct-permutation}, any feasible value associated with an incorrect permutation cannot be optimal. This proof relies on the event $\mathcal{T}_1 \cap \mathcal{T}_2 \cap \mathcal{T}_3 \cap \mathcal{T}_4 \cap \mathcal{T}_5$, which holds with probability at least $1 - 6/p$.

\vspace{0.2in}
\noindent\textbf{\textit{Step I.}} In this step, we use $\bar{\beta}$ to denote the solution to optimization problem \eqref{eq:obj-beta} over the additional constraint set $\mathcal{R}_1$. The estimated variance is hence $\bar{\sigma}_j^2:= \| X_j - Z_j \bar{\beta}_j \|_n^2 + \gamma_n \|\bar\beta^{-0}_j\|_2^2$. Appealing to the basic inequality similar as \eqref{eq:basic-inequality}, we can arrive at an upper bound for the objective function:
\begin{align} \label{eq:basic-inequality2} 
\begin{split}
    l(\bar{\beta})&=\sum_{j=1}^p \log  \frac{\bar{\sigma}_j^2}{ \sigma_j^{\star2} }  - \sum_{j=1}^p  \frac{\|\epsilon_j\|_n^2 - \sigma_j^{\star2} }{ \sigma_j^{\star2} } 
   \leq \sum_{j=1}^p \frac{ \| \sum_{k\in\mathrm{pa}(j)} D^\star_{kj} \|_n^2 }{\sigma_j^{\star2}} + \frac{2}{n}\sum_{j=1}^p \frac{ |\epsilon_j^\top \sum_{k\in\mathrm{pa}(j)} D^\star_{kj} | }{\sigma_j^{\star2}} + \gamma_n \sum_{j=1}^p \frac{\|\beta^{\star-0}_j\|_2^2}{\sigma_j^{\star2}}  \\
   &\lesssim s_n R_n^{-2\eta} + s_n \log p /n + \gamma_n s_n.
\end{split}
\end{align}
We call the upper bound $l_1: = s_n R_n^{-2\eta} + s_n \log p /n + \gamma_n s_n$. By Assumption \ref{assump:beta-min} and the choice of $\gamma_n$, we have $l_1 =o(\delta_{n,p} ) $.

\vspace{0.2in}
\noindent\textbf{\textit{Step II.}} In this step, we overload the notation $\bar{\beta}$ to represent the solution to optimization problem \eqref{eq:obj-beta} over the additional constraint set $\mathcal{R}_2$. The estimated variance is $\bar{\sigma}_j^2:= \| X_j - Z_j\bar{\beta}_j \|_n^2 +\gamma_n\|\bar{\beta}_j^{-0}\|_2^2$, and the estimated sparsity is $\bar{s}:=|\mathcal G(\bar{\beta})|$. 
We also define $\widetilde{\sigma}_j^2 := \|X_j - Z_j \bar{\beta}_j \|_n^2$. Fix $V>8K^2C^2+2\overline\sigma^2$.
Since $\bar\beta$ minimizes the objective, comparing it with fits obtained by shrinking node $j$'s coefficients towards zero gives
$\widetilde\sigma_j^2\leq\bar\sigma_j^2\leq\|X_j\|_n^2$. Since the true regression function is bounded by $2KC$,
on $\mathcal T_1$ we have
$$
\widetilde\sigma_j^2 \leq\|X_j\|_n^2 \leq8K^2C^2+2\|\epsilon_j\|_n^2 \leq8K^2C^2+2\overline\sigma^2+o(1) \leq V
$$
uniformly over $j$, for sufficiently large $n$.
Then by \eqref{eq:interm0}, \eqref{eq:interm2} and \eqref{eq:interm3},
\begin{align*}
\begin{split}
    l(\bar{\beta}) 
    &= \sum_{j=1}^p  \log \frac{\bar{\sigma}_j^2}{\sigma^{\star2}_j } - \sum_{j=1}^p \frac{ \|\epsilon_j\|_n^2 - \sigma_j^{\star2} }{ \sigma_j^{\star2}} +  \lambda_n^2 (\bar{s} - s^\star)
    \geq \sum_{j=1}^p  \log \frac{\widetilde{\sigma}_j^2}{\sigma^{\star2}_j } - \sum_{j=1}^p \frac{ \|\epsilon_j\|_n^2 - \sigma_j^{\star2} }{ \sigma_j^{\star2}} +  \lambda_n^2 (\bar{s} - s^\star)  \\
    &\geq \sum_{j=1}^p \left[ \frac{ \widetilde{\sigma}_j^2 - \sigma_j^{\star2} }{ \widetilde{\sigma}_j^2 } - \frac{ \|\epsilon_j\|_n^2 - \sigma_j^{\star2} }{\widetilde{\sigma}_j^2} \right] + \sum_{j=1}^p \left[ \frac{ \|\epsilon_j\|_n^2 - \sigma_j^{\star2} }{\widetilde{\sigma}_j^2} - \frac{ \|\epsilon_j\|_n^2 - \sigma_j^{\star2} }{ \sigma^{\star2}_j} \right] + K_0 \sum_{j=1}^p \left( \frac{ \sigma_j^{\star2} - \widetilde{\sigma}_j^2 }{ \widetilde{\sigma}_j^2 } \right)^2 + \lambda_n^2 (\bar{s} - s^\star) \\
    &\geq \frac{1}{2V}\sum_{j=1}^p \left\|Z_j(\beta_j^\star-\bar\beta_j)\right\|_n^2 -C_1\left[(p+s_nR_n)\log p/n+s_nR_n^{-2\eta}\right] +\lambda_n^2(\bar s-s^\star).
\end{split}
\end{align*}
Here $C_1>0$ is a fixed constant, and we take
$\delta_2=\delta_3=1/4$ and $0<\delta_4<K_0$
in \eqref{eq:interm2}--\eqref{eq:interm3}. Here $K_0$ is chosen so that \eqref{eq:interm0} holds with $\widehat\sigma_j^2$ replaced by $\widetilde\sigma_j^2$, using the residual lower bound from Lemma~\ref{lem:estimated-var-bound}.

Now we attempt to connect this lower bound with the number of missed edges $\gamma_1:= |\mathcal{G}^\star \setminus \mathcal{G}(\bar{\beta}) |$. We also define $\gamma_2:=| \mathcal{G}(\bar{\beta}) \setminus \mathcal{G}^\star |$. For any $\bar{\beta} \in \mathcal{R}_2$, we consider two cases. 

\underline{\textit{Case 1:}} If $|\mathcal{G}(\bar{\beta})| = s^\star$, then $\gamma_1 = \gamma_2 \geq 1$. Under the event $\mathcal{T}_5$, we have
\begin{equation*}
\sum_{j=1}^p  \left\| Z_j (\beta_j^\star - \bar{\beta}_j) \right\|_n^2 
\gtrsim \sum_{j=1}^p \left\|\beta_j^\star - \bar{\beta}_j \right\|_2^2 
\ge \gamma_1  \min_{(k,j) \in E^\star} \sum_{r=1}^{R_n} |\beta^\star_{rkj}|^2 
\geq \min_{(k,j) \in E^\star} \sum_{r=1}^{R_n} |\beta^\star_{rkj}|^2.
\end{equation*}
Then 
\begin{equation*}
    l(\bar{\beta}) \gtrsim \min_{(k,j) \in E^\star} \sum_{r=1}^{R_n} |\beta^\star_{rkj}|^2   - (p + s_n R_n) \log p / n  - s_n R_n^{-2\eta}.
\end{equation*}
We call this lower bound $l_2^{(1)}$. By Assumption \ref{assump:beta-min}, we have $l_2^{(1)} \gtrsim \delta_{n,p} > o(\delta_{n,p}) = l_1$ when $n$ is sufficiently large.

\underline{\textit{Case 2:}} If $|\mathcal{G}(\bar{\beta})| < s^\star$, then $\gamma_2 \geq 0$ and $\gamma_1 = s^\star - \bar{s} + \gamma_2$. Under the event $\mathcal{T}_5$, we have
    \begin{align*}
    \begin{split}
        \sum_{j=1}^p  \left\| Z_j (\beta_j^\star - \bar{\beta}_j) \right\|_n^2 
    &\ge \frac{c_0}{2}\sum_{j=1}^p \left\|\beta_j^\star - \bar{\beta}_j \right\|_2^2 
    \ge \frac{c_0}{2}\gamma_1 \min_{(k,j) \in E^\star} \sum_{r=1}^{R_n} |\beta^\star_{rkj}|^2  
    \geq \frac{c_0}{2}(s^\star - \bar{s})  \min_{(k,j) \in E^\star} \sum_{r=1}^{R_n} |\beta^\star_{rkj}|^2.
    \end{split}
    \end{align*}
    Then 
    \begin{align*}
    l(\bar\beta) \geq (s^\star-\bar s) \left[ \frac{c_0}{4V} \min_{(k,j)\in E^\star} \sum_{r=1}^{R_n}|\beta^\star_{rkj}|^2 -\lambda_n^2
    \right]
    -C_1\left[(p+s_nR_n)\log p/n+s_nR_n^{-2\eta}\right].
    \end{align*}
    We call this lower bound $l_2^{(2)}$. By Assumption~\ref{assump:beta-min} and $\lambda_n^2\leq c_u\delta_{n,p}$, taking $c_u<c_0c_\beta/(4V)$, gives for sufficiently large $n$,
    \begin{align*}
    l_2^{(2)} \geq (s^\star-\bar s) \left(\frac{c_0c_\beta}{4V}-c_u\right)\delta_{n,p} -C_1\left[(p+s_nR_n)\log p/n+s_nR_n^{-2\eta}\right] \gtrsim\delta_{n,p}>o(\delta_{n,p})=l_1 \, .
    \end{align*}
    Both cases lead to the result that $l_2 =\min\{l_2^{(1)}, l_2^{(2)}\} > l_1$ when $n$ is sufficiently large.

\vspace{0.2in}
\noindent\textbf{\textit{Step III.}} In this step, we overload the notation $\bar{\beta}$ to represent the solution to optimization problem \eqref{eq:obj-beta} over the additional constraint set $\mathcal{R}_3$. The estimated sparsity is also overloaded as $\bar{s}:=|\mathcal{G}(\bar{\beta})|$. Similarly as in Step II, we obtain a lower bound for $l(\bar{\beta})$:
\begin{align*}
    l(\bar{\beta}) &\gtrsim \sum_{j=1}^p  \left\| Z_j (\beta_j^\star - \bar{\beta}_j) \right\|_n^2  +  \lambda_n^2 (\bar{s} - s^\star) - (p + s_n R_n) \log p / n  - s_n R_n^{-2\eta} \\
    &\geq \lambda_n^2 - (p + s_n R_n) \log p / n  - s_n R_n^{-2\eta}
\end{align*}
We call this lower bound $l_3$. By the choice of $\lambda_n^2\gtrsim \delta_{n,p}$, we verify that $l_3 \gtrsim \delta_{n,p} > o(\delta_{n,p}) = l_1$ when $n$ is sufficiently large.
\end{proof}

\vspace{0.2in}
\subsection{Proof of Theorem \ref{thm:recovery} Part 2}
\begin{proof}
Note that the estimated permutations are correct: $\widehat{\Pi}^{\rm early} \subseteq \Pi^\star$ by Theorem \ref{thm:correct-permutation}.
Then the basic inequality for the early stopping solution $\widehat{\theta}^{\rm early}$ can be written as: 
\begin{equation*}
    l(\widehat{\beta}^{\rm early}) \leq l(\widehat{\beta}) + \tau^{\rm early} \leq l\left( \{ \beta_1^\star, \dots, \beta_p^\star  \} \right) + \tau^{\rm early},
\end{equation*}
We use the same proof technique as in Part 1. When the early stopping threshold satisfies $\tau^{\rm early} = o(p\xi_p \wedge \delta_{n,p})  $, the upper bound becomes $l_1+\tau^{\rm early} =o(\delta_{n,p})$. Moreover, the lower bounds $l_2$ and $l_3$ remain at least on the order of $\delta_{n,p}$. The result hence follows.

\end{proof}

\section{Additional discussion for interaction effects}
\label{sec:appendix:interaction}
\subsection{Arriving at MIP \texorpdfstring{\eqref{eqn:interaction_MIP}}{(10)} }
Recall we had the estimator
\begin{eqnarray*}
\begin{aligned}
	\min_{(\beta,\zeta) \in \tilde{\Theta},\sigma}&\quad \sum_{j=1}^p
	\left\{\log(\sigma_j^2)+ \frac{\left\|	X_j - Z_j\beta_j
	-\Omega_j\zeta_j\right\|_n^2 + \gamma_{n,M}\|\beta^{-0}_{j}\|_2^2+\gamma_{n,I}\|\zeta_j\|_2^2}{\sigma_j^2}\right\}\\&\quad+ \lambda_{n,M}^2 \sum_{(k,j)\in \mathcal{E}} \mathds{1}\left\{ \|\beta_{kj}\|_2 \neq 0 \right\}+ \lambda_{n,I}^2\sum_{(k,l,j)\in\tilde{\mathcal{E}}} \mathds{1}\left\{ \|\zeta_{\cdot\cdot{klj}}\|_2 \neq 0 \right\}.
    \end{aligned}
\end{eqnarray*}
Minimizing over the noise variances and dropping an additive constant gives:
\begin{eqnarray*}
\begin{aligned}
	\min_{(\beta,\zeta) \in \tilde{\Theta}}&\quad \sum_{j=1}^p
	\log\left(\left\|	X_j - Z_j\beta_j
	-\Omega_j\zeta_j\right\|_n^2 + \gamma_{n,M}\|\beta^{-0}_{j}\|_2^2+\gamma_{n,I}\|\zeta_j\|_2^2\right)\\&\quad+ \lambda_{n,M}^2 \sum_{(k,j)\in \mathcal{E}} \mathds{1}\left\{ \|\beta_{kj}\|_2 \neq 0 \right\}+ \lambda_{n,I}^2\sum_{(k,l,j)\in\tilde{\mathcal{E}}} \mathds{1}\left\{ \|\zeta_{\cdot\cdot{klj}}\|_2 \neq 0 \right\}.
    \end{aligned}
\end{eqnarray*}
With the binary variables described in Section~\ref{sec:interactions}, we arrive at the following MIP:
\begin{equation}
\begin{aligned}
&\smash{\min_{\substack{
\beta_j \in \mathbb{R}^{R_n p + 1}\\
\zeta_j \in \mathbb{R}^{R_n^2p(p-1)/2}\\
g_{kj}\in\{0,1\},\ \forall (k,j)\in\mathcal{E}\\
w_{klj}\in\{0,1\},\ \forall (k,l,j)\in\tilde{\mathcal{E}}\\
\psi\in[1,p]^p}}}
&&
\sum_{j=1}^p
\log\left(
\left\|X_j-Z_j\beta_j-\Omega_j\zeta_j\right\|_n^2
+\gamma_{n,M}\|\beta^{-0}_j\|_2^2
+\gamma_{n,I}\|\zeta_j\|_2^2
\right)
\\
&&&
+\lambda_{n,M}^2\sum_{(k,j)\in\mathcal E}g_{kj}
+\lambda_{n,I}^2\sum_{(k,l,j)\in\tilde{\mathcal E}}w_{klj},
\\
&\text{subject to}
&&
g_{kj}=0 \Rightarrow \beta_{kj}=0,
\qquad \forall (k,j)\in\mathcal E,
\\
&&&
w_{klj}=0 \Rightarrow \zeta_{\cdot\cdot klj}=0,
\qquad \forall (k,l,j)\in\tilde{\mathcal E},
\\
&&&
w_{klj}\leq g_{kj},\qquad w_{klj}\leq g_{lj},
\qquad \forall (k,l,j)\in\tilde{\mathcal{E}},
\\
&&&
1-p+pg_{kj}\leq \psi_j-\psi_k,
\qquad \forall (k,j)\in\mathcal E.
\end{aligned}
\label{eq:interaction_intermediate}
\end{equation}
Following arguments similar to those in Proposition~\ref{lem:convexity} and Section~\ref{sec:appendix:cutting-plane}, we minimize over
the unpenalized intercepts and the coefficient vectors to obtain the convex MIP in~\eqref{eqn:interaction_MIP}, where
$$
\tilde{f}(g,w):= \sum_{j=1}^p \log\left(X_j^\top H A_j(g,w)^{-1}H X_j\right),
$$
with $H=I_n-\mathbf{1}\mathbf{1}^\top/n$ and
$$
A_j(g,w):=nI_n +\gamma_{n,M}^{-1} \sum_{(k,j)\in\mathcal{E}} g_{kj}H Z_{j|k}Z_{j|k}^\top H +\gamma_{n,I}^{-1} \sum_{(k,l,j)\in \tilde{\mathcal{E}}} w_{klj}H\Omega_{j|kl}\Omega_{j|kl}^\top H.
$$
Here, $Z_{j|k}=[b_{rkj}(X_k)]_{r\in[R_n]} \in\mathbb{R}^{n\times R_n}$ and $\Omega_{j|kl} =[b_{r_1kj}(X_k)b_{r_2lj}(X_l)]_{r_1,r_2\in[R_n]} \in\mathbb{R}^{n\times R_n^2}$.
Multiplication by $H$ centers the responses and design columns, thereby eliminating the unpenalized intercepts. In the implementation, the centered design columns were also normalized to unit empirical second moment.

\subsection{Screening and penalty selection}
\label{sec:appendix:interaction:screening}

We used a two-stage procedure, related to the relaxation and thresholding approach of \citet{ibrahim2025predicting}, to reduce the candidate sets and identify promising penalty configurations before integer optimization. Both methods initially considered every directed edge between distinct nodes as a candidate. We searched the grid
$$
\lambda_{n,M}^2\in \frac{2p(\log p)^2}{n}  \left\{1,\frac12,\frac14,\frac18,\frac1{16},\frac1{32}\right\}, \qquad \lambda_{n,I}^2=\rho \lambda_{n,M}^2, \quad \rho\in\{1,2\}.
$$

For each penalty configuration, we relaxed $g$ and $w$ to $[0,1]$, retaining the constraints $w_{klj}\le g_{kj}$, $w_{klj}\le g_{lj}$ and the ordering inequalities. We solved the resulting convex relaxation by outer approximation. At each evaluated point $(\bar g,\bar w)$, we added the supporting cut
$$
t\ge \tilde{f}(\bar{g},\bar{w})
+\nabla_g\tilde{f}(\bar{g},\bar{w})^\top(g-\bar{g})
+\nabla_w\tilde{f}(\bar{g},\bar{w})^\top(w-\bar{w})
$$
to a linear master problem with objective $t+\lambda_{n,M}^2\sum g_{kj}+\lambda_{n,I}^2\sum w_{klj}$. Cuts were reused across penalty configurations within each method. Each relaxation had a time limit of 120 seconds and a relative-gap tolerance of $0.1\%$. All outer-approximation solves were additionally limited to 1,000 iterations.

For each penalty configuration, we constructed a candidate set by retaining edges with relaxed values $g_{kj}>0.3$ and interactions with $w_{klj}>0.3$. Both parent edges of every retained interaction were also included. The candidate edges could contain cycles. We therefore constructed a provisional DAG by adding them in decreasing order of their relaxed $g_{kj}$ values and rejecting any edge that created a cycle.
A retained interaction was included in the provisional model only when both parent edges were present. 
We refitted each provisional model by ridge regression and evaluated
$$
\mathrm{BIC}_{\mathrm{EDF}} = \sum_{j=1}^p n\log\left(\frac{\mathrm{RSS}_{j,\mathrm{ridge}}}{n}\right) +\log(n)\sum_{j=1}^p\operatorname{tr}(S_j),
$$
where $\mathrm{RSS}_{j,\mathrm{ridge}}$ is the residual sum of squares of the ridge fit. Writing $H_j$ for the combined selected, centered main-effect and interaction design matrix, excluding the intercept column, the ridge smoothing matrix is
$$
S_j=H_j(H_j^\top H_j+nD_j)^{-1}H_j^\top, \qquad D_j=\operatorname{diag}\left(\gamma_{n,M}I_{d_{j,M}},\gamma_{n,I}I_{d_{j,I}}\right),
$$
where $d_{j,M}$ and $d_{j,I}$ are the numbers of selected main-effect and interaction columns, respectively. Responses were also centered, and $S_j=0$ when no blocks were selected. For each method, the two configurations with the lowest provisional $\mathrm{BIC}_{\mathrm{EDF}}$ were shortlisted for integer optimization. Each solve used its original screened candidate set and the provisional DAG as a warm start, with a time limit of 300 seconds and a relative-gap tolerance of $0.1\%$. The resulting models were refitted by ridge regression, and the final model was selected from the evaluated candidates by minimizing $\mathrm{BIC}_{\mathrm{EDF}}$.

\subsection{Data-generating process and additional results for the interaction simulations}
\label{sec:appendix:interaction:dgp}

We used $p=8$, $n=300$ and 20 replications per setting. Noise standard deviations were generated as in Section~\ref{sec:experiment:comparison}, with $\tau\in\{0.1,0.2\}$ and average noise variance normalized to one. Conditional on these standard deviations, we generated independent errors $\epsilon_j\sim N(0,\sigma_j^2)$ and set $X_j=\epsilon_j$ for the three roots $j=1,2,3$.

Let $S$ denote empirical centering and scaling to unit standard deviation, applied to the vector of generated observations. For each nonlinear main effect, we independently sampled $f_{kj}$ uniformly from $\{\sin(x),\cos(x),\tanh(x),\arctan(x),\tanh(x^3)\}$ and wrote $h_{kj}(X_k)=S(f_{kj}(X_k))$. For each interaction, we sampled a function $F_{abj}$ uniformly from $\{uv, \cos(u+v), \sin(uv), \frac{u}{\sqrt{v^2+0.25}} \}$ and defined $P_{abj}=S(F_{abj}(S(X_a),S(X_b)))$. For the ratio function, the numerator parent was chosen with equal probability from the two parents.

The parent pairs of $X_4,X_5,X_6$ were $(1,2),(1,3),(2,3)$, respectively. Unless specified otherwise, a node with parents $(a,b)$ followed 
$$
X_j=0.5h_{aj}(X_a)+0.5h_{bj}(X_b)+\eta_jP_{abj}+\epsilon_j,
$$
where $\eta_j=\eta$ for nodes assigned an interaction and $\eta_j=0$ otherwise. With one interaction, only $X_4$ received
an interaction. With three interactions, all of $X_4,X_5,X_6$
received interactions. We considered the following settings, in the order used in Tables~\ref{tab:interaction-simulations-tau01} and \ref{tab:interaction-simulations-tau02}:
\begin{enumerate}
\item One interaction with nonlinear main effects and $\eta\in\{0.5,1,2\}$.
\item Three interactions with nonlinear main effects and $\eta\in\{0.5,1,2\}$.
\item Three interactions with linear main effects and $\eta\in\{0.5,2\}$. The main-effect terms at $X_4,X_5,X_6$ were $0.5S(X_a)+0.5S(X_b)$.
\item An additive setting with no interactions, replacing the equation for $X_4$ by $X_4=h_{14}(X_1)+0.5h_{24}(X_2)+\epsilon_4$.
\item A pure-interaction setting, replacing the equation for $X_4$ by $X_4=S(S(X_1)S(X_2))+\epsilon_4$, with no explicit main effects. Only $X_4$ had an interaction.
\item Two three-parent settings, replacing the equation for $X_4$ by $X_4=0.5h_{14}(X_1)+0.5h_{24}(X_2)+P_{124} +b h_{34}(X_3)+\epsilon_4$, with $b\in\{0.2,0.5\}$. Only $X_4$ had an interaction.
\item Three interactions with $\eta=1$ and cosine main effects. All main-effect functions were set to $\cos(x)$ before empirical centering and scaling.
\end{enumerate}
For each $j=7,8$, we sampled a parent $k$ uniformly from $\{4,\ldots,j-1\}$ and generated $X_j=h_{kj}(X_k)+\epsilon_j$. These effects were nonlinear even in the linear-main-effect settings, and were cosine functions in the cosine setting. Equations not explicitly modified above retained their default form. The 13 settings were each evaluated at both noise-heterogeneity levels, giving 26 conditions and 520 datasets. Within each replication, parent choices, function assignments and underlying standard Gaussian draws were shared across corresponding interaction-strength and noise-heterogeneity conditions.

\begin{table}[!ht]
\centering
\small
\resizebox{\textwidth}{!}{
\begin{tabular}{lcccccc}
\toprule
Setting & MIP-main & MIP-interaction & GLOBE & TOPIC & NOTEARS & GSGES\\
\midrule
\multicolumn{7}{l}{$\tau=0.2$}\\
\midrule
One interaction, $\eta=0.5$
& 1.85 (0.33) & 0.75 (0.25) & 6.70 (0.33) & 6.60 (0.39) & 5.30 (0.75) & 7.70 (0.51)\\
One interaction, $\eta=1$
& 1.70 (0.40) & 0.80 (0.22) & 6.70 (0.30) & 6.40 (0.40) & 4.90 (0.81) & 7.80 (0.49)\\
One interaction, $\eta=2$
& 2.40 (0.43) & 0.80 (0.25) & 6.40 (0.23) & 6.45 (0.46) & 4.90 (0.90) & 7.40 (0.45)\\
Three interactions, $\eta=0.5$
& 2.90 (0.43) & 1.15 (0.30) & 7.10 (0.39) & 6.55 (0.33) & 3.75 (0.74) & 7.80 (0.53)\\
Three interactions, $\eta=1$
& 3.65 (0.54) & 0.35 (0.18) & 6.70 (0.28) & 7.15 (0.42) & 5.10 (0.92) & 8.00 (0.51)\\
Three interactions, $\eta=2$
& 5.65 (0.73) & 0.15 (0.08) & 6.80 (0.20) & 7.50 (0.50) & 6.40 (1.08) & 8.00 (0.47)\\
Linear mains, three interactions, $\eta=0.5$
& 6.05 (0.53) & 2.20 (0.41) & 6.40 (0.26) & 8.30 (0.53) & 3.35 (0.56) & 7.50 (0.50)\\
Linear mains, three interactions, $\eta=2$
& 6.55 (0.80) & 0.40 (0.21) & 6.75 (0.22) & 7.90 (0.51) & 5.65 (0.89) & 7.30 (0.42)\\
Additive
& 0.70 (0.23) & 0.70 (0.22) & 6.30 (0.33) & 6.35 (0.32) & 6.20 (0.89) & 8.20 (0.61)\\
Pure interaction
& 2.50 (0.22) & 0.75 (0.27) & 6.75 (0.31) & 6.70 (0.34) & 8.25 (0.97) & 7.30 (0.49)\\
Three parents, $b=0.2$
& 2.70 (0.40) & 1.80 (0.24) & 7.35 (0.21) & 7.65 (0.53) & 5.30 (0.73) & 9.00 (0.49)\\
Three parents, $b=0.5$
& 2.05 (0.43) & 1.00 (0.34) & 7.50 (0.20) & 7.60 (0.55) & 6.00 (0.90) & 9.00 (0.51)\\
Cosine mains, three interactions, $\eta=1$
& 0.35 (0.17) & 0.00 (0.00) & 5.50 (0.14) & 4.95 (0.30) & 6.10 (1.15) & 7.00 (0.44)\\
\bottomrule
\end{tabular}
}
\caption{\small Mean directed SHD (standard error) over 20 replications with $p=8$ and $n=300$ and $\tau=0.2$. Both MIP estimators use $\mathrm{BIC}_{\mathrm{EDF}}$ for final selection. Main effects are sampled nonlinear functions unless otherwise indicated.}
\label{tab:interaction-simulations-tau02}
\end{table}

\begin{table}[!ht]
\centering
\small
\resizebox{\textwidth}{!}{
\begin{tabular}{lcccccc}
\toprule
Setting & MIP-main & MIP-interaction & GLOBE & TOPIC & NoTears & GSGES\\
\midrule
\multicolumn{7}{l}{$\tau=0.1$}\\
\midrule
One interaction, $\eta=0.5$ & 4.65 (1.26) & 1.85 (0.86) & 19.25 (1.03) & 17.25 (1.28) & 4.00 (0.91) & 24.15 (1.41)\\
One interaction, $\eta=1$ & 4.30 (1.20) & 1.25 (0.57) & 18.00 (0.83) & 17.05 (1.08) & 4.30 (1.09) & 24.55 (1.36)\\
One interaction, $\eta=2$ & 6.00 (1.27) & 1.70 (0.68) & 18.25 (0.80) & 19.15 (1.34) & 4.05 (1.02) & 23.55 (1.23)\\
Three interactions, $\eta=0.5$ & 8.15 (1.68) & 2.55 (1.16) & 20.05 (1.26) & 21.20 (1.48) & 1.35 (0.51) & 23.65 (1.41)\\
Three interactions, $\eta=1$ & 11.85 (2.01) & 0.45 (0.26) & 19.40 (0.96) & 18.80 (1.23) & 3.85 (1.33) & 24.45 (1.30)\\
Three interactions, $\eta=2$ & 15.80 (2.10) & 0.10 (0.10) & 19.55 (0.71) & 19.95 (1.07) & 1.55 (0.59) & 25.05 (1.41)\\
Linear mains, three interactions, $\eta=0.5$ & 19.30 (1.41) & 5.75 (1.42) & 20.00 (1.19) & 23.55 (1.39) & 0.75 (0.36) & 24.40 (1.14)\\
Linear mains, three interactions, $\eta=2$ & 18.60 (2.29) & 0.35 (0.20) & 19.90 (0.89) & 21.90 (1.34) & 2.55 (0.71) & 24.00 (1.43)\\
Additive & 1.25 (0.57) & 1.00 (0.49) & 18.20 (1.09) & 15.60 (1.40) & 4.00 (0.76) & 25.60 (1.66)\\
Pure interaction & 6.70 (0.60) & 1.25 (0.57) & 18.00 (0.66) & 18.40 (1.28) & 5.80 (1.29) & 23.65 (1.44)\\
Three parents, $b=0.2$ & 9.85 (1.37) & 5.70 (0.61) & 18.55 (0.66) & 23.20 (1.19) & 8.15 (1.36) & 32.20 (1.50)\\
Three parents, $b=0.5$ & 7.30 (1.78) & 2.15 (0.96) & 18.55 (0.66) & 25.40 (1.13) & 4.30 (1.04) & 32.70 (1.73)\\
Cosine mains, three interactions, $\eta=1$ & 0.00 (0.00) & 0.00 (0.00) & 15.50 (0.70) & 13.80 (0.89) & 3.55 (1.00) & 23.40 (1.33)\\
\bottomrule
\multicolumn{7}{l}{$\tau=0.2$}\\
\midrule
One interaction, $\eta=0.5$ & 5.25 (1.10) & 2.00 (0.76) & 19.00 (1.14) & 18.10 (0.90) & 5.55 (0.94) & 24.15 (1.46)\\
One interaction, $\eta=1$ & 5.25 (1.15) & 2.00 (0.64) & 19.40 (1.10) & 16.65 (1.26) & 5.30 (0.92) & 24.50 (1.40)\\
One interaction, $\eta=2$ & 6.70 (1.19) & 1.85 (0.70) & 18.20 (0.91) & 17.15 (1.40) & 5.75 (1.26) & 23.70 (1.19)\\
Three interactions, $\eta=0.5$ & 8.80 (1.67) & 3.20 (1.01) & 20.60 (1.34) & 18.40 (0.87) & 2.60 (0.71) & 24.40 (1.48)\\
Three interactions, $\eta=1$ & 11.55 (1.83) & 0.75 (0.40) & 19.75 (1.11) & 19.80 (1.25) & 5.75 (1.45) & 25.30 (1.42)\\
Three interactions, $\eta=2$ & 15.90 (2.12) & 0.10 (0.10) & 20.75 (0.86) & 20.10 (1.19) & 2.55 (0.73) & 25.60 (1.50)\\
Linear mains, three interactions, $\eta=0.5$ & 19.80 (1.50) & 6.65 (1.36) & 18.15 (1.01) & 22.00 (1.56) & 3.85 (0.78) & 24.20 (1.33)\\
Linear mains, three interactions, $\eta=2$ & 19.10 (2.27) & 0.55 (0.46) & 20.15 (0.86) & 20.85 (1.31) & 3.75 (0.80) & 23.45 (1.36)\\
Additive & 1.50 (0.64) & 1.45 (0.56) & 18.25 (1.28) & 17.55 (0.97) & 6.35 (1.21) & 25.55 (1.61)\\
Pure interaction & 6.60 (0.69) & 1.25 (0.57) & 18.75 (1.09) & 17.65 (1.28) & 8.55 (1.41) & 23.45 (1.35)\\
Three parents, $b=0.2$ & 10.25 (1.23) & 6.80 (0.85) & 19.45 (0.90) & 24.25 (1.59) & 8.65 (1.43) & 32.10 (1.61)\\
Three parents, $b=0.5$ & 7.40 (1.84) & 2.80 (1.02) & 19.75 (0.82) & 24.50 (1.46) & 8.00 (1.40) & 32.35 (1.66)\\
Cosine mains, three interactions, $\eta=1$ & 0.00 (0.00) & 0.00 (0.00) & 14.75 (0.34) & 12.25 (0.95) & 2.85 (0.88) & 22.45 (1.32)\\
\bottomrule
\end{tabular}
}
\caption{Mean SID (standard error) over 20 replications with $p=8$, $n=300$ and $\tau=0.1,0.2$. Both MIP estimators use $\mathrm{BIC}_{\mathrm{EDF}}$ for final selection. Main effects are sampled nonlinear functions unless otherwise indicated. Lower values indicate better intervention recovery.}
\label{tab:interaction-sid-tau0102}
\end{table}

\vspace{0.2in}
\section{More on experiments}
\label{sec:appendix:teaser}

Many experiments in this section are based on publicly available graph structures adopted from \citet{manzour2021integer} and \citet{xu2025integer}, including the \textit{Insurance Small} graph, which is used extensively throughout our experiments. Table~\ref{tab:compare-graph} summarizes the names of these graphs, together with their numbers of nodes and true edges.

\subsection{Simulation setup for Figure \ref{fig:intro}}
\label{sec:appendix:setup-intro}
The true model is given in \eqref{ex:teaser}, and each reported result is averaged over 100 independent trials. The \textit{CAM-IncEdge} procedure is implemented using the R package \texttt{CAM} with default settings, except that we set \texttt{pruning = TRUE}. The \textit{NPVAR} procedure is implemented using the publicly available package at \href{https://github.com/MingGao97/NPVAR}{\texttt{https://github.com/MingGao97/NPVAR}} with default parameters. For \textit{MIP-linear}, we use the procedure~\eqref{eq:MIP} with basis functions $\{b_r(X_k)\}_{k=1}^{R_n}$ taken simply as $X_k$ itself. Our proposed \textit{MIP-nonlinear} approach is implemented using the same procedure; however, in order to handle non-linearity, we use degree-three splines with three internal knots for all edges. We set $\gamma_n = 0.1$ and $\lambda_n^2=0.01$ for both procedures. No super-structure, partial order sets, stable sets, or early stopping is used in any of the MIP-based methods.



\subsection{Simulation setup for Table \ref{tab:oa-vanilla}}
\label{sec:appendix:setup-oa-vanilla}

In this experiment, we use five graph structures as shown in Table~\ref{tab:oa-vanilla}. For all $(k,j)\in E^\star$, the true function is set as $f^\star_{kj}(x) = ( \sin(x) - \mathbb{E}[\sin(X_k)] + \cos(x) - \mathbb{E}[\cos(X_k)] )  / 2$. 
The noise standard deviations $\sigma_j^\star$ are independently generated from $0.5\cdot\mathsf{Beta}(1,2)+0.5$ to induce unequal error variances. For each setting, we generate $n=500$ independent and identically distributed samples and repeat the experiment over 50 independent trials. The entries reported in Table~\ref{tab:oa-vanilla} are averages across these trials.

We set $\gamma_n=1$ and $\lambda_n^2=0.005$. The same superstructure is imposed on both the cutting-plane and vanilla procedures, and no additional prior information is used. The vanilla procedure refers to the formulation in Supp. Mat.~\ref{sec:appendix:equivalence}, which is solved directly using \texttt{Gurobi}. To make the objective values directly comparable, the objective value reported for the cutting-plane procedure is that of~\eqref{eq:MIP} with an additional constant $p$.

\subsection{Simulation setup for Figure \ref{fig:early_stopping}}
\label{sec:appendix:setup-early}

We use the \textit{Insurance Small} graph, which contains $p=15$ nodes and $s^\star=25$ true edges. The data are otherwise generated following the same procedure as in Section~\ref{sec:experiment:comparison}, using the representative setting $n=500$ and $\tau=0.1$. We fix the tuning parameters at $\gamma_n=0.2$ and $\lambda_n^2=0.01$, and consider early-stopping thresholds $\tau^{\rm early}\in\{0,\, 0.01,\, \ldots,\, 0.09,\, 0.1\}$. We use the true moral graph as the super-structure, with no additional prior structural information supplied to \textit{MIP}. Each run of \textit{MIP} is terminated after 500 seconds.


\subsection{Implementation details of Section \ref{sec:experiment:comparison}}
\label{sec:appendix:setup-comparison}

Several methods, including \textit{CAM}, \textit{EqVar}, \textit{NPVAR}, and \textit{RESIT} follow an ordering-based two-step procedure: they first estimate a topological ordering and then apply a pruning step to recover the DAG structure. These methods can also naturally incorporate a super-structure. To ensure a fair comparison, we impose the same estimated super-structure and use a common pruning procedure across these methods. Both preliminary neighborhood selection and pruning are adapted from \citet{buhlmann2014}, using GAMBoost for the former and GAM regression for the latter.

\begin{itemize}
    
    \item \textit{CAM}: The implementation is available in the R library \texttt{CAM}. We use the function \texttt{selGamBoost()} for preliminary neighborhood selection, with parameters \texttt{max\_neighbors=10L} and \texttt{at\_least\_selected=0.05}, to construct the super-structure. We then use \texttt{CAM()} to estimate the node ordering. Finally, we apply \texttt{selGam()} for pruning with \texttt{cutoff=0.001}.

    \vspace{0.05in}
    \item \textit{NPVAR}: The implementation is available at \texttt{\url{https://github.com/MingGao97/NPVAR}}.  
    We use the function \texttt{NPVAR()} with parameter \texttt{run.mgcv=TRUE} to estimate the topological ordering. The preliminary neighborhood selection and pruning steps are implemented in the same manner as for \textit{CAM}.

    \vspace{0.05in}
    \item \textit{EqVar}: The implementation is available at \texttt{\url{https://github.com/WY-Chen/EqVarDAG}}. We use the function \texttt{EqVarDAG\_TD()} with parameter \texttt{mtd="cvlasso"} to estimate the topological ordering. The preliminary neighborhood selection and pruning steps are implemented in the same manner as for \textit{CAM}.

    \vspace{0.05in}
    \item \textit{RESIT}: The implementation is available at \texttt{\url{http://people.tuebingen.mpg.de/jpeters/onlineCodeANM.zip}}. We use the function \texttt{fit\_and\_test\_independence()} with parameters \texttt{alpha=0.05} and \texttt{model=train\_gam}. The preliminary neighborhood selection and pruning steps are implemented in the same manner as for \textit{CAM}.

    \vspace{0.05in}
    \item \textit{NoTears}: The implementation is available at \texttt{\url{https://github.com/xunzheng/notears}}. We use neural networks to model nonlinear effects with a single hidden layer of 5 neurons. We run the function \texttt{notears\_nonlinear()} with parameter \texttt{lambda2=0.01}, while \texttt{lambda1} is selected according to the Bayesian information criterion.

    \vspace{0.05in}
    \item \textit{GSGES}: The implementation is available in the Python interface of the \texttt{causal-learn} package. We run \texttt{ges()} with \texttt{score\_func="local\_score\_CV\_general"}, using five folds and \texttt{maxP=5}, while \texttt{lambda} is selected according to the Bayesian information criterion. Since GES returns a CPDAG, we convert it to a consistent DAG extension using \texttt{pdag2dag()}.

    \vspace{0.05in}
    \item \textit{GLOBE}: The implementation is available at \texttt{\url{https://github.com/osman-mian/Globe}}. We use the function \texttt{GlobeWrapper()} with the parameter \texttt{max\_interactions=2}.

    \vspace{0.05in}
    \item \textit{TOPIC}: The implementation is available at \texttt{\url{https://github.com/srhmm/topic}}. We use \texttt{Topic()} with parameters \texttt{data\_type=DataType.CONTINUOUS} and \texttt{score\_type=ScoreType.GAM}. The method returns both the estimated topological ordering and the resulting DAG. The DAG estimation is integrated into TOPIC's ordering procedure through its information-theoretic edge addition and refinement steps, and therefore is not treated as a separate post-ordering pruning stage.

\end{itemize}

As for our method \textit{MIP}, we add more implementation details. First, for simpler graph structures with $p=10$ and $\rho\in\{0.3,0.4\}$, we do not incorporate prior structural information obtained by bootstrapping \textit{CAM} and \textit{NPVAR}, providing a direct comparison among the methods. For more complex graph structures, we use such information to restrict the search space and improve computational efficiency. Notably, \textit{MIP} still outperforms the procedures used to construct this screening information, indicating that its performance gains are not solely attributable to the screening step. Second, we select the tuning parameter $\lambda_n^2$ from the grid $\{0.015,\,0.02,\,0.025,\,0.03,\,0.035, 0.040\}$. For each run of \textit{MIP}, we therefore evaluate 12 configurations, corresponding to six choices of $\lambda_n^2$ and two choices of prior information (from \textit{CAM} and \textit{NPVAR} respectively), and select the configuration with the lowest BIC score. We fix $\gamma_n=0.2$; although a smaller $\gamma_n$ can improve graph recovery, it also increases the computational burden. As shown in Supp. Mat.~\ref{sec:appendix:sensitive-gamma}, $\gamma_n=0.2$ provides a reasonable trade-off between SHD and runtime. Third, rather than specifying a fixed early-stopping tolerance $\tau^{\rm early}$, we impose a time limit of 180 seconds on each cutting plane run. Most runs terminate within this time limit; otherwise, the final optimality gap $U^{(T)}-L^{(T)}$ at the termination iteration $T$ still provides a certificate of the remaining suboptimality.

Tables \ref{tab:compare-c1} -- \ref{tab:compare-c4} show complete results, including runtime. As expected, the runtime of \textit{MIP} generally increases with graph complexity, reflecting the increasing difficulty of the underlying combinatorial optimization problem. Nevertheless, \textit{MIP} is not uniformly more computationally expensive than the competing methods. In particular, \textit{GSGES}, \textit{RESIT}, \textit{NoTears}, and \textit{GLOBE} can require comparable or substantially longer runtimes for the larger graph settings, while typically yielding noticeably worse SHD, SID, and F1. On the other hand, methods such as \textit{EqVar}, \textit{NPVAR}, \textit{TOPIC}, and \textit{CAM} are generally faster, but this computational advantage is accompanied by lower recovery accuracy, especially as the graph becomes more complex or the equal-variance assumption is violated. Thus, the additional computational cost of \textit{MIP} relative to the fastest baselines translates into a substantial improvement in graph recovery, rather than simply reflecting an inefficient implementation. In summary, its runtime scales with problem difficulty, while its recovery performance remains consistently strong across the considered settings.

\subsection{More experiments on ten real-world networks}
\label{sec:appendix:10-graphs}
We further compare \textit{MIP} with the baseline methods on ten semi-synthetic datasets with fixed graph structures. As described at the beginning of this section, these networks are adopted from \citet{manzour2021integer} and \citet{xu2025integer}. Unlike the Erd\H{o}s-R\'enyi graphs considered above, these benchmark networks are derived from structured graphical models motivated by real-world applications, and therefore provide more heterogeneous and practically relevant graph topologies for causal discovery methods.

The implementation details are otherwise identical to those in Section~\ref{sec:experiment:comparison} and Supp. Mat.~\ref{sec:appendix:setup-comparison}, except that the underlying graph structure is fixed rather than randomly generated in each trial. The simulation results are reported in Table~\ref{tab:compare-graph}. Overall, the results are qualitatively consistent with those from the Erd\H{o}s-R\'enyi experiments, suggesting that the observed performance patterns persist across these more structured network topologies.

\subsection{Alternative basis functions}
\label{sec:appendix:alter-basis}

We also evaluate \textit{MIP} using an alternative basis. We use six Gaussian radial basis functions for each predictor $X_k$:
$$
b_r(x)=\exp\left\{-\frac12\left(\frac{x-c_{rk}}{0.5}\right)^2\right\}, \qquad r=1,\ldots,6.
$$
where the centers $c_{rk}$ are chosen as the six equally spaced empirical quantiles of $X_k$.
The data generation and implementation details are otherwise identical to those in Section~\ref{sec:experiment:comparison} and Supp. Mat.~\ref{sec:appendix:setup-comparison}.

The simulation results are reported in Figure~\ref{fig:newbasis}. Across all graph settings, \textit{MIP} remains the strongest-performing method regardless of whether spline or radial basis functions are used. Moreover, the two basis choices yield very similar SHD, SID, and F1 distributions, indicating that the performance of \textit{MIP} is stable with respect to the choice of basis functions.

\begin{figure}[!ht]
    \centering
    \includegraphics[width=\textwidth]{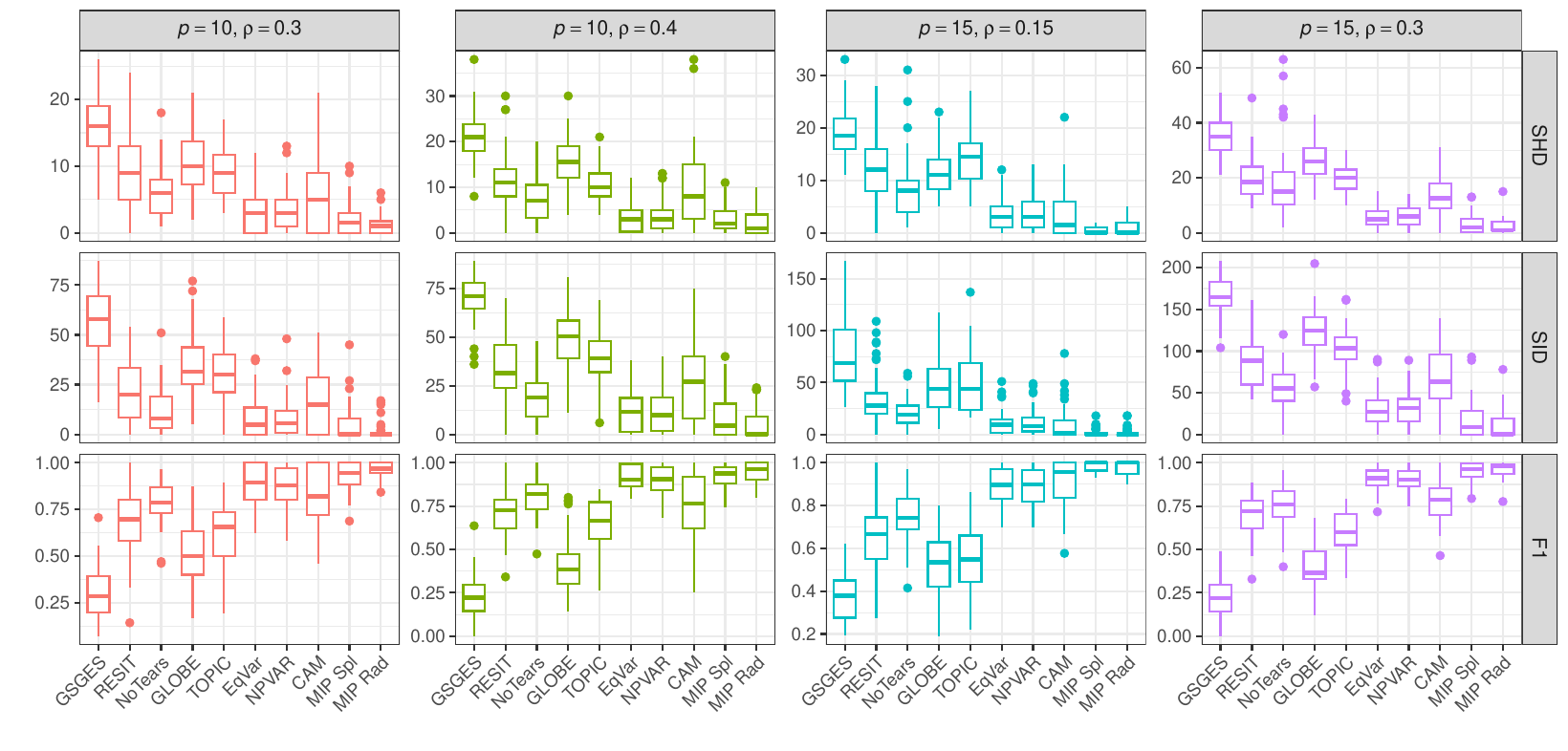}
    \caption{\small Performance comparison between our method MIP and the baseline methods with $n=500$ samples and noise-heterogeneity parameter $\tau=0.1$. \textit{MIP Spl} denotes our \textit{MIP} approach using spline basis functions, whereas \textit{MIP Rad} denotes \textit{MIP} using radial basis functions. Each column represents a different graph complexity (size $p$, density $\rho$), while each row corresponds to a different performance metric: SHD, SID, and F1. Each box summarizes results from 50 independent simulation trials.
    }
    \label{fig:newbasis}
\end{figure}

\subsection{A sensitivity analysis: the role of \texorpdfstring{$\gamma_n$}{gamma\_n}}
\label{sec:appendix:sensitive-gamma}

We conduct a sensitivity analysis to investigate the role of $\gamma_n$. We generate data following the same procedure as in Section~\ref{sec:experiment:comparison}, using a representative setting with graph size $p=10$, sample size $n=300$, noise-heterogeneity parameter $\tau=0.2$, and graph density $\rho=0.3$. We consider $\gamma_n^{-1} \in \{3,\,4,\ldots, 11\}$ and $\lambda_n^2\in\{0.01,\, 0.015,\, 0.02,\, 0.025,\, 0.03,\, 0.035,\, 0.04,\, 0.045\}$. For each value of $\gamma_n$, we evaluate all candidate values of $\lambda_n^2$ and select the one yielding the smallest BIC score. We do not consider values of $\gamma_n$ beyond this range because the associated computational cost becomes prohibitively large. No prior structural information is supplied to \textit{MIP} in this experiment, and each optimization run is restricted to a time limit of 180 seconds.

\begin{figure}[!ht]
    \centering
    \includegraphics[width=0.9\linewidth]{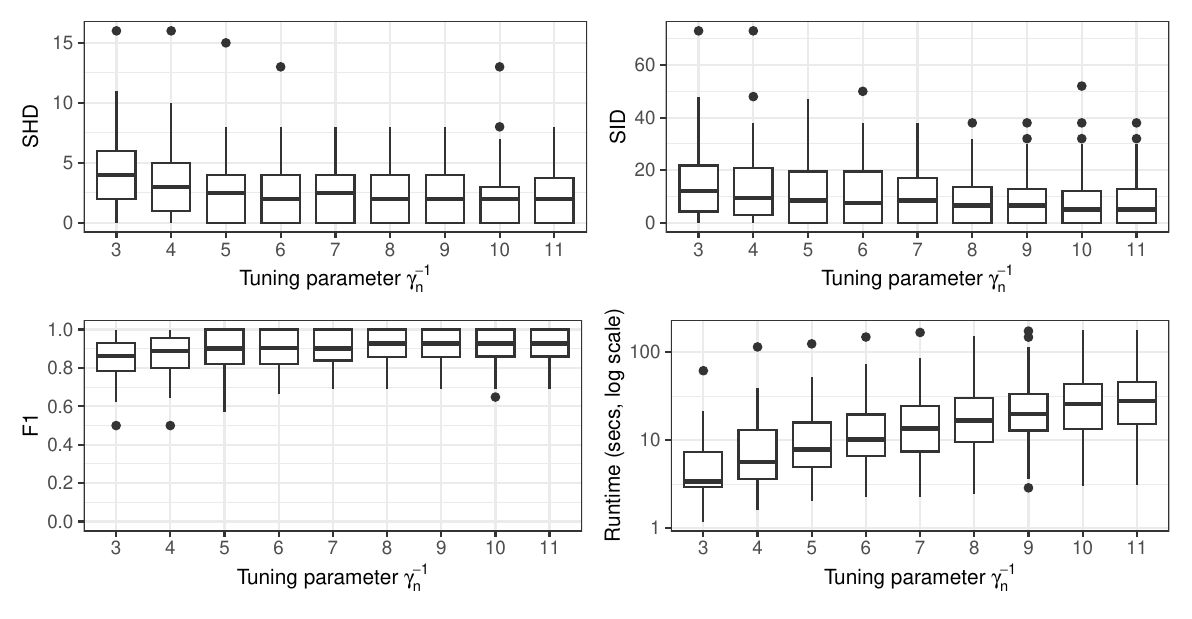}
    \caption{\small Effect of the tuning parameter $\gamma_n$ on the performance and computational cost of \textit{MIP}. The panels report SHD, SID, F1, and runtime across different values of $\gamma_n$, with runtime displayed on a logarithmic scale. Each box summarizes results over 50 independent simulation trials.}
    \label{fig:gamma_n}
\end{figure}

Figure~\ref{fig:gamma_n} illustrates the effect of the tuning parameter $\gamma_n$ on both graph recovery and computational cost. As $\gamma_n^{-1}$ increases (and hence $\gamma_n$ decreases), the recovery performance improves moderately: SHD and SID generally decrease, while F1 increases, although the improvements become relatively small once $\gamma_n^{-1}$ reaches a moderate value. In contrast, the computational cost increases much more rapidly with $\gamma_n^{-1}$, as is particularly apparent from the runtime panel shown on a logarithmic scale. Based on this trade-off, we set $\gamma_n^{-1}=5$ (which is $\gamma_n=0.2$) throughout the main experiments, which retains most of the improvement in recovery accuracy while keeping the computational cost manageable.

\subsection{A sensitivity analysis: the role of \texorpdfstring{$R_n$}{Rn}}
\label{sec:appendix:role-rn}

A larger number of basis functions $R_n$ improves the approximation accuracy of each function class $\mathcal{F}_{kj}$ by its truncated version $\mathcal{F}_{kj,n}$. Theoretically, a non-parametric model with $R_n \asymp n^{\zeta}$ for some appropriately chosen $\zeta > 0$ can substantially improve graph recovery and the convergence rate of variance estimators. We aim to demonstrate this through simulation.

In this simulation, we generate data following the same procedure as in Section~\ref{sec:experiment:comparison}, using a representative setting with graph size $p=15$, sample size $n=500$, noise-heterogeneity parameter $\tau =0.2$, and graph density $\rho=0.3$. 
We apply procedure \eqref{eq:MIP} and use the same set of degree-three splines across all edges. For a given $\zeta > 0$, we place $\lfloor n^{\zeta} \rfloor$ internal knots with evenly spaced percentiles, resulting in a total of $R_n = 3 + \lfloor n^{\zeta} \rfloor$ basis functions. The tuning parameters are fixed as $\lambda_n^2=0.01$ and $\gamma_n = 0.2$. We do not use any prior structural information, and each optimization run is restricted to a time limit of 180 seconds. 

Figure~\ref{fig:rn} examines the sensitivity of \textit{MIP} to the number of basis functions $R_n=3+\lfloor n^\zeta\rfloor$. Moving from the linear specification to a small number of basis functions leads to a substantial improvement in graph recovery, as reflected by a sharp decrease in SID and increase in F1. As $R_n$ increases further, however, the improvement quickly levels off, with only modest changes beyond a moderate basis size. Thus, using a moderate $R_n$ appears sufficient to capture most of the nonlinear structure without incurring the additional cost of a much larger basis.

\begin{figure}[!ht]
    \centering
    \includegraphics[width=0.95\textwidth]{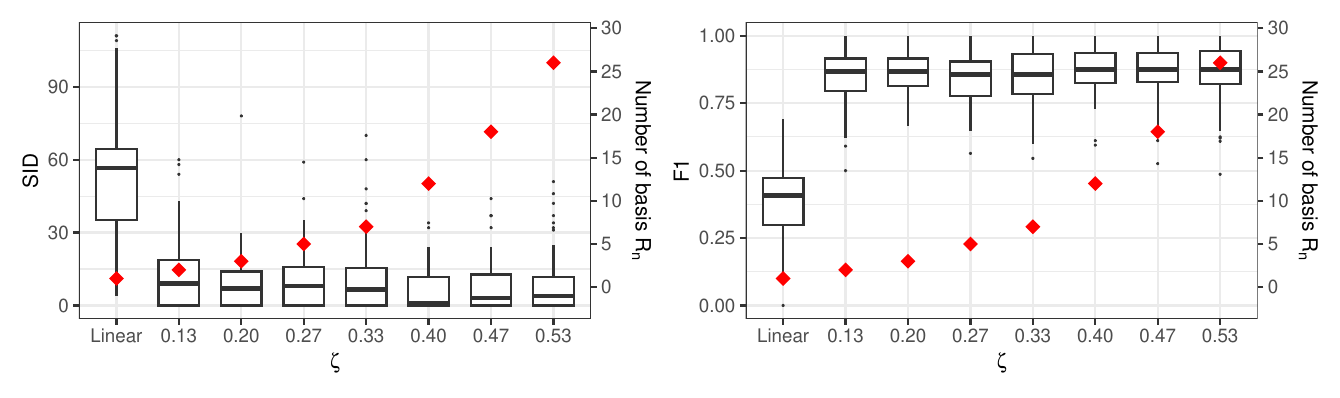}
    \caption{\small Sensitivity of \textit{MIP} to the number of basis functions $R_n=\lfloor n^\zeta\rfloor$. The left and right panels report SID and F1, respectively, across different values of $\zeta$, measured using the left vertical axis. Each box summarizes results from 50 independent simulation trials, while the red diamonds indicate the corresponding number of basis functions $R_n$, measured using the right vertical axis. The ``Linear'' specification corresponds to applying \textit{MIP} directly to the original design matrix $X$, without additional basis expansion.
    }
    \label{fig:rn}
\end{figure}

\subsection{A sensitivity analysis: non-Gaussian noise}
\label{sec:appendix:sensitive-nongaussian}
We conduct a sensitivity analysis on the noise distribution of $\epsilon_j$ in model~\ref{eq:DAG-model}. We generate data following the same procedure as in Section~\ref{sec:experiment:comparison},
using a representative setting with graph size $p=10$, sample size $n=500$, noise-heterogeneity parameter $\tau=0.2$, and graph density $\rho=0.3$. However, the noise terms are generated using Student-$t$ distribution:
$$
\epsilon_j=\sigma_j\sqrt{\frac{\nu-2}{\nu}}\,T_j,
\qquad T_j\overset{\mathrm{i.i.d.}}{\sim}t_\nu,
$$
so that $\mathrm{Var}(\epsilon_j) = \sigma_j^2$.
No prior structural information is supplied to MIP in this experiment, and each optimization run is restricted to a time limit of 180 seconds.

Figure~\ref{fig:tdistn} shows that the relative ranking of the methods is broadly stable across different degrees of freedom. As the degrees of freedom increase and the Student-$t$ noise becomes less heavy-tailed, the causal discovery problem becomes slightly easier overall, as reflected by modest decreases in SHD and SID and corresponding increases in F1 for several methods. Throughout these settings, \textit{MIP} consistently remains among the best-performing procedures, suggesting that its advantage is robust to substantial departures from Gaussian noise.

\begin{figure}[!ht]
    \centering
    \includegraphics[width=1\linewidth]{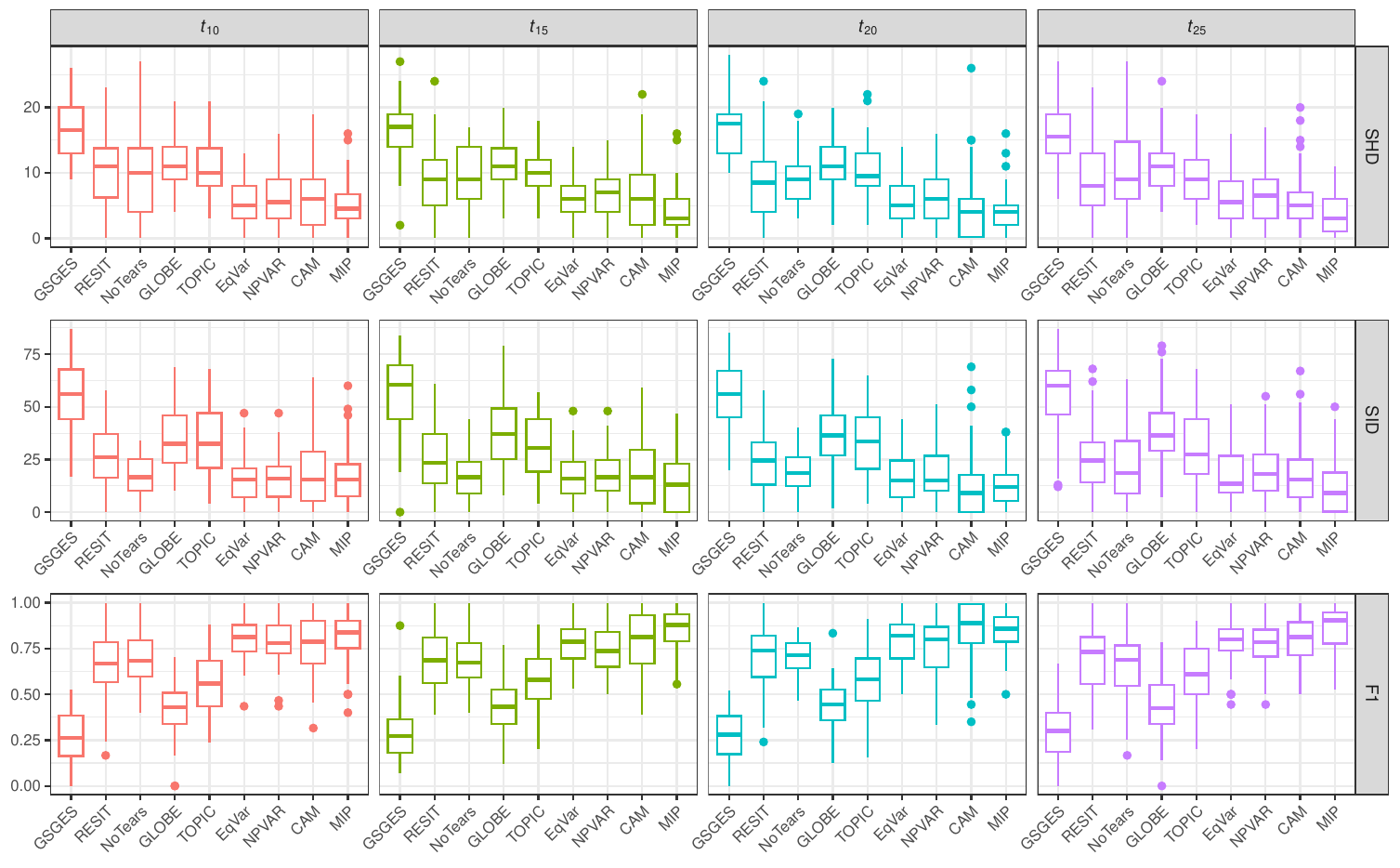}
    \caption{\small Sensitivity analysis under non-Gaussian structural noises. The columns correspond to Student-$t$ noises with degrees of freedom $\nu\in\{10,15,20,25\}$, while the rows report SHD, SID, and F1, respectively. Each box summarizes the results over 50 independent simulation trials.}
    \label{fig:tdistn}
\end{figure}

\subsection{Real data: causal chambers}
\label{sec:appendix:causal-chambers}
We collect the data using the Python package \texttt{causalchamber}, specifically through the dataset \texttt{lt\_interventions\_standard\_v1}. This dataset contains observations from multiple environments corresponding to different intervention settings. In our analysis, we only use the \texttt{uniform\_reference} environment as the observational reference dataset. From the nine recorded variables, we construct the polarization variable as the difference between \texttt{pol\_1} and \texttt{pol\_2} and retain the resulting eight-dimensional dataset for analysis. The ground truth is given in Figure~\ref{sfig:ground-truth}. The implementation details for both \textit{MIP} and the competing baseline methods are otherwise the same as those described in Section~\ref{sec:experiment:comparison} and Supp. Mat.~\ref{sec:appendix:setup-comparison}, except that no prior structural information is used for \textit{MIP} in this experiment.

The full dataset contains 10000 observations. To evaluate performance across different sample sizes, we partition the observations into 1000 consecutive blocks of size 10. For each of the 10 trials, we select one observation at the same within-block position from every block, thereby constructing 10 disjoint subsamples. For each target sample size $n\in\{200,300,400,500,600,700,800\}$, we retain the selected observations from the first $n$ blocks, resulting in 10 datasets for each value of $n$. 
These datasets are used in Figure~\ref{fig:causal_chamber_n}, which shows that \textit{MIP} consistently delivers the strongest overall performance across sample sizes, achieving substantially lower SHD, SID and FDR and markedly higher TDR than the competing methods.

\begin{figure}[!ht]
    \centering
    \begin{subfigure}[b]{0.32\textwidth}
        \centering
        \begin{tikzpicture}[scale=0.7,
            >=Stealth, every node/.style={font=\footnotesize}, arrow/.style={->, thin}, circnode/.style={anchor=center}]
            \def\r{2.2}
            \foreach \i/\name in {
                0/R,
                1/G,
                2/B,
                3/{\widetilde{C}},
                4/{\widetilde{I}_1},
                5/{\widetilde{I}_2},
                6/{\widetilde{I}_3},
                7/{\theta_1 - \theta_2}
            } {
                \node[circnode] (N\i) at ({\r*cos(360*\i/8)}, {\r*sin(360*\i/8)}) {$\name$};
            }
            \foreach \i in {0,1,2} {
                \draw[arrow] (N\i) -- (N3);
            }
            \foreach \src in {0,1,2} {
                \foreach \dst in {4,5,6} {
                    \draw[arrow] (N\src) -- (N\dst);
                }
            }
            \draw[arrow] (N7) -- (N6);
        \end{tikzpicture}
        \caption{\small Ground truth DAG.}
        \label{sfig:ground-truth}
    \end{subfigure}
    \hfill
    \begin{subfigure}[b]{0.32\textwidth}
        \centering
        \begin{tikzpicture}[
            scale=0.7,
            >=Stealth,
            every node/.style={font=\footnotesize},
            arrow/.style={->, thin},
            falseedge/.style={->, thin, red},
            circnode/.style={anchor=center}
        ]
            \def\r{2.2}
            \foreach \i/\name in {
                0/R,
                1/G,
                2/B,
                3/{\widetilde{C}},
                4/{\widetilde{I}_1},
                5/{\widetilde{I}_2},
                6/{\widetilde{I}_3},
                7/{\theta_1-\theta_2}
            } {
                \node[circnode] (N\i) at
                ({\r*cos(360*\i/8)}, {\r*sin(360*\i/8)})
                {$\name$};
            }
            \foreach \src in {0,1,2} {
                \foreach \dst in {3,4,5,6} {
                    \draw[arrow] (N\src) -- (N\dst);
                }
            }
            \draw[arrow] (N7) -- (N6);
            \draw[falseedge] (N3) -- (N4);
            \draw[falseedge] (N3) -- (N5);
            \draw[falseedge] (N5) -- (N4);
            \draw[falseedge] (N6) -- (N3);
            \draw[falseedge] (N6) -- (N4);
            \draw[falseedge] (N6) -- (N5);
        \end{tikzpicture}
        \caption{\small DAG obtained by \textit{MIP}.}
        \label{sfig:chamber-mip}
    \end{subfigure}
    \hfill
    \begin{subfigure}[b]{0.32\textwidth}
        \centering
        \scriptsize
        \begin{tabular}{ccc}
            \toprule
            Method & SHD & SID  \\
            \midrule
            GSGES       & 20 & 40  \\
            RESIT       & 29 & 39  \\
            NoTears     & 30 & 40  \\
            GLOBE       & 14 & 16  \\
            TOPIC       & 21 & 34  \\
            EqVar       & 32 & 37  \\
            NPVAR       & 32 & 37  \\
            CAM         & 24 & 36  \\
            MIP         & 6  & 0  \\
            \bottomrule
        \end{tabular}
        \caption{\small Performance comparison.}
        \label{sfig:chamber-comparison}
    \end{subfigure}
    \caption{\small Learning causal models of the light tunnel chamber in~\cite{gamella2025causal}. Left panel: the ground truth. Middle panel: the estimated graph by our method \textit{MIP}, where red edges indicate incorrect discoveries. Right panel: performance comparison for all methods. }
    \label{fig:causal_chamber}
\end{figure}

Finally, we collect $n = 800$ observational samples of the eight variables and apply both \textit{MIP} and the baseline methods.  Figure~\ref{sfig:chamber-comparison} shows that \textit{MIP} achieves the best performance in terms of both SHD and SID. As illustrated in Figure~\ref{sfig:chamber-mip}, \textit{MIP} recovers all true causal edges, including the nonlinear effect, while introducing six additional edges. Interestingly, despite these false discoveries, the estimated graph has SID equal to zero. This indicates that the falsely discovered edges do not induce errors in pairwise intervention-effect identification, since the adjustment sets implied by the estimated graph remain valid under the true DAG.

\subsection{Real data: microarray data}
\label{sec:appendix:microarray}

We collect the data from \texttt{\url{https://media.springernature.com/original/springer-static/esm/art\%3A10.1186\%2Fgb-2004-5-11-r92/MediaObjects/13059_2004_896_MOESM1_ESM.txt}}. We extract the 15 genes in the mevalonate (MVA) pathway from the original microarray dataset. The pathway is shown in Figure~\ref{fig:gene_pathway}. Each gene is standardized to have zero mean and unit variance across the 118 samples. We generate 30 bootstrap samples from the original data to assess the stability of the methods. The implementation details for both \textit{MIP} and the competing baseline methods are otherwise the same as those described in Section~\ref{sec:experiment:comparison} and Supp. Mat.~\ref{sec:appendix:setup-comparison}, except that we set $\lambda_n^2 = 0.08$, and no
prior structural information is used for MIP in this experiment.

We first construct a directed graph using the dashed arrows in Figure~\ref{fig:gene_pathway} and take its transitive closure to obtain the set $\mathcal{P}_0$ of known ordered gene pairs, with $|\mathcal{P}_0|=91$. Given an estimated DAG with adjacency matrix $\widehat A$, we construct its reachability matrix $\widehat R$ by
$$
\widehat R_{kj} = \mathds{1}\left\{ \left[(I-\widehat A)^{-1}\right]_{kj}>0 \right\}, \qquad k\neq j.
$$
Thus, $\widehat R_{kj}=1$ if there exists a directed path from gene $k$ to gene $j$ in the estimated DAG, and $0$ otherwise.

For each $(k,j)\in\mathcal{P}_0$, we classify the estimated relation as correct if $\widehat R_{kj}=1$, reversed if $\widehat R_{jk}=1$, and missing otherwise. Let $C$ and $R$ denote the numbers of correct and reversed relations, respectively. We compute the influence score $C/|\widehat E|$ and the net influence score $(C-R)/|\widehat E|$, where $|\widehat E|$ is the number of estimated edges. Both scores are set to zero when $|\widehat E|=0$.

\vspace{0.1in}
\begin{figure}[!ht]
    \centering
    \includegraphics[width=1\linewidth]{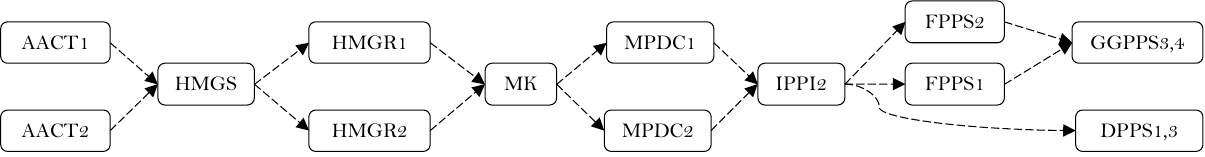}
    \caption{\small Gene expressions in the cytoplasm (MVA pathway). The dashed arrows indicate the causal directions within the MVA pathway.}
    \label{fig:gene_pathway}
\end{figure}

\clearpage
\begingroup
\small
\setlength{\tabcolsep}{2.5pt}
\renewcommand{\arraystretch}{1.05}
\begin{longtable}{c@{\hspace{5pt}\vrule\hspace{5pt}}cccc@{\hspace{5pt}\vrule\hspace{5pt}}cccc@{\hspace{5pt}\vrule\hspace{5pt}}cccc@{\hspace{5pt}\vrule\hspace{5pt}}c}
\toprule
& \multicolumn{4}{c}{GSGES} & \multicolumn{4}{c}{RESIT} & \multicolumn{4}{c}{NoTears} & \\
\cmidrule{2-5}\cmidrule{6-9}\cmidrule{10-13}
$(p,\,\rho)$ & SHD & SID & F1 & Time & SHD & SID & F1 & Time & SHD & SID & F1 & Time & \\
\midrule
$(10,\,0.30)$ & 16.0 & 56.5 & 0.30 & 86.8 & 9.9 & 22.6 & 0.68 & 10.2 & 6.2 & 12.2 & 0.79 & 23.6 & \\
$(10,\,0.40)$ & 21.0 & 71.1 & 0.23 & 85.2 & 11.2 & 32.9 & 0.71 & 10.2 & 7.4 & 18.8 & 0.81 & 29.0 & \\
$(10,\,0.50)$ & 26.4 & 79.4 & 0.16 & 85.5 & 13.9 & 39.4 & 0.70 & 10.2 & 10.6 & 24.9 & 0.78 & 38.0 & \\
$(15,\,0.15)$ & 19.0 & 76.2 & 0.38 & 230.0 & 12.3 & 34.8 & 0.66 & 42.5 & 8.7 & 21.3 & 0.76 & 30.9 & \\
$(15,\,0.30)$ & 35.1 & 165.2 & 0.22 & 245.4 & 19.7 & 85.6 & 0.69 & 50.7 & 18.3 & 56.0 & 0.75 & 58.9 & \\
$(20,\,0.15)$ & 31.1 & 185.9 & 0.35 & 393.7 & 23.5 & 101.4 & 0.63 & 133.8 & 13.6 & 55.2 & 0.77 & 58.5 & \\
$(20,\,0.20)$ & 44.4 & 275.1 & 0.25 & 390.2 & 25.6 & 144.7 & 0.69 & 145.6 & 20.9 & 98.0 & 0.76 & 84.9 & \\
$(25,\,0.20)$ & 66.2 & 468.1 & 0.21 & 622.8 & 38.5 & 283.8 & 0.69 & 291.0 & 30.7 & 198.2 & 0.75 & 142.7 & \\
\midrule
& \multicolumn{4}{c}{Globe} & \multicolumn{4}{c}{TOPIC} & \multicolumn{4}{c}{EqVar} & \\
\cmidrule{2-5}\cmidrule{6-9}\cmidrule{10-13}
$(p,\,\rho)$ & SHD & SID & F1 & Time & SHD & SID & F1 & Time & SHD & SID & F1 & Time & \\
\midrule
$(10,\,0.30)$ & 10.4 & 34.9 & 0.50 & 35.1 & 9.3 & 29.5 & 0.62 & 2.6 & 3.4 & 8.7 & 0.88 & 0.9 & \\
$(10,\,0.40)$ & 15.2 & 48.1 & 0.42 & 38.1 & 10.7 & 39.8 & 0.64 & 2.1 & 3.4 & 12.0 & 0.91 & 1.0 & \\
$(10,\,0.50)$ & 19.4 & 58.4 & 0.38 & 41.8 & 11.8 & 39.1 & 0.67 & 2.2 & 4.5 & 16.9 & 0.90 & 1.0 & \\
$(15,\,0.15)$ & 11.8 & 49.2 & 0.52 & 79.9 & 14.2 & 47.7 & 0.54 & 3.0 & 3.4 & 10.9 & 0.90 & 1.5 & \\
$(15,\,0.30)$ & 26.0 & 121.1 & 0.40 & 101.0 & 20.0 & 102.5 & 0.60 & 5.7 & 5.6 & 31.5 & 0.91 & 1.8 & \\
$(20,\,0.15)$ & 21.5 & 119.0 & 0.48 & 134.5 & 23.0 & 116.1 & 0.54 & 7.9 & 6.1 & 29.3 & 0.89 & 1.9 & \\
$(20,\,0.20)$ & 33.4 & 195.1 & 0.40 & 159.3 & 28.2 & 170.3 & 0.57 & 12.8 & 8.1 & 59.2 & 0.90 & 2.2 & \\
$(25,\,0.20)$ & 51.9 & 350.4 & 0.34 & 257.4 & 42.1 & 310.1 & 0.54 & 17.2 & 11.0 & 123.7 & 0.91 & 2.5 & \\
\midrule
& \multicolumn{4}{c}{NPVAR} & \multicolumn{4}{c}{CAM} & \multicolumn{4}{c}{MIP} & Bootstrap \\
\cmidrule{2-5}\cmidrule{6-9}\cmidrule{10-13}\cmidrule{14-14}
$(p,\,\rho)$ & SHD & SID & F1 & Time & SHD & SID & F1 & Time & SHD & SID & F1 & Time & Time \\
\midrule
$(10,\,0.30)$ & 3.6 & 8.2 & 0.87 & 1.4 & 6.0 & 16.5 & 0.81 & 2.7 & 2.0 & 5.8 & 0.93 & 18.0 & 3.2 \\
$(10,\,0.40)$ & 3.8 & 12.3 & 0.90 & 1.4 & 9.6 & 26.0 & 0.76 & 2.7 & 2.9 & 9.4 & 0.93 & 29.2 & 3.2 \\
$(10,\,0.50)$ & 5.1 & 16.9 & 0.88 & 1.4 & 15.5 & 42.7 & 0.66 & 2.9 & 3.2 & 11.0 & 0.93 & 6.1 & 3.4 \\
$(15,\,0.15)$ & 3.9 & 11.5 & 0.88 & 3.7 & 3.4 & 9.8 & 0.91 & 5.8 & 0.5 & 1.6 & 0.98 & 5.8 & 8.8 \\
$(15,\,0.30)$ & 6.0 & 30.7 & 0.90 & 4.4 & 14.2 & 65.2 & 0.78 & 5.9 & 3.2 & 17.2 & 0.95 & 29.0 & 8.8 \\
$(20,\,0.15)$ & 6.6 & 27.3 & 0.89 & 8.2 & 7.5 & 39.1 & 0.88 & 7.7 & 1.9 & 10.7 & 0.97 & 43.2 & 15.3 \\
$(20,\,0.20)$ & 8.4 & 56.9 & 0.89 & 8.9 & 15.8 & 91.3 & 0.81 & 7.6 & 2.6 & 18.4 & 0.97 & 65.9 & 13.0 \\
$(25,\,0.20)$ & 14.5 & 144.4 & 0.88 & 15.5 & 19.9 & 161.0 & 0.84 & 9.5 & 5.0 & 49.3 & 0.96 & 95.8 & 25.1 \\
\bottomrule
\caption{\small Full simulation results for $n=500$ and $\tau=0.1$. Each entry reports the mean over 50 independent simulation trials. Runtime is reported in seconds. For parallelizable computations in \textit{MIP}, \textit{NoTears}, \textit{GSGES} and bootstrap, runtime is reported per run.}
\label{tab:compare-c1}
\end{longtable}
\endgroup

\clearpage
\begingroup
\small
\setlength{\tabcolsep}{2.5pt}
\renewcommand{\arraystretch}{1.05}
\begin{longtable}{c@{\hspace{5pt}\vrule\hspace{5pt}}cccc@{\hspace{5pt}\vrule\hspace{5pt}}cccc@{\hspace{5pt}\vrule\hspace{5pt}}cccc@{\hspace{5pt}\vrule\hspace{5pt}}c}
\toprule
& \multicolumn{4}{c}{GSGES} & \multicolumn{4}{c}{RESIT} & \multicolumn{4}{c}{NoTears} & \\
\cmidrule{2-5}\cmidrule{6-9}\cmidrule{10-13}
$(p,\,\rho)$ & SHD & SID & F1 & Time & SHD & SID & F1 & Time & SHD & SID & F1 & Time & \\
\midrule
$(10,\,0.30)$ & 14.6 & 52.5 & 0.34 & 22.0 & 9.3 & 29.0 & 0.66 & 6.2 & 6.9 & 15.0 & 0.76 & 20.9 & \\
$(10,\,0.40)$ & 21.2 & 67.7 & 0.21 & 21.4 & 13.3 & 37.2 & 0.63 & 7.5 & 9.2 & 19.4 & 0.76 & 27.1 & \\
$(10,\,0.50)$ & 28.1 & 82.2 & 0.14 & 21.8 & 18.3 & 52.3 & 0.61 & 6.1 & 10.7 & 26.6 & 0.78 & 36.7 & \\
$(15,\,0.15)$ & 19.0 & 69.4 & 0.35 & 54.2 & 12.7 & 40.8 & 0.62 & 25.8 & 9.2 & 21.9 & 0.72 & 25.0 & \\
$(15,\,0.30)$ & 37.6 & 177.0 & 0.20 & 60.9 & 24.0 & 114.4 & 0.64 & 25.4 & 19.0 & 68.4 & 0.74 & 53.2 & \\
$(20,\,0.15)$ & 31.5 & 187.3 & 0.33 & 102.5 & 22.3 & 115.2 & 0.62 & 58.7 & 19.2 & 58.5 & 0.69 & 49.1 & \\
$(20,\,0.20)$ & 43.3 & 259.9 & 0.24 & 114.0 & 27.4 & 162.7 & 0.65 & 57.6 & 23.6 & 89.4 & 0.72 & 66.9 & \\
$(25,\,0.20)$ & 69.0 & 477.3 & 0.19 & 151.6 & 41.3 & 327.8 & 0.66 & 108.6 & 50.4 & 206.9 & 0.66 & 127.0 & \\
\midrule
& \multicolumn{4}{c}{Globe} & \multicolumn{4}{c}{TOPIC} & \multicolumn{4}{c}{EqVar} & \\
\cmidrule{2-5}\cmidrule{6-9}\cmidrule{10-13}
$(p,\,\rho)$ & SHD & SID & F1 & Time & SHD & SID & F1 & Time & SHD & SID & F1 & Time & \\
\midrule
$(10,\,0.30)$ & 11.4 & 41.1 & 0.25 & 21.3 & 14.1 & 43.9 & 0.35 & 1.5 & 3.3 & 11.7 & 0.87 & 0.9 & \\
$(10,\,0.40)$ & 16.5 & 51.7 & 0.22 & 23.4 & 16.8 & 52.9 & 0.38 & 1.5 & 4.2 & 14.9 & 0.88 & 1.1 & \\
$(10,\,0.50)$ & 22.6 & 66.9 & 0.22 & 28.4 & 21.4 & 65.8 & 0.36 & 1.4 & 5.1 & 20.0 & 0.89 & 0.9 & \\
$(15,\,0.15)$ & 13.9 & 51.5 & 0.22 & 43.7 & 20.0 & 62.0 & 0.30 & 2.3 & 3.6 & 13.5 & 0.89 & 1.3 & \\
$(15,\,0.30)$ & 31.2 & 142.4 & 0.21 & 57.1 & 32.7 & 147.4 & 0.31 & 3.9 & 8.0 & 52.8 & 0.88 & 1.3 & \\
$(20,\,0.15)$ & 25.4 & 142.4 & 0.22 & 82.7 & 32.5 & 160.8 & 0.30 & 5.1 & 7.0 & 41.2 & 0.88 & 1.7 & \\
$(20,\,0.20)$ & 34.6 & 194.8 & 0.19 & 93.8 & 41.2 & 216.9 & 0.28 & 4.6 & 7.2 & 53.8 & 0.90 & 1.7 & \\
$(25,\,0.20)$ & 58.1 & 382.9 & 0.17 & 152.3 & 64.7 & 415.7 & 0.24 & 13.7 & 12.6 & 136.9 & 0.89 & 2.2 & \\
\midrule
& \multicolumn{4}{c}{NPVAR} & \multicolumn{4}{c}{CAM} & \multicolumn{4}{c}{MIP} & Bootstrap \\
\cmidrule{2-5}\cmidrule{6-9}\cmidrule{10-13}\cmidrule{14-14}
$(p,\,\rho)$ & SHD & SID & F1 & Time & SHD & SID & F1 & Time & SHD & SID & F1 & Time & Time \\
\midrule
$(10,\,0.30)$ & 3.5 & 11.3 & 0.86 & 1.3 & 5.0 & 15.9 & 0.82 & 2.4 & 1.5 & 4.7 & 0.94 & 7.0 & 3.2 \\
$(10,\,0.40)$ & 4.5 & 15.3 & 0.87 & 1.6 & 8.8 & 29.1 & 0.76 & 3.0 & 3.2 & 11.4 & 0.91 & 13.0 & 3.3 \\
$(10,\,0.50)$ & 5.9 & 21.5 & 0.87 & 1.3 & 19.1 & 51.8 & 0.58 & 2.7 & 5.3 & 18.0 & 0.89 & 19.4 & 3.4 \\
$(15,\,0.15)$ & 4.0 & 12.3 & 0.87 & 3.1 & 3.5 & 9.6 & 0.89 & 4.9 & 0.9 & 2.4 & 0.97 & 8.0 & 9.3 \\
$(15,\,0.30)$ & 7.4 & 48.0 & 0.89 & 3.1 & 17.6 & 89.2 & 0.73 & 4.4 & 5.0 & 31.6 & 0.93 & 69.9 & 7.5 \\
$(20,\,0.15)$ & 7.9 & 41.4 & 0.86 & 7.1 & 7.8 & 44.1 & 0.87 & 6.4 & 2.7 & 14.5 & 0.95 & 60.1 & 15.6 \\
$(20,\,0.20)$ & 8.7 & 55.5 & 0.88 & 7.1 & 14.0 & 87.8 & 0.82 & 6.1 & 3.8 & 26.8 & 0.95 & 79.8 & 13.0 \\
$(25,\,0.20)$ & 13.1 & 141.6 & 0.89 & 13.7 & 22.1 & 192.6 & 0.82 & 7.9 & 8.0 & 77.8 & 0.93 & 127.2 & 22.9 \\
\bottomrule
\caption{\small Full simulation results for $n=300$ and $\tau=0.1$. Each entry reports the mean over 50 independent simulation trials. Runtime is reported in seconds. For parallelizable computations in \textit{MIP}, \textit{NoTears}, \textit{GSGES} and bootstrap, runtime is reported per run.}
\label{tab:compare-c2}
\end{longtable}
\endgroup

\clearpage
\begingroup
\small
\setlength{\tabcolsep}{2.5pt}
\renewcommand{\arraystretch}{1.05}
\begin{longtable}{c@{\hspace{5pt}\vrule\hspace{5pt}}cccc@{\hspace{5pt}\vrule\hspace{5pt}}cccc@{\hspace{5pt}\vrule\hspace{5pt}}cccc@{\hspace{5pt}\vrule\hspace{5pt}}c}
\toprule
& \multicolumn{4}{c}{GSGES} & \multicolumn{4}{c}{RESIT} & \multicolumn{4}{c}{NoTears} & \\
\cmidrule{2-5}\cmidrule{6-9}\cmidrule{10-13}
$(p,\,\rho)$ & SHD & SID & F1 & Time & SHD & SID & F1 & Time & SHD & SID & F1 & Time & \\
\midrule
$(10,\,0.30)$ & 16.3 & 54.0 & 0.31 & 96.4 & 8.7 & 22.9 & 0.72 & 9.9 & 8.1 & 18.4 & 0.73 & 23.8 & \\
$(10,\,0.40)$ & 19.9 & 68.3 & 0.26 & 86.2 & 9.4 & 27.9 & 0.75 & 10.3 & 11.4 & 23.9 & 0.71 & 26.0 & \\
$(10,\,0.50)$ & 26.6 & 79.9 & 0.15 & 84.3 & 12.6 & 37.0 & 0.72 & 10.3 & 12.2 & 29.9 & 0.74 & 38.4 & \\
$(15,\,0.15)$ & 19.1 & 73.2 & 0.37 & 213.5 & 11.7 & 33.8 & 0.67 & 40.9 & 13.9 & 29.9 & 0.64 & 29.1 & \\
$(15,\,0.30)$ & 37.2 & 172.9 & 0.20 & 208.3 & 22.2 & 100.7 & 0.67 & 39.9 & 21.7 & 78.0 & 0.70 & 58.3 & \\
$(20,\,0.15)$ & 32.0 & 195.2 & 0.34 & 393.7 & 22.1 & 108.0 & 0.66 & 139.8 & 19.3 & 77.4 & 0.68 & 56.1 & \\
$(20,\,0.20)$ & 42.0 & 262.6 & 0.26 & 390.9 & 25.9 & 148.1 & 0.67 & 133.4 & 24.8 & 116.9 & 0.68 & 75.4 & \\
$(25,\,0.20)$ & 66.6 & 479.0 & 0.19 & 682.7 & 38.5 & 299.6 & 0.68 & 289.7 & 41.6 & 230.3 & 0.68 & 148.7 & \\
\midrule
& \multicolumn{4}{c}{Globe} & \multicolumn{4}{c}{TOPIC} & \multicolumn{4}{c}{EqVar} & \\
\cmidrule{2-5}\cmidrule{6-9}\cmidrule{10-13}
$(p,\,\rho)$ & SHD & SID & F1 & Time & SHD & SID & F1 & Time & SHD & SID & F1 & Time & \\
\midrule
$(10,\,0.30)$ & 11.7 & 37.8 & 0.42 & 36.0 & 10.5 & 32.9 & 0.58 & 1.7 & 5.8 & 17.7 & 0.80 & 0.9 & \\
$(10,\,0.40)$ & 14.0 & 47.6 & 0.43 & 35.7 & 11.0 & 38.4 & 0.62 & 2.4 & 6.2 & 20.4 & 0.82 & 1.0 & \\
$(10,\,0.50)$ & 20.8 & 62.0 & 0.33 & 41.5 & 12.7 & 42.8 & 0.65 & 3.2 & 7.5 & 26.2 & 0.83 & 1.0 & \\
$(15,\,0.15)$ & 11.6 & 45.2 & 0.53 & 68.7 & 14.2 & 46.4 & 0.54 & 3.8 & 7.9 & 23.5 & 0.76 & 1.5 & \\
$(15,\,0.30)$ & 29.1 & 133.1 & 0.35 & 91.7 & 21.5 & 104.9 & 0.59 & 5.7 & 12.2 & 65.6 & 0.81 & 1.5 & \\
$(20,\,0.15)$ & 22.5 & 123.7 & 0.46 & 144.7 & 23.7 & 121.1 & 0.54 & 7.1 & 11.9 & 59.0 & 0.80 & 2.0 & \\
$(20,\,0.20)$ & 31.5 & 190.4 & 0.40 & 151.6 & 26.7 & 158.2 & 0.58 & 10.1 & 13.8 & 84.8 & 0.82 & 1.9 & \\
$(25,\,0.20)$ & 52.3 & 367.9 & 0.34 & 258.3 & 41.3 & 309.1 & 0.56 & 22.4 & 20.4 & 187.1 & 0.83 & 2.5 & \\
\midrule
& \multicolumn{4}{c}{NPVAR} & \multicolumn{4}{c}{CAM} & \multicolumn{4}{c}{MIP} & Bootstrap \\
\cmidrule{2-5}\cmidrule{6-9}\cmidrule{10-13}\cmidrule{14-14}
$(p,\,\rho)$ & SHD & SID & F1 & Time & SHD & SID & F1 & Time & SHD & SID & F1 & Time & Time \\
\midrule
$(10,\,0.30)$ & 6.0 & 17.5 & 0.79 & 1.3 & 6.7 & 17.4 & 0.79 & 2.8 & 3.4 & 10.5 & 0.89 & 14.6 & 3.4 \\
$(10,\,0.40)$ & 7.3 & 22.8 & 0.79 & 1.4 & 8.1 & 24.4 & 0.79 & 2.9 & 3.3 & 12.3 & 0.91 & 26.8 & 3.2 \\
$(10,\,0.50)$ & 9.5 & 30.5 & 0.79 & 1.4 & 17.1 & 45.2 & 0.63 & 3.1 & 5.9 & 21.6 & 0.86 & 7.2 & 3.5 \\
$(15,\,0.15)$ & 8.8 & 25.5 & 0.74 & 3.6 & 3.7 & 11.6 & 0.90 & 5.7 & 1.8 & 6.5 & 0.95 & 4.2 & 7.8 \\
$(15,\,0.30)$ & 14.9 & 71.7 & 0.78 & 3.6 & 14.9 & 68.0 & 0.78 & 5.1 & 6.6 & 36.6 & 0.90 & 36.1 & 8.3 \\
$(20,\,0.15)$ & 14.0 & 66.5 & 0.77 & 8.5 & 6.7 & 32.7 & 0.89 & 7.5 & 3.4 & 18.2 & 0.94 & 31.7 & 16.3 \\
$(20,\,0.20)$ & 16.0 & 97.8 & 0.79 & 8.1 & 13.0 & 74.1 & 0.84 & 7.2 & 5.5 & 38.9 & 0.93 & 58.0 & 14.7 \\
$(25,\,0.20)$ & 23.6 & 205.3 & 0.80 & 15.6 & 19.4 & 165.1 & 0.84 & 9.4 & 11.2 & 111.4 & 0.90 & 100.4 & 21.8 \\
\bottomrule
\caption{\small Full simulation results for $n=500$ and $\tau=0.2$. Each entry reports the mean over 50 independent simulation trials. Runtime is reported in seconds. For parallelizable computations in \textit{MIP}, \textit{NoTears}, \textit{GSGES} and bootstrap, runtime is reported per run.}
\label{tab:compare-c3}
\end{longtable}
\endgroup

\clearpage
\begingroup
\small
\setlength{\tabcolsep}{2.5pt}
\renewcommand{\arraystretch}{1.05}
\begin{longtable}{c@{\hspace{5pt}\vrule\hspace{5pt}}cccc@{\hspace{5pt}\vrule\hspace{5pt}}cccc@{\hspace{5pt}\vrule\hspace{5pt}}cccc@{\hspace{5pt}\vrule\hspace{5pt}}c}
\toprule
& \multicolumn{4}{c}{GSGES} & \multicolumn{4}{c}{RESIT} & \multicolumn{4}{c}{NoTears} & \\
\cmidrule{2-5}\cmidrule{6-9}\cmidrule{10-13}
$(p,\,\rho)$ & SHD & SID & F1 & Time & SHD & SID & F1 & Time & SHD & SID & F1 & Time & \\
\midrule
$(10,\,0.30)$ & 14.9 & 52.6 & 0.34 & 22.5 & 9.3 & 27.5 & 0.67 & 6.1 & 10.4 & 20.7 & 0.65 & 21.5 & \\
$(10,\,0.40)$ & 20.2 & 68.0 & 0.25 & 22.2 & 14.5 & 43.5 & 0.61 & 6.0 & 10.0 & 23.1 & 0.74 & 26.8 & \\
$(10,\,0.50)$ & 27.3 & 80.1 & 0.16 & 21.7 & 15.5 & 50.0 & 0.65 & 6.1 & 14.3 & 32.9 & 0.72 & 37.4 & \\
$(15,\,0.15)$ & 18.9 & 71.9 & 0.36 & 62.0 & 12.1 & 38.6 & 0.64 & 26.0 & 13.0 & 27.0 & 0.63 & 26.0 & \\
$(15,\,0.30)$ & 37.7 & 168.1 & 0.19 & 54.1 & 23.9 & 109.6 & 0.63 & 25.4 & 21.6 & 65.5 & 0.70 & 50.6 & \\
$(20,\,0.15)$ & 30.7 & 178.9 & 0.35 & 118.5 & 22.4 & 109.6 & 0.62 & 71.6 & 25.6 & 71.1 & 0.61 & 52.8 & \\
$(20,\,0.20)$ & 43.3 & 263.0 & 0.26 & 117.4 & 31.1 & 181.2 & 0.61 & 71.7 & 33.4 & 117.4 & 0.63 & 75.9 & \\
$(25,\,0.20)$ & 67.9 & 474.6 & 0.19 & 174.5 & 41.2 & 330.7 & 0.65 & 137.7 & 58.5 & 232.1 & 0.60 & 126.7 & \\
\midrule
& \multicolumn{4}{c}{Globe} & \multicolumn{4}{c}{TOPIC} & \multicolumn{4}{c}{EqVar} & \\
\cmidrule{2-5}\cmidrule{6-9}\cmidrule{10-13}
$(p,\,\rho)$ & SHD & SID & F1 & Time & SHD & SID & F1 & Time & SHD & SID & F1 & Time & \\
\midrule
$(10,\,0.30)$ & 11.9 & 40.9 & 0.27 & 22.0 & 12.8 & 39.9 & 0.42 & 1.3 & 5.5 & 16.6 & 0.79 & 0.9 & \\
$(10,\,0.40)$ & 15.8 & 52.7 & 0.29 & 24.7 & 16.5 & 51.9 & 0.40 & 1.8 & 6.7 & 22.7 & 0.81 & 0.9 & \\
$(10,\,0.50)$ & 22.3 & 65.9 & 0.24 & 29.1 & 19.8 & 61.2 & 0.41 & 1.7 & 8.4 & 30.5 & 0.81 & 0.9 & \\
$(15,\,0.15)$ & 13.8 & 53.6 & 0.27 & 48.1 & 19.4 & 60.8 & 0.33 & 3.1 & 7.8 & 24.7 & 0.75 & 1.3 & \\
$(15,\,0.30)$ & 29.6 & 125.4 & 0.22 & 55.5 & 31.1 & 135.0 & 0.34 & 4.7 & 10.7 & 55.6 & 0.83 & 1.3 & \\
$(20,\,0.15)$ & 24.8 & 135.9 & 0.24 & 95.2 & 32.5 & 152.6 & 0.31 & 7.0 & 12.7 & 63.5 & 0.78 & 2.1 & \\
$(20,\,0.20)$ & 35.5 & 209.6 & 0.24 & 106.1 & 40.0 & 222.4 & 0.32 & 7.5 & 15.3 & 106.5 & 0.80 & 2.2 & \\
$(25,\,0.20)$ & 55.8 & 380.8 & 0.20 & 166.4 & 62.4 & 409.3 & 0.27 & 8.2 & 21.3 & 204.2 & 0.82 & 2.8 & \\
\midrule
& \multicolumn{4}{c}{NPVAR} & \multicolumn{4}{c}{CAM} & \multicolumn{4}{c}{MIP} & Bootstrap \\
\cmidrule{2-5}\cmidrule{6-9}\cmidrule{10-13}\cmidrule{14-14}
$(p,\,\rho)$ & SHD & SID & F1 & Time & SHD & SID & F1 & Time & SHD & SID & F1 & Time & Time \\
\midrule
$(10,\,0.30)$ & 6.0 & 17.4 & 0.77 & 1.3 & 4.4 & 14.0 & 0.84 & 2.5 & 3.4 & 11.6 & 0.87 & 12.2 & -- \\
$(10,\,0.40)$ & 7.6 & 24.6 & 0.78 & 1.3 & 10.1 & 32.6 & 0.72 & 2.4 & 5.3 & 20.4 & 0.85 & 29.9 & 3.1 \\
$(10,\,0.50)$ & 9.4 & 33.0 & 0.79 & 1.3 & 16.5 & 48.6 & 0.63 & 2.6 & 7.6 & 26.3 & 0.84 & 15.1 & 3.2 \\
$(15,\,0.15)$ & 8.2 & 25.9 & 0.74 & 3.2 & 4.0 & 13.4 & 0.87 & 4.7 & 2.8 & 10.8 & 0.91 & 9.7 & 7.3 \\
$(15,\,0.30)$ & 13.0 & 63.7 & 0.79 & 3.1 & 15.1 & 70.2 & 0.76 & 4.6 & 6.6 & 38.4 & 0.89 & 54.4 & 7.1 \\
$(20,\,0.15)$ & 13.6 & 64.5 & 0.76 & 8.8 & 7.7 & 39.7 & 0.86 & 7.9 & 5.5 & 29.5 & 0.90 & 64.4 & 15.9 \\
$(20,\,0.20)$ & 17.1 & 113.7 & 0.78 & 8.8 & 17.2 & 106.2 & 0.79 & 8.0 & 8.6 & 61.3 & 0.89 & 92.7 & 15.8 \\
$(25,\,0.20)$ & 24.9 & 224.7 & 0.78 & 17.2 & 23.1 & 198.3 & 0.81 & 9.4 & 14.0 & 137.1 & 0.88 & 134.3 & 24.6 \\
\bottomrule
\caption{\small Full simulation results for $n=300$ and $\tau=0.2$. Each entry reports the mean over 50 independent simulation trials. Runtime is reported in seconds. For parallelizable computations in \textit{MIP}, \textit{NoTears}, \textit{GSGES} and bootstrap, runtime is reported per run.}
\label{tab:compare-c4}
\end{longtable}
\endgroup

\clearpage
\begingroup
\small
\setlength{\tabcolsep}{2.5pt}
\renewcommand{\arraystretch}{1.05}
\begin{longtable}{l@{\hspace{5pt}\vrule\hspace{5pt}}cccc@{\hspace{5pt}\vrule\hspace{5pt}}cccc@{\hspace{5pt}\vrule\hspace{5pt}}cccc@{\hspace{5pt}\vrule\hspace{5pt}}c}
\toprule
& \multicolumn{4}{c}{GSGES} & \multicolumn{4}{c}{RESIT} & \multicolumn{4}{c}{NoTears} & \\
\cmidrule{2-5}\cmidrule{6-9}\cmidrule{10-13}
Network.$p$.$s^\star$ & SHD & SID & F1 & Time & SHD & SID & F1 & Time & SHD & SID & F1 & Time & \\
\midrule
Dsep.6.6 & 6.8 & 17.4 & 0.39 & 25.9 & 4.6 & 9.0 & 0.65 & 2.9 & 2.6 & 5.0 & 0.79 & 13.2 & \\
Asia.8.8 & 10.0 & 33.7 & 0.33 & 52.7 & 7.3 & 17.4 & 0.60 & 5.8 & 5.0 & 9.1 & 0.72 & 16.2 & \\
Bowling.9.11 & 16.4 & 56.0 & 0.13 & 67.5 & 10.1 & 22.4 & 0.61 & 7.7 & 5.2 & 8.3 & 0.78 & 19.6 & \\
Rain.14.18 & 13.1 & 65.9 & 0.58 & 185.3 & 12.4 & 56.4 & 0.69 & 31.1 & 12.2 & 50.0 & 0.70 & 32.6 & \\
InsuranceSmall.15.25 & 26.7 & 147.9 & 0.31 & 213.3 & 15.7 & 78.1 & 0.70 & 40.8 & 16.1 & 70.7 & 0.70 & 55.0 & \\
Cloud.16.19 & 16.9 & 117.9 & 0.50 & 247.9 & 14.0 & 82.4 & 0.66 & 53.4 & 15.3 & 89.3 & 0.62 & 40.4 & \\
Funnel.18.18 & 25.4 & 124.6 & 0.27 & 318.5 & 16.3 & 42.9 & 0.61 & 87.5 & 9.7 & 21.4 & 0.75 & 35.1 & \\
Galaxy.20.22 & 21.8 & 191.0 & 0.47 & 394.6 & 18.6 & 101.7 & 0.62 & 134.3 & 13.7 & 85.3 & 0.69 & 51.5 & \\
Insurance.27.52 & 53.2 & 524.1 & 0.32 & 739.4 & 28.0 & 278.2 & 0.74 & 377.9 & 31.0 & 250.2 & 0.72 & 173.4 & \\
Factors.27.68 & 75.0 & 556.2 & 0.20 & 770.1 & 29.0 & 223.3 & 0.78 & 478.3 & 27.6 & 177.6 & 0.80 & 273.1 & \\
\midrule
& \multicolumn{4}{c}{Globe} & \multicolumn{4}{c}{TOPIC} & \multicolumn{4}{c}{EqVar} & \\
\cmidrule{2-5}\cmidrule{6-9}\cmidrule{10-13}
Network.$p$.$s^\star$ & SHD & SID & F1 & Time & SHD & SID & F1 & Time & SHD & SID & F1 & Time & \\
\midrule
Dsep.6.6 & 4.0 & 11.5 & 0.57 & 10.8 & 3.3 & 9.2 & 0.71 & 0.8 & 2.0 & 4.1 & 0.84 & 0.6 & \\
Asia.8.8 & 4.7 & 16.8 & 0.64 & 20.5 & 5.7 & 17.3 & 0.63 & 1.0 & 2.8 & 6.6 & 0.83 & 0.8 & \\
Bowling.9.11 & 8.0 & 24.5 & 0.51 & 27.8 & 8.5 & 27.3 & 0.60 & 1.7 & 2.7 & 6.5 & 0.88 & 0.8 & \\
Rain.14.18 & 13.1 & 87.9 & 0.53 & 68.1 & 10.0 & 45.5 & 0.69 & 4.7 & 5.8 & 33.1 & 0.84 & 1.4 & \\
InsuranceSmall.15.25 & 19.9 & 116.1 & 0.47 & 94.0 & 14.4 & 85.1 & 0.67 & 4.2 & 7.3 & 39.5 & 0.86 & 1.4 & \\
Cloud.16.19 & 12.7 & 116.3 & 0.59 & 93.5 & 9.5 & 62.3 & 0.74 & 6.6 & 8.3 & 59.4 & 0.78 & 1.6 & \\
Funnel.18.18 & 12.8 & 43.1 & 0.51 & 93.8 & 19.0 & 64.8 & 0.47 & 5.7 & 5.3 & 15.6 & 0.86 & 1.7 & \\
Galaxy.20.22 & 14.3 & 109.9 & 0.58 & 133.3 & 16.1 & 104.2 & 0.62 & 9.3 & 10.8 & 77.3 & 0.76 & 1.9 & \\
Insurance.27.52 & 43.4 & 468.2 & 0.42 & 305.1 & 27.4 & 306.7 & 0.69 & 23.2 & 14.1 & 158.5 & 0.87 & 2.6 & \\
Factors.27.68 & 69.5 & 601.4 & 0.23 & 328.9 & 31.0 & 256.3 & 0.72 & 20.7 & 15.7 & 179.5 & 0.88 & 3.2 & \\
\midrule
& \multicolumn{4}{c}{NPVAR} & \multicolumn{4}{c}{CAM} & \multicolumn{4}{c}{MIP} & Bootstrap \\
\cmidrule{2-5}\cmidrule{6-9}\cmidrule{10-13}\cmidrule{14-14}
Network.$p$.$s^\star$ & SHD & SID & F1 & Time & SHD & SID & F1 & Time & SHD & SID & F1 & Time & Time \\
\midrule
Dsep.6.6 & 1.9 & 3.6 & 0.85 & 0.6 & 2.7 & 5.2 & 0.79 & 1.1 & 0.1 & 0.2 & 0.99 & 0.4 & 1.1 \\
Asia.8.8 & 2.9 & 6.3 & 0.83 & 0.4 & 0.8 & 2.2 & 0.95 & 1.4 & 0.3 & 0.8 & 0.98 & 0.6 & 2.0 \\
Bowling.9.11 & 3.2 & 7.3 & 0.86 & 1.1 & 4.6 & 10.5 & 0.81 & 1.9 & 0.2 & 0.7 & 0.99 & 1.0 & 2.3 \\
Rain.14.18 & 5.9 & 28.5 & 0.84 & 3.0 & 3.5 & 15.4 & 0.91 & 5.0 & 0.9 & 5.1 & 0.98 & 10.5 & 6.7 \\
InsuranceSmall.15.25 & 6.4 & 33.7 & 0.87 & 3.5 & 11.1 & 59.2 & 0.78 & 5.0 & 4.0 & 19.2 & 0.92 & 24.4 & 8.1 \\
Cloud.16.19 & 6.9 & 46.3 & 0.82 & 4.3 & 4.2 & 26.3 & 0.89 & 6.0 & 1.1 & 6.8 & 0.97 & 18.0 & 9.9 \\
Funnel.18.18 & 6.1 & 15.3 & 0.84 & 6.0 & 1.6 & 2.8 & 0.96 & 6.3 & 0.5 & 1.6 & 0.99 & 5.2 & 10.2 \\
Galaxy.20.22 & 9.8 & 65.1 & 0.78 & 8.1 & 3.0 & 16.3 & 0.93 & 6.6 & 1.2 & 4.6 & 0.97 & 20.4 & 12.6 \\
Insurance.27.52 & 14.2 & 154.9 & 0.86 & 19.3 & 14.6 & 148.2 & 0.86 & 9.4 & 7.3 & 68.7 & 0.93 & 117.6 & 29.4 \\
Factors.27.68 & 17.5 & 186.6 & 0.87 & 24.0 & 39.0 & 265.6 & 0.70 & 13.7 & 25.2 & 174.4 & 0.82 & 172.1 & 27.9 \\
\bottomrule
\caption{\small Full simulation results on the benchmark graph structures. Each entry reports the mean over 50 independent simulation trials. Runtime is reported in seconds; for parallelizable computations in \textit{MIP}, \textit{NoTears}, \textit{GSGES}, and bootstrap, the reported runtime corresponds to a single run.}
\label{tab:compare-graph}
\end{longtable}
\endgroup

\newpage
\singlespacing
\bibliography{refs.bib}

\end{document}